\documentclass{article}
\usepackage{subfiles}
\usepackage{graphicx,times,fullpage} 
\usepackage{amsmath}  

\usepackage{amssymb} 
\usepackage{amsthm} 
\usepackage{url}
\usepackage{doi}
\usepackage{hyperref}
\usepackage{xurl} 
\usepackage{amsfonts}
\usepackage[round]{natbib}
\usepackage{booktabs}   
\usepackage{tabularx}   
\usepackage{subcaption} 
\usepackage{xcolor}
\usepackage{pgfplotstable}
\usepackage{bm}
\pgfplotsset{compat=1.18} 

\theoremstyle{plain} 

\newtheorem{lemma}{lemma}
\newtheorem{definition}{definition}

\title{Optimal Study Interior Partitioning (OSIP): An Optimization Approach to Observational Study Design under Fine Balance Constraints}
\author{Ron Adar\thanks{Faculty of Data and Decision Sciences, Technion Israel Institute of Technology, Haifa, Israel. \texttt{adar127@gmail.com}.} \and Asaf Levin\thanks{Faculty of Data and Decision Sciences, Technion Israel Institute of Technology, Haifa, Israel. \texttt{levinas@technion.ac.il}. {Partially supported by ISF - Israel Science Foundation grant number 1467/22.}}}

\begin{document}

\maketitle

\begin{abstract}
Covariate matching between treatment and control groups remains a foundational methodology in observational causal inference, designed to replicate the covariate balance of randomized controlled trials. Despite substantial methodological advancements, applied researchers continue to face an inherent trade-off between achieving rigorous balance and preserving the statistical power of the matched sample. We introduce a novel framework for a matching method named {\em Optimal Study Interior Partitioning (OSIP)}. OSIP frames observational matching as a joint optimization problem governed simultaneously by a maximum cardinality constraint on the number of strata, and a strict fine balance threshold on the matching carried out in each strata.

OSIP is based on two-stage optimization algorithm. In both stages it considers a projection of the (high-dimensional) data into one-dimensional space by evaluating an estimated propensity score for each unit. The method is based on enforcing a partition of the input to subsets of consecutive intervals of estimated propensity scores. The novelty arises because we consider this task as a global optimization problem. In the first stage, we employ an optimal partitioning algorithm with respect to an auxiliary (one-dimensional) goal function. In the second phase, we go back to the original distance function and apply heuristics approaches locally optimizing the solution from the first stage. We compare our framework with established methodologies from the literature across well-studied empirical benchmarks. We use the standard benchmark of Lindner cardiovascular dataset to demonstrate that OSIP provides a highly efficient “sweet spot” solution for the corresponding multi-criteria goal.  This conclusion is verified by considering also other benchmark datasets.
\end{abstract}

\section{Introduction}
\label{Intro}

Traditional experimental design assigns units to either a treated (also known as treatment) group, or a control group through randomization, comparing their outcomes to estimate the treatment effect. By blinding evaluators to treatment assignments, researchers can isolate the true treatment effect from potential biases. However, in many real-world scenarios, random assignment is impractical, infeasible, or unethical. Observational studies address these limitations, and designing such studies is our focus here.

\subsection{Background and Motivation}
In the context of observational studies, the objective is to infer causal relationships. Matching treated and control units is a common strategy for adjusting a vector of observed covariates $X$ in the design of observational studies. Such studies are particularly important in settings where controlled experimentation or random assignment to treatment conditions is infeasible \citep{cochran1965planning}). Matching may be performed one-to-one ($1:1$), in which each treated unit is matched to a single control unit, or in a one-to-many ($1:L$) setting, called variable-ratio matching, in which each treated unit is matched to L distinct control units. Or it can be some other variation. Let $\mathcal{T}$ and $\mathcal{C}$ denote the sets of treated and control units, respectively. The objective of the matching procedure is to construct treated-control pairs $(t,c)$, ($t \in \mathcal{T}$, $c \in \mathcal{C}$) such that the paired units are as similar as possible with respect to the observed covariates. The degree of similarity between the matched units is quantified through a covariate distance measure defined on the observed covariates.
Some matching procedures operate directly on multivariate covariate distances, while others first reduce the observed covariates to a scalar balancing score. 

\medskip

In the latter setting, the propensity score, defined as the conditional probability of treatment assignment given the observed covariates, provides a one-dimensional representation of the observed covariate information. Here, we suggest studying standard one-dimensional distance metrics. 
We propose using the $L_1$ or $L_2$ norm, to quantify the similarity between units based on these scores as an auxiliary objective.

For methods that operate on multi-dimensional covariate distance directly, we use the Mahalanobis distance metric. For most methods discussed below, we use its strict version. For OSIP, we also allow using the robust (rank-based) version. 
Matching may constitute a standalone optimization procedure, or serve as a key step within a multi-step framework. 

The second foundational technique in the design of observational studies is subclassification: the partitioning of a sample into a finite (small) set of strata based on observed covariates $\mathfrak{X}$. Although often discussed separately, these two techniques frequently overlap in implementation. For instance, \textbf{Coarsened Exact Matching} (CEM) is fundamentally a subclassification and weighting technique that does not inherently matching pairs.
Consequently, the distinction maintained in this description is intended for conceptual clarity, recognizing that the boundary between matching and subclassification is often vague in practice.

OSIP applies the two mentioned techniques. The strata for OSIP would also be called intervals or sub-intervals.

\subsection{Description of matching based methods}
Next, we consider the existing methods in the literature to which we compare ourselves.  First, the matching-based methods, and later the subclassification methods.

\indent \textbf{Optimal Matching}  is a $1:1$ matching method that constructs treated-control pairs by globally minimizing the total covariate distance across all matched pairs simultaneously, possibly subject to additional matching constraints. It is described in \cite{rosenbaum1989optimal}. Here, the landscape of computationally efficient procedures allow limited additional constraints, while allowing other (perhaps contradicting) constraints make the matching problem NP-hard, and under standard assumptions a problem that cannot be solved within a reasonable amount of time on a large instance.

\textbf{Optimal Full Matching}  is a matching algorithm, that uses all units for globally minimizing the total covariate distance (i.e., keeps all treated and control units), allowing treated units to be matched to multiple control units and vice-versa. Possibly subject to additional matching constraints. See  \citet{hansen2006optmatch} for a comprehensive discussion of this method.

\textbf{Exact matching} is another matching technique, that requires identical values of an observed covariate for matched individuals, e.g., boys matched to boys, girls matched to girls. Here, there might be additional covariates on which this identical value constraints are not enforced.

\textbf{Genetic Matching} (or \textbf{GenMatch}). \citet{diamond2013genetic}  suggested to utilize an evolutionary search algorithm to check and improve the covariate balance. This is achieved by optimizing a weight matrix, $W$, within a Generalized Mahalanobis Distance (GMD) framework.  Evolutanary search algorithms are used as meta-heuristics for optimizing functions subject to constraints in many research fields.  In every step it keeps a population of solutions, and allow two possible steps on a population.  The first such operation is crossover, while the second is mutation.  These are randomized operations and by keeping only a subset of the resulting solutions after each step, the algorithm can direct the population to include better solutions as the current-best solution. In GenMatch, the initialization is a heuristic starting population, the algorithm iteratively explores the covariate space $\mathfrak{X}$ to minimize a loss function, until some stopping criterion on the $p$-value (or the minimum loss) is met. In its default configuration, GenMatch seeks to optimize the minimum $p$-value observed across a variety of balance tests (such as $t$-tests and Kolmogorov-Smirnov tests). By focusing on the smallest $p$-value, the algorithm effectively minimizes the maximum imbalance across all covariates, though the specific loss function can be modified by the user. 

\textbf{Cardinality Matching}  is a two-step optimization matching method. At the balancing step, it identifies the largest matched sample that satisfies a set of pre-specified balance constraints. This is achieved through a set of linear inequality constraints. The result of this step is two balanced sets where the first consists of a subset of treatment units and the second subset consists of controls units (where the balance is achieved with respect to the pre-defined condition stated by the balance constraints).  This balancing step is followed by a pairing step with respect to a parameter named matching ratio. 
The matching ratio must be specified prior to the balancing step. The procedure may incorporate all treated units by assigning each to $L$ distinct controls (variable-ratio matching), or alternatively, it can enforce a $1:1$ ratio to facilitate exact pair-matching in the subsequent stage.
The pairing step is the actual matcher. Its goal is to minimize the total sum of distances between units in a pair (thereby minimizing the heterogeneity of the treated minus control pair differences) with respect to some distance metric, like the Mahalanobis distance. 
This method is introduced and analyzed in \citet{zubizarreta2014cardinality}. 

\subsection{Methods for Covariate Adjustment via Subclassification}

\indent \textbf{Subclassification} (or \textbf{stratification}) aims to reduce selection bias, by partitioning the sample into a finite set of strata based on observed covariates $X$.  It was introduced by  \citet{cochran1965planning} and \citet{cochran1968subclassification}. The goal of this approach is to achieve "local balance" by comparing treatment and control units within relatively homogeneous subsets rather than a global comparison. 
Combining subclassification with the propensity score is also common (see, e.g., \citealp{imbens2015causal}).
\textbf{Coarsened Exact Matching} (or \textbf{CEM}). \citep{cem2012}  developed a sub-method of the \textbf{Monotonic Imbalance Bounding}, see \cite{iacus2011mib}. This method works in 3 steps: 
1) \textbf{Coarsening}: Split (or coarsen) any of the original covariates into sub-categories, hence make a numeric value into categorical (i.e., years of education which is numeric transposed into categories like “elementary school”, “high school” “college” etc.). 
2) \textbf{Stratum Creation}: After this pre-processing step, any such category is used as a “key” to a separate stratum. Since in many scenarios, there is usually more than a single observed covariate, a stratum is defined by the Cartesian product of all sub-categories, across all the newly created categorical covariates. 
3)	\textbf{Estimation and Pruning}: In the final step, strata containing exclusively control units are discarded. Among strata that contain both treatment and control units, the sample average treatment effect on the treated units (SATT) is estimated using a weighted difference-in-means, where the control weights are scaled both globally, and per stratum. Regarding strata containing only treatment units (if such exist) they could either be ignored (as most implementations of the \textbf{CEM} package do) or used with extrapolated values of the control units. Here, when we evaluate CEM such strata containing only treated units are discarded.

\textbf{Quintile Subclassification}  is one of the first subclassification methods in the observational study’s literature (see \citep{cochran1968subclassification, rosenbaum1984quintile}).   \cite{cochran1968subclassification} showed that subclassification can be highly effective for bias reduction. Specifically, under a model assuming a normal distribution on the single continuous covariate and a linear relationship between the covariate and the outcome, partitioning the sample into five equal-sized subclasses (quintiles) removes approximately $90\%$ of the initial bias. While this result holds primarily for monotonic relationships and near-normal distributions, it established the 'rule of five' that remains a benchmark in observational study design since then. 

\textbf{Refined Quintile Subclassification} are discussed in \cite {pimentel2015large, brumberg2024optimal}. This method represents a modern adaption of the classical quintile subclassification approach. \cite{pimentel2015large} introduced a framework for optimal matching with \textbf{refined covariate balance} structured as a hierarchical system. In this approach, there’s a pre-defined lexicographic priority over a nested sequence of nominal covariates, $v_1, \dots, v_m$. The algorithm treats the first covariate in the sequence as the most critical to balance; each subsequent covariate acts as a "refinement" (a more fine partition) of the previous ones. By optimizing balance in this prioritized order, the method ensures that the most essential confounders are strictly equated, while constraints on finer, secondary variables are systematically relaxed only when perfect balance is infeasible. \citet{brumberg2024optimal} utilize propensity score quintile subclassification, 
the modern extension of Cochran's 'rule of five' as a starting point, motivated by the goal of pushing bias reduction from $90\%$ toward $100\%$. Their new approach, known as optimal refinement of strata, aims to reduce the residual bias that traditional quintile subclassification leaves behind. This is accomplished via two operations. First, formulating the split of a stratum into two strata as an integer program, that minimizes within-stratum covariate imbalance. 
Second, approximating the optimal integer solution through a randomized rounding procedure applied on a linear programming relaxation of the integer program mentioned above. Through these steps, the method identifies near-optimal boundaries that subdivide (or refine) the original strata, simultaneously improving balance across many covariates.

\subsection{Paper outlines}
In this paper, we introduce the OSIP framework in Section~\ref{OSIP} and evaluate its performance against established benchmark methods in the literature using standard statistical criteria. We begin by describing the Lindner dataset used to motivate and demonstrate our approach alongside the relevant methodological background. Afterwards, we conclude with a summary of our results on additional datasets.

\section{The Lindner dataset}
\label{LINDNER}

The Lindner dataset is considered e.g., in \cite{helmreich2009psagraphics, li2020doublyrobust, obenchain2020localcontrol}.  It contains observational medical data collected in $1997$ from $996$ patients undergoing an initial Percutaneous Coronary Intervention (PCI), commonly known as coronary angioplasty, at Ohio Heart Health, Christ Hospital in Cincinnati. The patients were followed for at least six months by the medical staff at the Lindner Center.  This dataset served as a standard benchmark in observational study design.  We use it as our main example of our approach to demonstrate the details of our method and its potential advantages.  In the appendix we provide additional support of our claimed advantages based on other well-established benchmark datasets. 

{\bf The causal inference challenge (selection bias).}
The dataset is known for showcasing severe treatment selection bias.
    Here, the treatment is the use of the drug {\em abciximab}.  This drug is a high-molecular-weight, expensive IIb/IIIa cascade blocker designed to prevent blood clots. The standard medical protocols causes the following well known bias.   
 Physicians deliberately assigned abciximab to patients who were significantly more severely diseased (e.g., those experiencing acute myocardial infarctions or displaying lower ejection fractions). Because treatment was not randomized, a naive comparison of treated and control groups would heavily biased the true treatment effect.

{\bf Dataset structure and key covariates.}  Next we list the covariates associated with each patient in this study.  Some approaches use the distinction between a continuous covariate and a binary covariate or a nominal covariate or an ordinal covariate, etc.. Thus, we present the covariates used in the Lindner dataset with their types. 

\begin{table}[htbp]
\centering
\caption{Description of Variables in the Lindner Dataset ($N = 996$)}
\label{tab:lindner_variables}
\begin{tabularx}{\textwidth}{l l X}
\toprule
\textbf{Variable Name} & \textbf{Data Type} & \textbf{Description} \\
\midrule
\texttt{abcix} & Binary (0/1) & Treatment Indicator: 1 if the patient received abciximab; 0 if they received usual care alone (298 patients in control, 698 patients treated). \\
\addlinespace
\texttt{lifepres} / \texttt{sixMonthSurvive} & Numeric value (0 or 11.4) & Primary Outcome: Originally recorded as years of life preserved; often recoded as 6-month survival rate. \\
\addlinespace
\texttt{cardbill} & Numeric value & Secondary Outcome: Total cardiac hospital billing costs in 1998 US dollars. \\
\addlinespace
\texttt{stent} & Binary (0/1) & Covariate:  Coronary stent deployment (1 = Yes, 0 = No). \\
\addlinespace
\texttt{height} & Numeric integer (108-196) &  Covariate: Patient height in centimeters. \\
\addlinespace
\texttt{female} & Binary (0/1) & Covariate: Biological sex (1 = Female, 0 = Male). \\
\addlinespace
\texttt{diabetic} & Binary (0/1) & Covariate: Diabetes diagnosis (1 = Yes, 0 = No). \\
\addlinespace
\texttt{acutemi} & Binary (0/1) & Covariate: Acute myocardial infarction within 24 hours (1 = Yes, 0 = No). \\
\addlinespace
\texttt{ejecfrac} & Numeric value (0-90) & Covariate: Left ventricular ejection fraction (percentage measure of heart pump function). \\
\addlinespace
\texttt{ves1proc} & Numeric integer (0-5) & Covariate: Number of coronary vessels involved in the initial PCI procedure (integers ranges from 0 to 5). \\
\bottomrule
\end{tabularx}
\end{table}

Following standard methodological practice established in the propensity score literature, we define the secondary outcome variable as the logarithm of the cardiac billing costs (\texttt{cardbill}) rather than using its raw monetary value \citep{helmreich2009psagraphics}. Consequently, this transformed outcome variable is evaluated across our matched strata as \texttt{log\_cardbill}.  Similarly, we use this adapted dataset across all methods used to evaluate this Lindner dataset.

{\bf The \texttt{ejecfrac} distribution.}
Among the continuous covariates evaluated, the left ventricular ejection fraction (\texttt{ejecfrac}) is widely recognized as a particularly challenging variable to balance \citep{kereiakes2000abciximab, helmreich2009psagraphics}. A statistical summary of the aggregate baseline data yields a mean of $50.97\%$ and a median of $55.00\%$, spanning a wide operational range from a minimum of $0.00\%$ to a maximum of $90.00\%$. This pronounced gap between the mean and median mathematically confirms that the variable follows a distinct left-skewed distribution, with the bulk of the population concentrated tightly within the upper-middle quartile bounds ($Q_1 = 45.00\%$, $Q_3 = 56.00\%$). To visually capture the underlying selection bias and structural differences between the cohorts before matching, Figure~\ref{fig:ejecfrac_baseline} illustrates the unadjusted probability density distributions for both the treated and control groups. It  demonstrates the left-skewed nature of this clinical covariate and the substantial initial layout imbalance between the groups

\begin{figure}[tbph]
    \centering
    \includegraphics[width=0.7\textwidth]{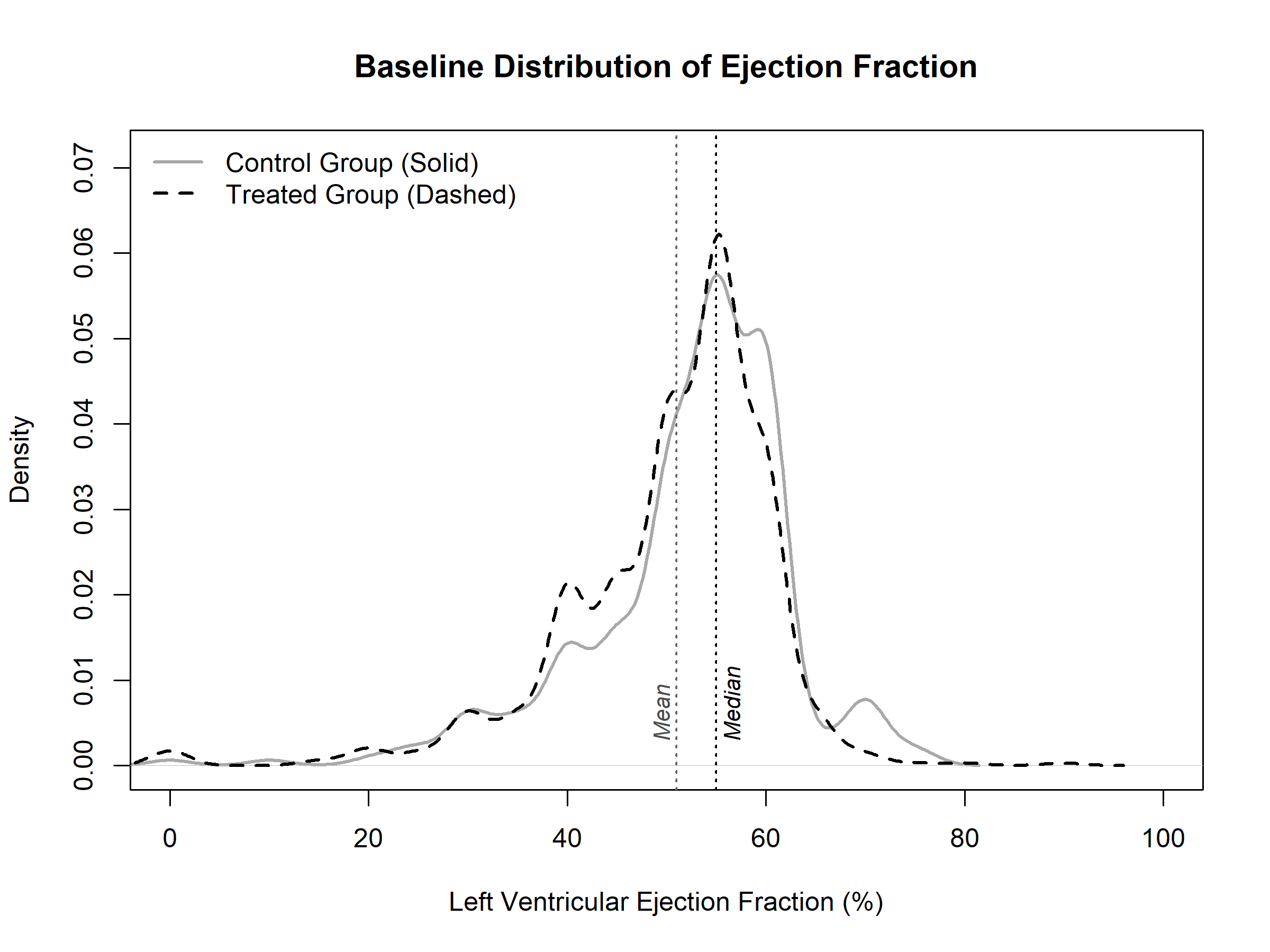}
    \caption{Baseline probability density distribution of left ventricular ejection fraction (\texttt{ejecfrac}) across the unadjusted Lindner dataset (we used R's density function to create the plot). The control cohort ($N_c = 298$) is represented by the solid gray line, and the treated cohort ($N_t = 698$) is denoted by the dashed black line. Vertical dotted reference markers identify the overall population mean ($50.97\%$) and median ($55.00\%$).}
    \label{fig:ejecfrac_baseline}
\end{figure}

\section{Preliminary Data Processing}
Our next goal is to present standard preprocessing steps applied by all methods.  We will demonstrate these steps using the Lindner dataset.
The primary objective of any matching or subclassification framework is to construct a balanced observational sample that closely mimics the properties of a randomized controlled trial \citep{rubin2007design}. Crucially, following the foundational tenets of observational study design, the outcome or response covariate is strictly blinded during this design phase and is only introduced after the balancing mechanism is finalized. Depending on the specific methodology used, this conditioning process yields either a set of discrete strata, a collection of optimized matched pairs, or a weighted sample where units are assigned strategic weights to ensure group comparability.

To prepare the empirical Lindner dataset—or any subsequent experimental data matrix—for optimization, we implement a sequence of four standardized preprocessing procedures. To ensure a completely consistent baseline for algorithmic comparison, these data-conditioning steps are executed uniformly across all datasets and across all design methods prior to initiating the design methods.

{\bf Data Harmonization and Pre-processing}
We implement a uniform pre-processing pipeline consisting of standard and well-established procedures. While the propensity score calculation (Step 1 of the pre-processing) is a prerequisite specifically for propensity-based methods and irrelevant for other methods, the subsequent trimming and standardization (Steps 2, 3, and 4 of the pre-processing) are applied to the global dataset. Consequently, every algorithm evaluated in this study operates on the same standardized and common-support-filtered data.

\begin{enumerate}
    
    \item \textbf{Propensity Score Projection.}     \label{ps}
    The propensity score is a well-established criterion in the matching literature \citep{rosenbaum1983central, rosenbaum2020modern, helmreich2009psagraphics, austin2011tutorial, imbens2015causal}. By projecting the multidimensional baseline covariate space onto a one-dimensional scalar probability, defined as $e(X) = P(T = 1 \mid X)$, it serves as a fundamental balancing score. Conditioning on this bounded metric ($[0,1]$) effectively aligns the multivariate distributions of the treated and control groups. Thereby mitigating selection bias, mimicking the properties of a randomized experiment.  Furthermore, it establishes the partition boundaries for the OSIP method, presented at Section \ref{OSIP}.
    However, regardless of the method to be used, it also serves as a preliminary step, needed for the next trimming step. 

    \medskip
    
    Note that the propensity score of a unit is unobserved, it cannot be evaluated exactly. Thus, we use the standard approach, and instead of using the propensity score we use an estimated value of it.
    In our implementation, we load the Lindner dataset using the R (CRAN) package \textbf{PSAgraphics} \citep{helmreich2009psagraphics}. 
    We use the widely used logistic regression model for estimating the propensity score \citep{rosenbaum1983central, helmreich2009psagraphics}. Specifically, the log-odds of receiving the treatment ($T_i = 1$) for an individual $i$ is modeled as a linear function of their baseline clinical characteristics:
    \begin{equation*}
    \ln\left( \frac{e(X_i)}{1 - e(X_i)} \right) = \beta_0 + \beta_1 \text{stent}_i + \beta_2 \text{height}_i + \beta_3 \text{female}_i + \beta_4 \text{diabetic}_i + \beta_5 \text{acutemi}_i + \beta_6 \text{ejecfrac}_i + \beta_7 \text{ves1proc}_i
    \end{equation*}
      where $e(X_i) = P(T_i = 1 \mid X_i)$ denotes the (estimated) propensity score, $\beta_0$ is the intercept, and $\beta_1$ through $\beta_7$ represent the regression coefficients estimated via maximum likelihood.
    
    \item \textbf{Common Support Trimming (Universal).}  \label{commonsupportstep}
    Common support trimming (also known as enforcing strict overlap) is a preprocessing procedure where units exhibiting extreme propensity score values are removed from the dataset. By eliminating these non-overlapping boundary cases, the analysis focuses exclusively on the region of high covariate overlap between the treated and control cohorts. This is a critical practice for ensuring internal validity and reducing model dependence \citep{crump2009dealing, rosenbaum2020modern, traskin2011defining, fogarty2016discrete, li2019propensity, stuart2010matching}.    
    To ensure that the comparison remains valid across all methods, we enforce a strict common support requirement. Let $\mathcal{P}_T$ and $\mathcal{P}_C$ represent the empirical ranges of the estimated propensity scores for the treated and control groups, respectively:
    \begin{align*}
        \mathcal{P}_T = [\min(e(X) \mid T=1), \max(e(X) \mid T=1)] \mbox{        and       }
        \mathcal{P}_C = [\min(e(X) \mid T=0), \max(e(X) \mid T=0)] \ .
    \end{align*}
    The common support region $\mathcal{S}$ is defined as the interval $[L_{\mathcal{S}}, R_{\mathcal{S}}]$, where the boundaries are given by
    $L_{\mathcal{S}} = \max(\min \mathcal{P}_T, \min \mathcal{P}_C)$ and $R_{\mathcal{S}} = \min(\max \mathcal{P}_T, \max \mathcal{P}_C)$.
    Any unit whose propensity score falls outside $I = [L_{\mathcal{S}}, R_{\mathcal{S}}]$ is removed. This ensures that all methods are evaluated only on the subset of the population where treatment and control groups overlap. 
    
    We remark that this step deviates from the approach in \cite{levin2024}, which considers the entire theoretical range of $[0,1]$.  This trimming step allows focusing the computational effort of the dynamic programming algorithm of OSIP strictly on the region of empirical overlap. 
     
    \item \textbf{Pooled Variance Standardization (Universal).}     \label{pooledvariancestep}
    To prevent covariates with naturally large scales from disproportionately dominating the distance metrics across the evaluated frameworks, all continuous numeric covariates are standardized. Specifically, we scale each covariate $X_m$ by its \textbf{pooled within-group standard deviation}, following the structural variance-weighting conventions \citep{rubin1980bias, rosenbaum1985constructing, stuart2010matching, austin2014moving}:
    \[ \sigma_{m,\text{pooled}} = \sqrt{\frac{S^2_{m,T} + S^2_{m,C}}{2}} \]
    where $S^2_{m,T}$ and $S^2_{m,C}$ represent the unbiased sample variances within the treated ($T$) and control ($C$) groups, respectively, computed as:
    \[ S^2_{m,T} = \frac{1}{N_T - 1} \sum_{i: T_i=1} (X_{im} - \bar{X}_{m,T})^2 \quad \text{and} \quad S^2_{m,C} = \frac{1}{N_C - 1} \sum_{j: T_j=0} (X_{jm} - \bar{X}_{m,C})^2 \ . \]
    Here, $N_T$ and $N_C$ denote the respective group sample sizes, while $\bar{X}_{m,T}$ and $\bar{X}_{m,C}$ represent the group-specific sample means. 
    
    To preserve maximal numeric fidelity for the optimization stages, the square root calculation is maintained at full double-precision floating-point representation without arbitrary rounding. However, to ensure numerical stability against invariant, near-constant, or missing-value covariates, a regularizing safety threshold is implemented: if $\sigma_{m,\text{pooled}}$ is mathematically undefined or falls below $10^{-10}$, the scaling factor defaults strictly to 1 to prevent division-by-zero errors. Each covariate value is subsequently transformed via $X'_{im} = \frac{X_{im}}{\sigma_{m,\text{pooled}}}$.

     \item \textbf{Assuming More Control Than Treated (Universal).} 
     \label{fliprolesstep}
     To simplify the presentation of our method, we would like to assume that the number of treated units is at most the number of control units. This idea is common practice  $1:1$ matching \citep{fredrickson2020comment}. 
     If this assumption does not hold, (with respect to the subset of units left after the common support trimming), we change the definition of treatment and control by switching their roles.  After the output designs are constructed, we can switch the treated and control roles back to their original meaning. This step is carried out for all methods that we consider. For many of those it has absolutely no impact, and for our method it means that our formulas presented below do not have different cases based on the condition (more control units than treated units) being satisfied or not. 
\end{enumerate}

\begin{figure}[htbp]
    \centering
    \small
    \begin{subfigure}{0.48\textwidth}
        \centering
        \includegraphics[width=\textwidth]{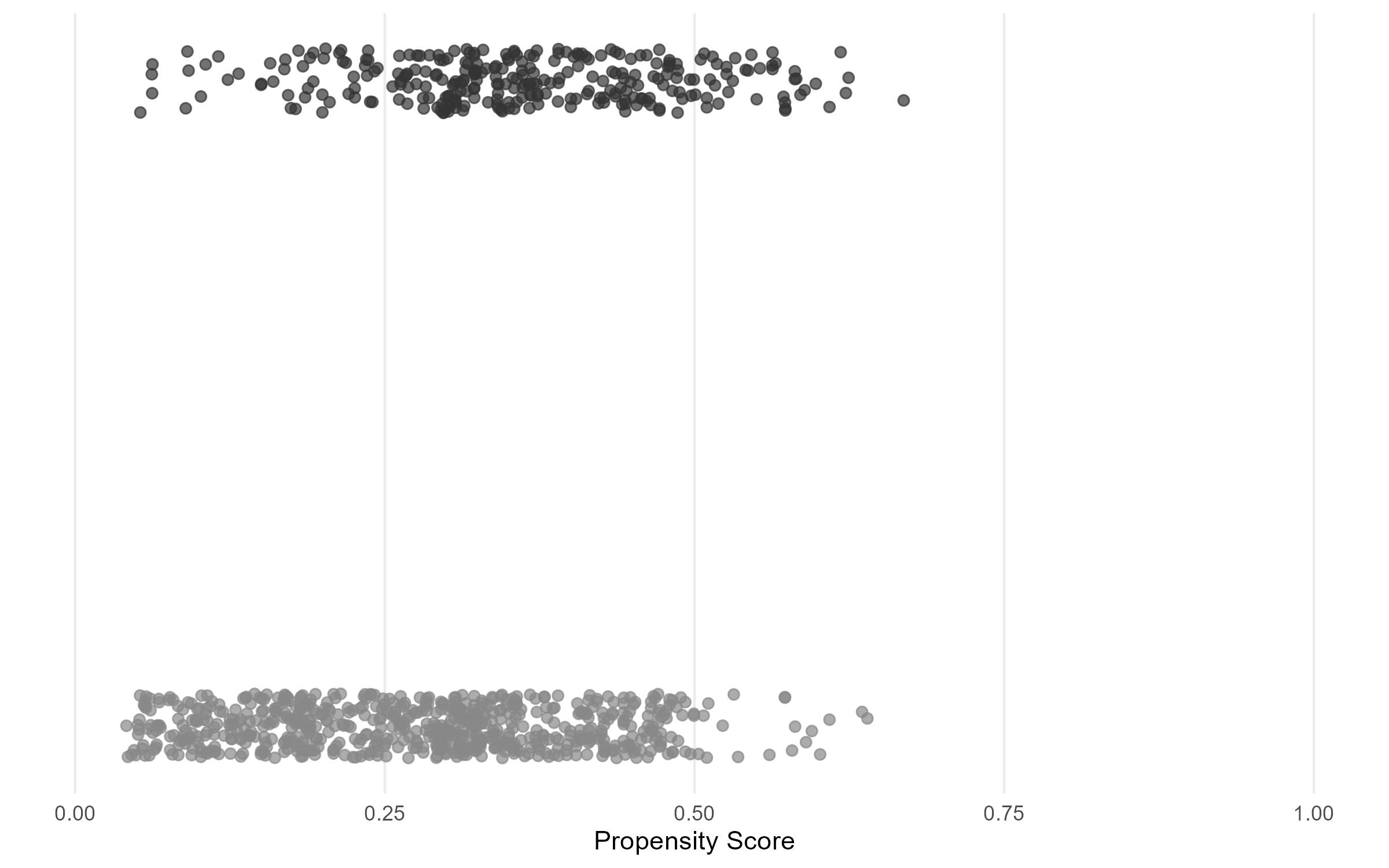}
        \caption{Full Sample (Before Trimming)}
        \label{fig:ps_untrimmed}
    \end{subfigure}
    \hfill 
    \begin{subfigure}{0.48\textwidth}
        \centering
        \includegraphics[width=\textwidth]{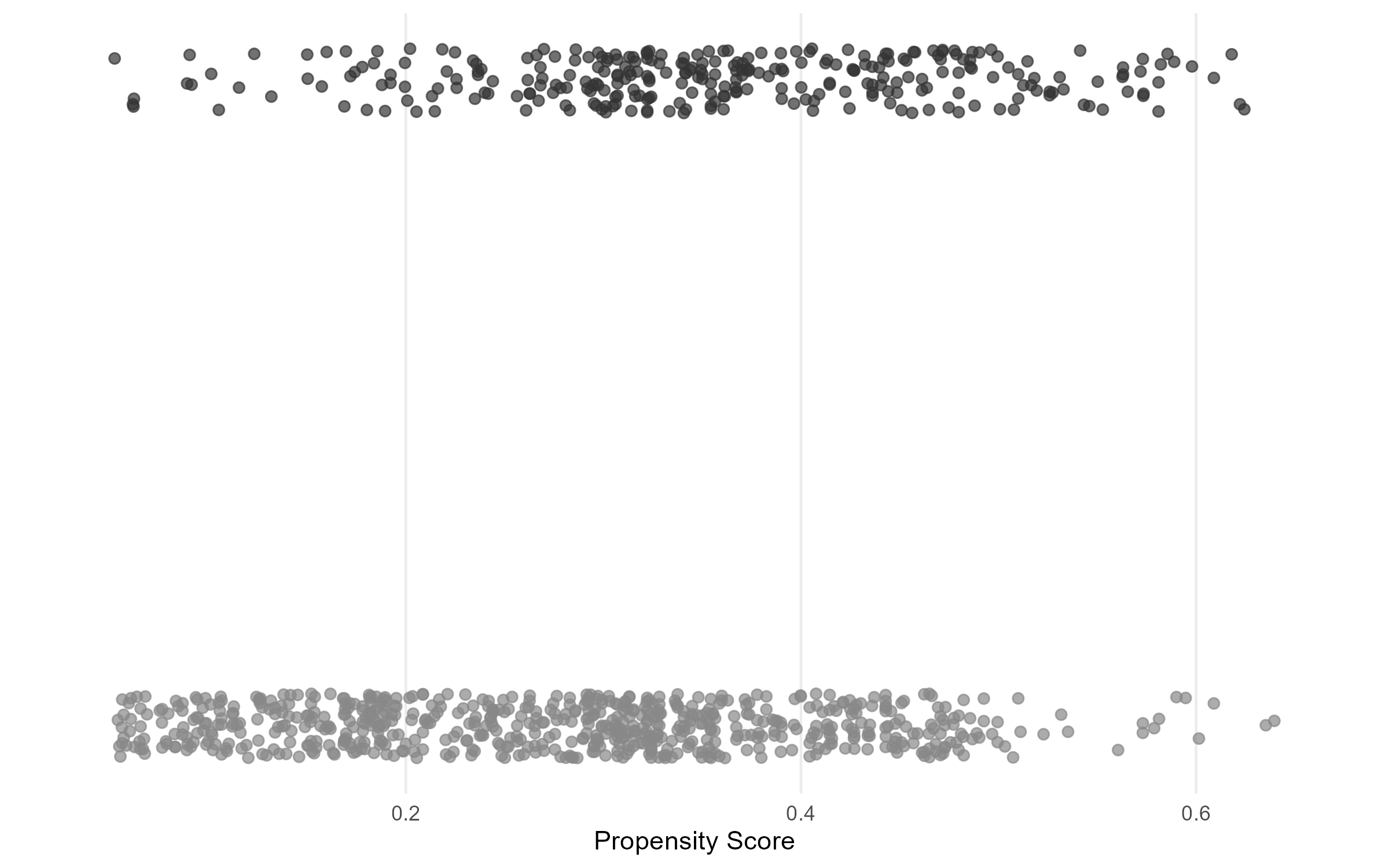}
        \caption{Common Support Sample (After Trimming)}
        \label{fig:ps_trimmed}
    \end{subfigure}

    \caption{Propensity score distributions for the Lindner dataset before and after enforcing common support constraints. Dark circles (top) denote treated units; light gray circles (bottom) denote control units. Trimming removes non-overlapping units at the support boundaries.}
    \label{fig:lindner_ps_comparison}
\end{figure}

For the Lindner dataset, applying step~\ref{commonsupportstep} shifts the interval boundaries from $[0,1]$ to $[0.0417, 0.6878]$. Figure~\ref{fig:lindner_ps_comparison} illustrates the propensity score distributions before and after trimming. Step~\ref{fliprolesstep} is used because the dataset originally contains $298$ control units and $698$ treated units. After restricting to the common support and reversing the group roles, the sample consists of $297$ treated units and $684$ control units.

\section{Optimal Study Interior Partitioning (OSIP) method}
\label{OSIP}
Optimal Study Interior Partitioning (OSIP) is a two-stage optimization algorithm designed to construct matched treated and control groups. The framework is motivated by the work of \cite{levin2024}. 
Although this work utilizes the general structure of Levin’s work, it serves as an independent study focused on application using a different algorithm. Here, we provide a stand-alone comparison of OSIP against the established methods mentioned in Section~\ref{Intro}, using various benchmark datasets. Because this work involves both the proposal and the evaluation of OSIP, we first define it. 

\subsection{OSIP overview}
As mentioned earlier, following the pre-processing steps mentioned above, OSIP applies a two-stage approach. The first stage is using dynamic-programming optimization to find a candidate feasible solution. The second step is using a heuristic search to try optimizing the solution achieved from the first stage. We now elaborate on the specific mechanics, objectives, and constraints of each stage. From a high-level perspective, OSIP identifies a near optimal $1:1$ matching between treated and control units by structurally restricting the search space. To establish a rigorous foundation for both stages, we first formalize the strict limitations governing the feasibility of any candidate solution. These limitations were suggested by \cite{levin2024}.

The feasibility is defined based on two parameters denoted as $K,\delta$.
We define a mathematically feasible solution to be a partition of the trimmed support interval $[L_{\mathcal{S}}, R_{\mathcal{S}}]$ into a sequence of $k$ contiguous, non-overlapping active sub-intervals (strata in our approach):
$I_1, I_2, \dots, I_k $
defined by a sequence of boundaries $L_{\mathcal{S}} = p_0 < p_1 < \dots < p_k = R_{\mathcal{S}}$ (where $k\leq K$). To ensure that every unit within the support domain is uniquely assigned to exactly one stratum, the nonempty sub-intervals are conventionally defined as left-open and right-closed, taking the form $I_m = (p_{m-1}, p_m]$ for $1 < m \le k$. To accommodate boundary constraints, the initial leftmost sub-interval is explicitly closed on both ends, denoted as $I_1 = [p_0, p_1]$. Beyond these standard partitions, alternative interval topologies may theoretically arise. Specifically, the framework admits the existence of empty strata, containing zero treated and zero control units, which manifest as open intervals of the form $(p_{m-1}, p_m)$. Furthermore, an empty interval may terminate at an arbitrary value $b$ within the support range, resulting in an open partition $(p_{m-1}, b)$, which subsequently requires the succeeding nonempty interval to be closed on both endpoints. While these empty configurations represent valid theoretical edge cases, we subsequently demonstrate that an appropriate choice of the parameter $\delta$ analytically eliminates the necessity of partitions of the form $(p_{m-1}, b)$. Moreover, empirical evidence from the comprehensive datasets utilized in this study indicates that empty intervals of the form $(p_{m-1}, p_m)$ are practically redundant when dealing with real-world applications.

As suggested in \cite{levin2024}, the structure of the partition space is bounded by two global structural parameters, $\delta$ and $K$ using the following three families of constraints.
\begin{itemize}
    \item \textbf{Length Constraint ($\delta$):} No sub-interval may exceed a maximum width of $\delta$ in the propensity score space, that is, we require $|I_m| = (p_m - p_{m-1}) \le \delta$ for all $m \in \{1, \dots, k\}$. This ensures local homogeneity by preventing treated units from being matched to controls with significantly different propensity scores.  That is, this constraint heuristically tries to limit ourselves to balanced matchings.
    
    \item \textbf{Partition Size Limit ($K$):} The parameter $K$ represents the maximum allowable granularity of the stratification, acting as a strict upper bound on the number of active sub-intervals such that $k \le K$. 
Similarly to the analysis in \cite{levin2024}, we can always assume that the pair ($\delta$, $K$) satisfies the structural baseline relationship:
\begin{equation}
\label{eq:delta_k_bound}
\frac{K}{2} \le \left\lceil \frac{p_k-p_0}{\delta} \right\rceil \le K .
\end{equation}
\item \textbf{Fine-balance constraint.} For the resulting partition to be considered feasible, every individual active sub-interval $I_m$ must satisfy the fine-balance constraint with respect to a new artificial covariate being the index of the sub-interval containing (the propensity score of) the unit. This condition ensures that within each stratum, a sufficient density of control units exists to locally satisfy a $1:1$ treatment-to-control ratio.  Formally, for each interval $I_m$, let $n_{m,T}$ be the number of treated units in $I_m$, and let  $n_{m,C}$ be a count of control units in $I_m$.
The fine balance constraint holds if, for all $m$, we have
\begin{equation}
\label{eq:fine_balance_capacity}
n_{m,C} \ge n_{m,T} \ . 
\end{equation}
\end{itemize}
\medskip

\subsection{Pairing and distance metrics} 

To formalize the optimization framework, we define the structural unit of our post-stratification matching space. Let $\mathcal{T}$ be the set of treated units, and $\mathcal{C}$ the set of control units, with respect to a given dataset. Here, we assume that a feasible solution is known based on context.

\begin{definition}[Matched Pair]
A matched pair $M_{ij} = (t_i, c_j)$ is a tuple consisting of exactly one treated unit $t_i \in \mathcal{T}$ and one control unit $c_j \in \mathcal{C}$ that have been assigned to the same stratum (interval) $I_m$.
\end{definition}

The following definition simulates the core idea of the OSIP method.  

\begin{definition}[Pairwise Distance]
Let $x_i$ and $x_j$ be the covariate vectors (or scalar propensity scores) for units $t_i$ and $c_j$, respectively. The distance between the matched pair $(t_i,c_j)$ is given by a generalized metric $d(x_i, x_j) \ge 0$, which quantifies the local covariate discrepancy.
\end{definition}

\subsubsection{\texorpdfstring{Global Match Discrepancy and $L_p$ Aggregation}{Global Match Discrepancy and Lp Aggregation}}
To minimize global imbalance across the entire matched sample, the objective function aggregates individual pairwise distances across all selected matches. For a complete matching configuration $\mathcal{M} = \{(i, j)\}$, the global match discrepancy is evaluated by applying a generalized $L_p$ aggregation norm over the individual pair distances:

\begin{equation}
\mathcal{D}_p(\mathcal{M}) = \left( \sum_{(i, j) \in \mathcal{M}} \left[ d(\mathbf{x}_i, \mathbf{x}_j) \right]^p \right)^{\frac{1}{p}}
\end{equation}
\noindent where $p \ge 1$, and $d(\mathbf{x}_i, \mathbf{x}_j)$ represents the chosen distance metric between unit $i$ and unit $j$. 

This formulation provides flexibility at both the unit level and the sample level.
 The function $d(\mathbf{x}_i, \mathbf{x}_j)$ can accommodate multivariate covariate vectors via metrics listed below, or it can be restricted to the one-dimensional case where $\mathbf{x}_i$ and $\mathbf{x}_j$ represent scalar propensity scores and $d(\mathbf{x}_i, \mathbf{x}_j)= |\mathbf{x}_i- \mathbf{x}_j|$.
 Setting $p=1$ minimizes the total absolute discrepancy across the sample. Conversely, setting $p=2$ represents a root-mean-square optimization, which more heavily penalizes large individual pair mismatches to protect against localized balance failures.

\subsubsection{Multi-Dimensional Mahalanobis Distance Metric} \label{multidimmet}
While the scalar propensity score defines the structural partition boundaries of the strata, matching purely on the score can leave residual imbalances within the intervals if covariates are highly non-linear or high-dimensional. To guard against this, the underlying distance metric $d(x_i, x_j)$ used for our cost function can be defined using the Mahalanobis distance.  We use it in the second stage of our optimization method.

The \textbf{Mahalanobis distance} \citep{rubin1979, rubin1980bias, rosenbaum1985constructing, stuart2010matching} between a treated unit vector $x_i$ and a control unit vector $x_j$ scale-adjusts the discrepancies by factoring in the sample covariance matrix:
\begin{equation}
d_M(x_i, x_j) = \sqrt{(x_i - x_j)^T \Sigma^{-1} (x_i - x_j)}
\end{equation}
where $\Sigma$ represents the pooled, multi-dimensional covariance matrix of the covariates calculated across the baseline sample:
$\Sigma = \frac{\Sigma_T + \Sigma_C}{2}$.

Unlike the standard Euclidean distance—which implicitly assumes all covariates are orthogonal and measured on identical scales—the Mahalanobis metric offers two critical structural advantages for a matching method.  These are described as follows.
\begin{enumerate}
    \item \textbf{Automatic Scale Normalization:} It multiplies the vector differences by $\Sigma^{-1}$, effectively dividing each covariate by its variance. This prevents variables with naturally large scales  from mathematically overwhelming variables with small scales. Since we apply the pre-processing steps, the first advantage occurs mainly in covariates which were not standardized.   
    \item \textbf{Correlation Correction:} By utilizing the off-diagonal entries of the covariance matrix, it accounts for the linear dependencies between covariates. If two covariates are highly collinear, the Mahalanobis metric avoids ``double-counting'' that specific dimension of imbalance during pair construction.
\end{enumerate}

The standard version of the Mahalanobis distance presented above is highly sensitive to the presence of outliers and the scale distortions introduced by rare binary covariates \citep{rosenbaum2020design}. To address these limitations, we employ the \textbf{rank-based Mahalanobis distance} (also referred to as the \textbf{robust Mahalanobis distance}) \citep{rosenbaum2020design}, which modifies the traditional calculation through three main steps:

\begin{enumerate}
    \item \textbf{Rank Transformation:} Each covariate is replaced, one at a time, by its corresponding ranks, utilizing average ranks to resolve ties.
    \item \textbf{Covariance Adjustment:} The covariance matrix of these ranks is pre- and post-multiplied by a diagonal matrix. The diagonal elements of this matrix are defined as the ratios of the standard deviation of untied ranks (the sequence $1 \dots n$, where $n$ is the total number of units) to the standard deviations of the actual tied ranks of the covariates.
    \item \textbf{Distance Calculation:} The final robust Mahalanobis distance is computed using the transformed ranks and this adjusted covariance matrix.
\end{enumerate}

\subsection{Scores Calculation}
Consider a solution, and let $I_m$ as before be the interval 
$I_m = (p_{m-1}, p_m]$. Assume that we have some pairwise distance metric $d$ defined, we can define the score (or cost) of $I_m$. 

\begin{definition}[True Interval Matching Cost]
Let $\mathcal{T}_m$ and $\mathcal{C}_m$ denote the subsets of treated and control units falling within the interval boundaries $I_m = (p_{m-1}, p_m]$. The (true) matching cost $\text{MatchCost}(I_m)$ with respect to the general distance metric $d(\cdot, \cdot)$ is defined as the minimum total discrepancy of a valid injection:
\begin{equation}
\text{MatchCost}(I_m) = \min_{\pi} \sum_{i \in \mathcal{T}_m} d(x_i, x_{\pi(i)})
\end{equation}
\noindent subject to the constraint that $\pi: \mathcal{T}_m \to \mathcal{C}_m$ is an injective mapping (1:1 matching without replacement). If no feasible matching configuration can be constructed due to sample starvation ($|\mathcal{C}_m| < |\mathcal{T}_m|$), the cost is set to infinity, $\text{MatchCost}(I_m) = \infty$.
\end{definition}

Given an interval $I_m$ we can compute the value of $\text{MatchCost}(I_m)$ in polynomial time by computing an optimal matching in a bipartite graph.  These algorithms tend to have a relatively large computation time (though polynomial time), and thus in the sequel we will also use a relaxed estimated cost $\text{EstCost}(I_m)$ defined below. 
It is trivial to note that for any valid interval $I_m$, the auxiliary heuristic score provides a mathematical lower bound to the true matching cost, i.e., $\text{EstCost}(I_m) \le \text{MatchCost}(I_m)$.

\begin{definition}[Auxiliary Heuristic Estimation Cost]
An auxiliary estimated cost, $\text{EstCost}(I_m)$, is defined by relaxing the injective constraint by allowing matching with replacement:
\begin{equation}
\text{EstCost}(I_m) = \sum_{i \in \mathcal{T}_m} \min_{j \in \mathcal{C}_m} d(x_i, x_j)
\end{equation}
\noindent where each treated unit is independently mapped to its single $1$-nearest neighbor ($1$-NN) in the control pool within the same stratum. If $\mathcal{C}_m = \emptyset$, then $\text{EstCost}(I_m) = \infty$. 
\end{definition}

\subsection{Stage 1: Dynamic Programming Interval Initial Selection}
\label{DP_Step}
The dynamic programming (DP) algorithm identifies the sequence of stratum boundaries $\{p_1, \dots, p_{k-1}\}$ that minimizes a global auxiliary objective function. Crucially, during this recursive search phase, the true matching cost $\text{MatchCost}(I_m)$ is not explicitly computed due to prohibitive computational overhead. Instead, the DP algorithm uses the fast heuristic proxy.  That is, our goal is to solve
$\min \sum_{m=1}^{k} \text{EstCost}(I_m)$ subject to the three constraints stated above.

\subsubsection{Interval Types}
To ensure complete clarity regarding the feasibility boundary, a strict operational distinction is maintained between three distinct structural states of an interval. Any feasible partition of some interval $I$ is combined of these three distinct types of sub-intervals: 
\begin{enumerate}
    \item \textbf{Empty (or Degenerate) Interval ($n_{m,T} = 0$ and $n_{m,C} = 0$):} An open interval that captures zero units of either type. Note 
    that it does not violate the fine-balance constraint. However, we will show in Lemma~\ref{lem:feasibility_bound} that if  $\delta$ is chosen in the correct range of values, such empty intervals are redundant. 
    \item \textbf{Controls-Only Interval ($n_{m,T} = 0$ and $n_{m,C} > 0$):} An interval that captures only control units and zero treated units. Because $n_{m,C} > 0$ satisfies the fine-balance constraint when $n_{m,T} = 0$, this case is feasible. Crucially, it does not count as an empty interval; rather, it represents a valid stratum of unmatched reservoirs of control units that the algorithm may utilize or bypass depending on global optimization goals.
    \item \textbf{Balanced Interval}: An interval that has both treated and control units, and satisfies the fine-balance constraint. 
\end{enumerate}

Note that an interval with $n_{m,T} > 0$ and $n_{m,C} = 0$ is infeasible as it explicitly violates the fine-balance constraint ($n_{m,C} \ge n_{m,T}$). 
With a slight abuse of notation, we consider empty intervals as well as controls-only intervals  as any balanced interval, and any one of them is adding one to the total number of sub-intervals in a partition of $I$. 

Figure \ref{fig:partition_example} shows a syntactic example of a feasible solution for the parameter values $\delta = 0.12$ and $K = 10$. The solution has $k=7$ intervals of the three types mentioned above.  

\begin{figure}[htbp]
    \centering
    \includegraphics[width=0.4\textwidth]{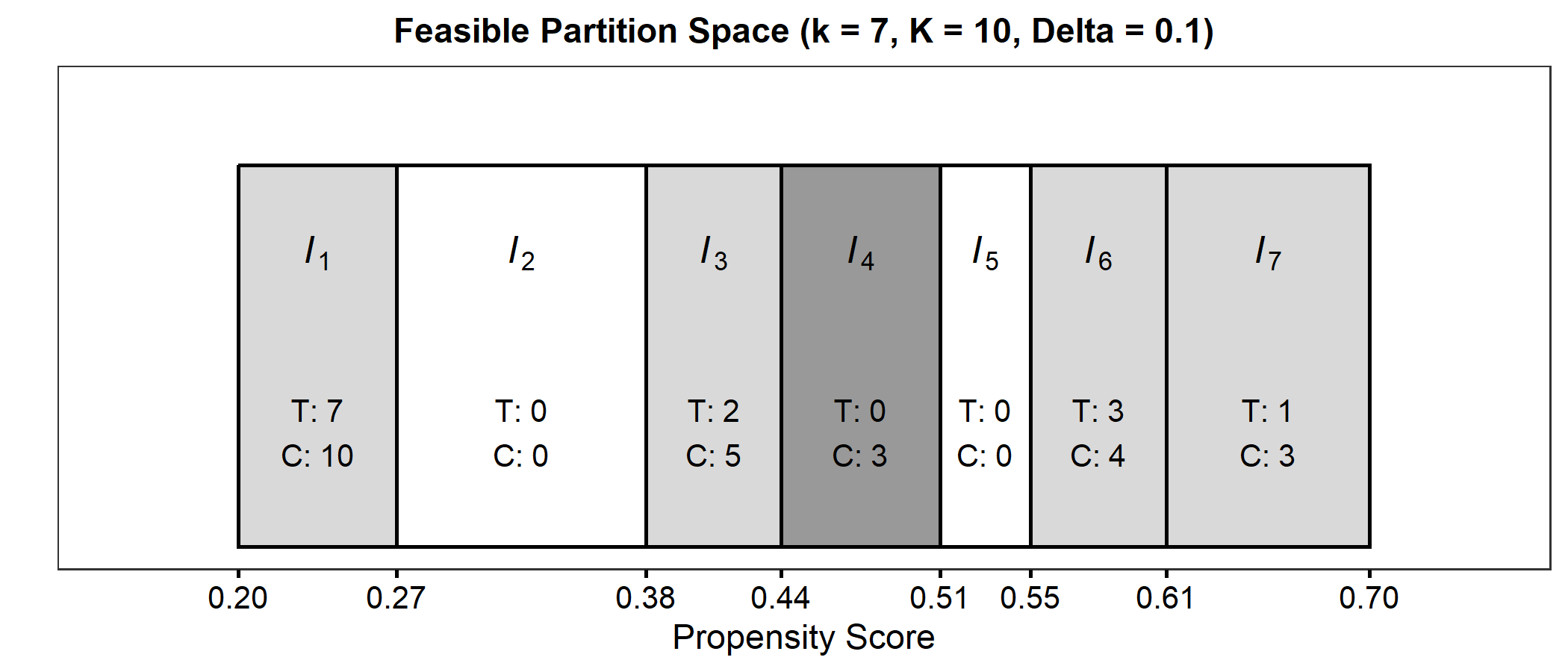}
    \caption{Illustrative example of a feasible solution. $[L_{\mathcal{S}}, R_{\mathcal{S}}] = [0.2, 0.7]$, $\delta = 0.12$, and $K = 10$. The intervals are 
    (1) $I_1 = [0.2, 0.27]$, a balanced interval containing $n_{1,T}=7$ treated and $n_{1,C}=10$ control units; 
    (2) $I_2 = (0.27, 0.38)$, an empty interval ($n_{2,T}=0, n_{2,C}=0$); 
    (3) $I_3 = [0.38, 0.44]$, a balanced interval containing $n_{3,T}=2$ treated and $n_{3,C}=5$ control units; 
    (4) $I_4 = (0.44, 0.51]$, a controls-only interval containing $n_{4,T}=0$ treated and $n_{4,C}=3$ control units; 
    (5) $I_5 = (0.51, 0.55)$, an empty interval ($n_{5,T}=0, n_{5,C}=0$); 
    (6) $I_6 = [0.55, 0.61]$, a balanced interval containing $n_{6,T}=3$ treated and $n_{6,C}=4$ control units; and 
    (7) $I_7 = (0.61, 0.70]$, a balanced interval containing $n_{7,T}=1$ treated and $n_{7,C}=3$ control units.}
    \label{fig:partition_example}
\end{figure}

\subsubsection{Gaps Handling and the Choice of \texorpdfstring{$\delta$}{Delta}}
\label{delta_choice}

We define an empty interval as any data-void stratum of the form $I_m = (p_{m-1}, p_m)$ (or an early truncation of the form $I_m = (p_{m-1}, b)$, where $b > p_{m-1}$ represents an intermediate value along the trimmed interval support) that contains zero sample observations from both the treated and control pools. To prevent the partitioning framework from artificially expanding the number of active strata by concatenating multiple small empty intervals back-to-back to bridge a sparse region, we define a canonical structural assumption: any sequence of consecutive empty intervals can be uniquely consolidated and replaced by a single, longer empty gap. 

Under this consolidation rule, we define a \textbf{large gap} as any single consolidated empty interval $I_m = (p_{m-1}, b)$ whose total width strictly exceeds the maximum allowed threshold, i.e., $b - p_{m-1} > \delta$. 
We first establish the geometric condition required to guarantee that the sample density is sufficient to prevent the occurrence of such large gaps.

Clearly, the choice of the threshold parameter $\delta$ is crucial to the structural feasibility of the partition. We define the baseline maximum individual pairing distance with respect to the local $L_p$ metric along the propensity score scale as:
\begin{equation}
\delta_{\text{pair}} = \max_{i \in \mathcal{T}} \left( \min_{j \in \mathcal{C}} \|x_i - x_j\|_p \right)
\end{equation}

\noindent where $x_i$ and $x_j$ represent the estimated propensity scores for treated unit $i$ and control unit $j$, respectively. Using this definition, we state the following lemma establishing the necessary and sufficient baseline bound on $\delta$ required to eliminate these structural data-voids.  Similarly we define \begin{equation}
\delta_{\text{cons}} = \max_{i \in \mathcal{T}\cup \mathcal{C}: x_i\neq R_{\mathcal{S}}} \left( \min_{j \in \mathcal{T}\cup \mathcal{C}: x_i<x_j} \|x_i - x_j\|_p \right).
\end{equation}

\begin{lemma}
\label{lem:feasibility_bound}
Let $\mathcal{S} = \{x_{(1)}, x_{(2)}, \dots, x_{(N)}\}$ represent the complete set of treated and control units pooled and ordered by their estimated propensity scores along the trimmed support $[L_{\mathcal{S}}, R_{\mathcal{S}}]$. Under the canonical consolidation rule, a valid post-stratification partition into active strata $I_m = (p_{m-1}, p_m]$ satisfying the maximum length constraint $|I_m| \le \delta$ is guaranteed to exist without the forced introduction of large 
gaps if 
$\delta \ge \delta_{\text{cons}}$.  On the other hand, if $\delta < \max\{ \delta_{\text{pair}},\delta_{\text{cons}}\} $, then such partition does not exist. Last, if there is a feasible solution to the problem satisfying all constraints with length parameter $\delta'$, then there is a feasible solution with length parameter $\delta'+\delta_{\text{cons}}$ of the same cost.
\end{lemma}

\begin{proof}
First, suppose that the length constraint parameter $\delta$  
is selected such that $\delta < \max\{ \delta_{\text{pair}},\delta_{\text{cons}}\}$. By definition, either there exists an empty interval of length larger than $\delta$  using the fact that every sequence of consecutive empty intervals consists of at most one interval, or there exists at least one required counterfactual pair whose local discrepancy cannot be bridged within the allowed threshold width on the propensity score scale. 

Next, if $\delta \ge \delta_{\text{cons}}$, then it is impossible that a large gap occurs  by definition of $\delta_{\text{cons}}$. Last assume the feasibility of the problem with $delta$ value of $\delta'$.  By the canonical structural assumption, every non-empty interval is preceded by at most one empty interval whose length is at most $\delta_{\text{cons}}$, and then by uniting this possible empty interval to its neighbor to the left, we get a feasible solution with length bound of at most $\delta'+\delta_{\text{cons}}$ and the same partition of the treatment set as well as the same partition of the control set, and thus the resulting solution has the same cost.
\end{proof}

As a direct result of Lemma~\ref{lem:feasibility_bound} and the upper bound connecting the values of $\delta$ and $K$ stated in \eqref{eq:delta_k_bound}, we obtain a constructive approach for initializing these parameters
to verify that a feasible solution exists without the presence of large gaps. First, we pick a value of $\delta$ that is large enough and in particular satisfies the necessary selection criterion of Lemma~\ref{lem:feasibility_bound} to exclude the possibility of large gaps.

Let $I = [p_0, p_K]$ denote the full trimmed support of the pooled sample, where $p_0 = L_{\mathcal{S}}$ and $p_K = R_{\mathcal{S}}$ represent the lower and upper support bounds, respectively. Given $\delta$, we choose the number of active strata $K$ such that the pair $(\delta, K)$ satisfies the structural baseline relationship:
$K \ge \dfrac{|I|}{\delta}$. However, local gaps of the form $(p_{m-1}, p_m)$, defined as empty intervals whose length is smaller than $\delta$ and that begin and end at observed treated or control units, can still appear in some cases. An example of this behavior is illustrated in the following syntactic example.  Nevertheless, in dense, real-life datasets, such as all the datasets evaluated in this work, this type of gap is statistically rare. Hence, to simplify the dynamic programming formulation, we assume that no such gaps exist. 

{\bf Example.}
Consider a partition of the active interval $I=[0.2,0.37]$ with a partition into three intervals with boundary points $0.2=p_0<0.26=p_1< 0.31=p_2 < 0.37=p_3$ where there are no sample units with scores in the interval $(0.26,0.31)$. Assume that all control units in $[p_0,p_1]$ are at $p_1$, while the treated units in this interval are spread densely along this interval. Similarly, assume that all control units in $[p_2,p_3]$ are at $p_2$ and the treated units are spread densely along this sub-interval. In this example $\delta_{\text{pair}}=0.06$ while $\delta_{\text{cons}}=0.05$. Then a feasible solution satisfying the fine-balance constraint cannot partition $I_1$ or $I_3$.  For $\delta>0.11$, the given partition is a unique feasible solution and we cannot merge two intervals due to the length constraint. For $0.06\leq \delta<0.11$, the given partition is the unique feasible solution admitting an empty interval but no large empty interval.     

For the Lindner dataset, the maximum absolute distance between a control unit and its nearest treated neighbor is calculated as:
$|0.6878247 - 0.6740549| = 0.01376973$.
As illustrated in Figure~\ref{fig:ps_trimmed}, this maximum is driven by the propensity score values of the two most extreme units remaining in the trimmed distribution. Similarly, this is also the value of $\delta_{\text{pair}}$. Consequently, by Lemma~\ref{lem:feasibility_bound} selecting a threshold parameter $\delta < 0.01376973$ would fail to bridge this boundary gap, structurally generating an empty interval of the form $(p_{m-1}, b)$. So, we should pick a value of $\delta$ that is larger than $0.0138$, then pick a value of $K$ larger than or equal to $\lceil \dfrac{1}{\delta} \rceil$. In practice (See Section~\ref{sub:comp_delta}) we pick much larger values of $\delta$, so this condition is satisfied. 

Furthermore, this geometric observation allows us to rule out redundant splits within already feasible strata. Consider an active stratum $I_m = (p_{m-1}, p_m]$ that already satisfies the length constraint, i.e., $p_m - p_{m-1} \le \delta$. Introducing an intermediate boundary point $b \in (p_{m-1}, p_m)$ to divide $I_m$ into two non-empty sub-strata $(p_{m-1}, b]$ and $(b, p_m]$ yields no reduction in objective cost.

 Such a split yields zero reduction in the objective function while consuming an additional unit of the stratum allocation budget $K$. Consequently, our dynamic programming state space search completely circumvents these redundant states, significantly reducing computational complexity without loss of global optimality.

\subsubsection{Dynamic Programming Formulation and Constraints}

Given a pair of parameters $(\delta, K)$ selected in accordance with the prior framework such that all large gaps are structurally omitted, the optimization problem reduces to finding an optimal sequence of active cutting points over the ordered empirical grid of available propensity scores. Crucially, our dynamic programming formulation restricts the search space such that any interval boundary must start and end precisely on an observed unit's propensity score within the pooled, sorted grid $\mathcal{P}_{\text{sorted}} = \{p_1, p_2, \dots, p_{N_p}\}$, which comprises all unique treated and control scores, with respect to the \textit{trimmed} range of the propensity score.  

Furthermore, to maintain local stratum homogeneity, every candidate interval must satisfy the strict maximum length constraint. 

For a partition consisting of $m$ active strata where $m \in \{1, \dots, K\}$, we let $v_{m-1}$ and $v_m$ denote the grid indices in $\mathcal{P}_{\text{sorted}}$ corresponding to the left and right boundaries of the $m$-th interval. The length of any proposed sub-interval $(p_{v_{m-1}}, p_{v_m}]$ cannot exceed the length parameter threshold $\delta$.

In addition to this geometric threshold, the algorithm enforces a fundamental capacity check to satisfy the matching structure: the number of available control units falling within the boundaries of the candidate sub-interval $(p_{v_{m-1}}, p_{v_m}]$ must satisfy the fine-balance constraint established in \eqref{eq:fine_balance_capacity}. 
Any forward step violating one of these two conditions is structurally inadmissible and omitted from the state transitions.

Let $V(i, v, m)$ represent the value function yielding the minimum cumulative estimated matching cost given that the first $i$ ordered treated units have been covered, and the rightmost boundary of the $m$-th interval resides exactly at index $v$ of the pooled grid $\mathcal{P}_{\text{sorted}}$. Transitioning forward to establish the next stratum boundary at index $v_m$ from a previous one at index $v_{m-1}$, the recurrence relation is formalized from a forward-looking perspective as:
\begin{equation}
V(i_m, v_m, m) = \min_{v_{m-1} < v_m \,:\, p_{v_m} - p_{v_{m-1}} \le \delta} \left\{ V(i_{m-1}, v_{m-1}, m-1) + \text{EstCost}\left(p_{v_{m-1}}, p_{v_m}\right] \right\} 
\end{equation}
where the state transitions from a valid prior boundary index $v_{m-1}$ to a forward horizon index $v_m$. Such a transition captures the sequential block of treated units indexed from $i_{m-1} + 1$ through $i_m$, and incurs an incremental local \textit{estimated} matching cost defined over the interval $(p_{v_{m-1}}, p_{v_m}]$. Here, we impose the fine-balance constraint directly onto the state-space transitions by defining $\text{EstCost}(p_{v_{m-1}}, p_{v_m}] = \infty$ if the candidate interval fails to satisfy the unit capacity requirement established in \eqref{eq:fine_balance_capacity}.  While the state label contains the value of $i_m$ as its first coordinate, it is completely identified by the value of $v_m$ and we use it in the label for notational ease.

{\bf Base Cases and Grid Initialization.}
The base cases for the recurrence relation directly mirror the structural initialization of the multi-dimensional state list. We initialize the primary anchor state at zero cost:
$V(0, 1, 0) = 0$
which represents a baseline of zero matching cost for allocating zero units across zero intervals at the first index of the sorted propensity score grid. Conversely, all other configurations are initialized to infinity:
$V(i, v, m) = \infty \quad \forall (i, v, m) \neq (0, 1, 0)$
reflecting that it is structurally impossible to match sample observations without valid path progressions.  Using this base case and the recurrence formula, we compute the values of $V$ for every entry in the dynamic programming table (i.e., we increase the value of $v_m$ until we reach the right boundary of the trimmed common support).

{\bf The final solution obtained by the dynamic programming.}
The ultimate goal of the algorithm is to identify the globally optimal partition profile by evaluating the minimum over all feasible terminal states that span the entire trimmed support. A state is structurally recognized as terminal the moment its right interval boundary reaches the designated \texttt{right\_border} (corresponding to the terminal index $n_p$ of the grid), signaling that the entire support has been successfully partitioned. Formally, the optimization problem solves for
$V^* = \min_{1 \le m \le K} V(n_t, n_p, m)$
where $n_t$ represents the total count of matched treated units, $n_p$ corresponds to the terminal index of $\mathcal{P}_{\text{sorted}}$, and $K$ is the maximum permitted stratification capacity.

For the \texttt{Lindner} dataset, the optimal intervals after the operation of the dynamic programming step (\ref{DP_Step}), for $\delta = 0.2$, $K = 10$ is shown in Table~\ref{tab:best_intervals_step1}. It can be observed that the dynamic programming solution has only $4$ sub-intervals, and it satisfies all constraints. We present in Table \ref{tab:best_intervals_step1} the counts of treated and control units within each sub-interval, along with its lower and upper boundaries. The \texttt{estimated} cost for this solution, using the standard Mahalanobis distance metric, rounded to four decimal places, is $125.0851$.

\begin{table}[htbp]
  \centering
  \caption{Optimal Sub-interval Bounds At the End of Step 1}
  \label{tab:best_intervals_step1}
  \begin{filecontents*}{step1_optimal_partition_delta_0.20.csv}
start,end,treated,control
0.0417,0.2414,51,270
0.2414,0.4367,165,329
0.4367,0.6252,80,83
0.6252,0.6878,1,2
\end{filecontents*}
\pgfplotstabletypeset[
    col sep=comma,
    string type,
    every head row/.style={before row=\toprule, after row=\midrule},
    every last row/.style={after row=\bottomrule}
  ]{step1_optimal_partition_delta_0.20.csv}
\end{table}

\subsection{Stage 2: Heuristic search to optimize the solution}
\label{Heuristics_Step}

The Dynamic Programming (DP) algorithm identifies a solution 
with respect to the propensity score and the $L_p$ norm,  while strictly adhering to the previously described constraints. Because this approach satisfies the complete constraint set, it inherently yields a feasible solution for the original high-dimensional covariate space, as feasibility is independent of the underlying distance metric. 

However, this solution is not necessarily near-optimal when evaluated in the full covariate space, even with respect to the estimated matching cost. Incorporating a multi-dimensional distance function inherently increases the complexity of the underlying optimization problem. Therefore, we utilize the output of the dynamic programming algorithm as an initial baseline for the heuristic search procedures designed to refine the solution under the multi-dimensional distance function. Specifically, we first evaluate Stage~$1$ baseline solution using the multi-dimensional Mahalanobis distance metric (or rank-based Mahalanobis distance as an alternative), as defined in Section~\ref{multidimmet}.

 We then apply heuristic search methods to iteratively improve (in terms of the estimated cost) upon this solution with respect to the Mahalanobis 
distance (or using the rank-based Mahalanobis as an alternative), thereby lowering the total cost, while strictly preserving all 
relevant constraints to guarantee continued feasibility. For this purpose, we use three 
heuristic methods, namely, \textbf{local search}, \textbf{enhanced local search} and \textbf{simulated annealing}. The general approach of using a local search based heuristics is a standard approach in the design of heuristics for a combinatorial optimization problem but their details are problem-specific and we define those details next. In each of these heuristics, the algorithm maintains a feasible solution (to the underlying optimization problem) and defines a neighborhood (of feasible solutions) of any such feasible solution. 

The search algorithm systematically explores the local neighborhood of the current feasible solution. If an adjacent candidate solution yields an improved objective cost while preserving feasibility, the algorithm accepts the move, updates the current state, and constructs a new neighborhood around the accepted solution. This iterative process continues until no neighbor offers a strict improvement (or an alternative termination criterion is satisfied), at which point the algorithm terminates and returns the final local optimum as the heuristic output.

\subsubsection{Local Search (LS)}

While the dynamic programming framework provides a feasible solution, the resulting boundary set can be further refined to unlock the \textit{estimated} cost reductions by systematically exploring adjacent grid arrangements. To exploit these local adjustments, we implement a customized, multi-layered \textbf{Local Search} (LS) \citep{talbi2009metaheuristics} framework that operates directly on the discrete indices of our empirical grid.

Let $\mathcal{B} = \{b_1, b_2, \dots, b_M\}$ represent the vector of interior partition boundary indices mapped directly onto the sorted, pooled propensity score grid $\mathcal{P}_{\text{sorted}}$, where $n_{ps} = |\mathcal{P}_{\text{sorted}}|$. The neighborhood of the current partition, $\mathcal{N}(\mathcal{B})$, is explored by varying individual boundary indices. For any given interior boundary index $b_m$ (where $m \in \{2, \dots, M-1\}$), a candidate boundary index is generated by applying an integer step size along either direction:
\begin{equation}
b_{\text{cand},\gamma} = b_m + d \cdot \gamma
\end{equation}
where $d \in \{-1, 1\}$ dictates the sweep direction, and $\gamma \in \{1, 2, \dots, \gamma_{\max}\}$ represents a step size integer parameter. The upper bound $\gamma_{\max}$ is a user-defined parameter dictating the maximum allowable jump size. By modifying indices over the discrete vector $\mathcal{P}_{\text{sorted}}$, this coordinate-shifting operation translates to changing the cut-points to alternative unique propensity score values.

To ensure mathematical validity before executing a move, every candidate boundary configuration must maintain feasibility, that is,  satisfy three explicit structural guardrails:
\begin{enumerate}
    \item \textbf{Boundary Preservation:} The candidate index must remain strictly within the interior limits of the empirical array, satisfying that every component of $b_{\text{cand},\gamma}$ is strictly larger than $1$ and strictly smaller than $n_{ps}$.
    \item \textbf{Strict Monotonicity:} To prevent adjacent intervals from collapsing into zero-length domains or crossing over one another, the updated index must preserve strict inequality relative to its neighboring boundary points:
  $b_{m-1} < b_{\text{cand},\gamma} < b_{m+1}$.
    \item \textbf{Feasibility Hard-Constraints:} The distance between the candidate boundary and its unmodified neighbors on the scalar propensity score scale must not violate the length constraint:
    \begin{equation}
    p_{b_{\text{cand},\gamma}} - p_{b_{m-1}} \le \delta \quad \text{and} \quad p_{b_{m+1}} - p_{b_{\text{cand},\gamma}} \le \delta \ .
        \end{equation}
        Any candidate step that fails these conditions is instantly omitted from consideration.
\end{enumerate}

The search follows a \textit{First-Improvement} selection strategy. The algorithm sweeps through the grid neighborhoods to shift the index positions. If a boundary shift fails to yield an estimated cost reduction exceeding our precision tolerance $\epsilon_{\text{cost}} = 10^{-9}$, the local search rejects the proposed move and retains the original boundary position.
Crucially, the moment any valid neighborhood move decreases the global estimated cost by at least $\epsilon_{\text{cost}}$, the proposed configuration becomes  the new baseline solution and we repeat the process of defining its neighborhood, and restarts a brand-new outer search cycle beginning from the first interior index $m = 2$. This compounding, greedy reset loop continues iteratively until a complete left-to-right pass across all interior boundaries yields no modifications.

For the \texttt{Lindner} dataset, the final solution after the operation of the local-search (LS) heuristics for $\delta = 0.2$, $K = 10$ is shown in Table~\ref{tab:ls_intervals}. Note that the resulting solution still has only $4$ sub-intervals. The \texttt{estimated} cost for this solution, using the standard Mahalanobis distance metric, rounded to four decimal places is $121.6308$.

\begin{table}[htbp]
  \centering
  \caption{Sub-interval Bounds At the End of Step 2 (LS)}
  \label{tab:ls_intervals}
  \begin{filecontents*}{chapter4_schemes/step2_heuristic_local_search_strict_partition_delta_0.20.csv}
start,end,treated,control
0.0417,0.2389,47,266
0.2389,0.4312,165,329
0.4312,0.6252,84,86
0.6252,0.6878,1,3
\end{filecontents*}
\pgfplotstabletypeset[
    col sep=comma,
    string type,
    every head row/.style={before row=\toprule, after row=\midrule},
    every last row/.style={after row=\bottomrule}
  ]{chapter4_schemes/step2_heuristic_local_search_strict_partition_delta_0.20.csv}
\end{table}

\subsubsection{Enhanced Local Search (ELS)}

Exploration via \textbf{Enhanced Local Search} (ELS) builds directly upon the standard local search framework to offer a more flexible exploration strategy, i.e., we extend the neighborhood of each feasible solution. While it retains the core boundary shifting mechanism of the standard local search, ELS expands the neighborhood search space by introducing two additional topological operations: stratum splitting and stratum merging. Mechanically, ELS can be viewed as a specific instantiation of a broader meta-heuristic class designed to systematically transition across multiple distinct neighborhood structures, formally referred to as \textbf{Variable Neighborhood Descent} (VND) \citep{talbi2009metaheuristics}. Because the boundary shifting mechanism mirrors the standard local search logic described previously, we focus our exposition on the two new structural operations: stratum splitting and stratum merging.


\textbf{Stratum Splitting}:
When the shifting operator does not produce an improvement, the algorithm attempts to increase the granularity of the partition by adding a new boundary, expanding the number of strata from $k$ to $k+1$. This structural change is bounded by the global maximum partition size limit $K$, 
namely only if $k+1\leq K$.  So we still require that every intermediate solution satisfies this constraint (as well as all other constraints).

For each individual stratum $m \in \{1, \dots, k\}$ defined by its bounding array indices $b_m$ and $b_{m+1}$, a split is structurally feasible if and only if the number of unique elements between these boundaries allows for further partition adjustments. This requires the index distance to be strictly greater than two ($b_{m+1} - b_m > 2$). If this condition is met, a new candidate boundary index is calculated as the discrete midpoint:
\begin{equation*}
b_{\text{split}} = \left\lfloor \frac{b_m + b_{m+1}}{2} \right\rfloor .
\end{equation*}
Recalling that a partition of $k$ strata is defined by $M = k + 1$ boundaries, the new candidate partition vector $\mathcal{B}_{\text{cand}}$ is constructed by inserting $b_{\text{split}}$ into the sequence, expanding the total number of boundary indices to $M+1$:
\begin{equation}
\mathcal{B}_{\text{cand}} = \{b_1, \dots, b_m, b_{\text{split}}, b_{m+1}, \dots, b_{k+1}\}
\end{equation}
Here, the length constraint 
is satisfied by the candidate solution. This holds since we added a middle point inside an interval the maximum length of any sub-interval can only reduce,  
and the constraint was satisfied by the original solution. As a result, we just need to check whether the candidate solution satisfies the fine-balance constraint.

\textbf{Stratum Merging}:
Conversely, if neither the shifting nor the splitting operators reduce the \textit{estimated} cost objective, the algorithm attempts to decrease the granularity of the partition by deleting an interior boundary. This operator collapses two adjacent strata into a single, combined interval, reducing the number of strata from $k$ to $k-1$. 

To ensure the partition remains structurally valid, a merge is only attempted if the current number of strata is strictly greater than one ($k > 1$, or equivalently, the length of the boundary vector satisfies $M > 2$). The algorithm systematically evaluates each interior boundary index $b_m$ where $m \in \{2, \dots, k\}$. The candidate partition vector $\mathcal{B}_{\text{cand}}$ of reduced length $M-1$ is constructed by omitting that specific interior boundary from the sequence:
\begin{equation}
\mathcal{B}_{\text{cand}} = \{b_1, \dots, b_{m-1}, b_{m+1}, \dots, b_{k+1}\}
\end{equation}
By removing the cut-point $b_m$, the units between indices $b_{m-1}$ and $b_{m+1}$, which were previously divided into two separate strata, are seamlessly fused into a single larger stratum.  Here, we are guaranteed that the fine-balance constraint will be satisfied (as well as the upper bound on $K$ on the number of strata) since the original solution satisfied it, and we only need to check the length constraint of the merged stratum.

This completes the characterization of the neighborhood of the current feasible solution. Similarly to the standard local search, the ELS algorithm adopts a \textit{First-Improvement} strategy among its operators. The search routine sequentially evaluates transitions across its neighborhood layers: boundary shifting, stratum splitting, and stratum merging. 
As with the local search heuristic, (estimated) cost improvements at every stage are evaluated against our tolerance threshold, $\epsilon_{\text{cost}}$.

For the \texttt{Lindner} dataset, the final solution after the operation of the enhanced local-search (ELS) heuristic for $\delta = 0.2$, $K = 10$ is shown in Table~\ref{tab:els_intervals}. The \texttt{estimated} cost for this solution, using the standard Mahalanobis distance metric, rounded to four decimal places is $121.6308$. We can observe that for this case, LS and ELS found similar, though not identical, solutions while achieving the exact same \texttt{estimated} cost.  In particular, the partition of the treated units into the $4$ sub-intervals is identical for the two solutions while the partition of the control units differ between them.

\begin{table}[htbp]
  \centering
  \caption{Sub-interval Bounds At the End of Step 2 (ELS)}
  \label{tab:els_intervals}
  \begin{filecontents*}{chapter4_schemes/step2_heuristic_enhanced_ls_strict_partition_delta_0.20.csv}
start,end,treated,control
0.0417,0.2375,47,264
0.2375,0.4307,165,330
0.4307,0.6252,84,87
0.6252,0.6878,1,3
\end{filecontents*}
\pgfplotstabletypeset[
    col sep=comma,
    string type,
    every head row/.style={before row=\toprule, after row=\midrule},
    every last row/.style={after row=\bottomrule}
  ]{chapter4_schemes/step2_heuristic_enhanced_ls_strict_partition_delta_0.20.csv}
\end{table}

\subsubsection{Simulated Annealing (SA)}

Unlike the deterministic hill-climbing approaches of LS and ELS, \textbf{Simulated Annealing} (SA) \citep{metropolis1953equation, talbi2009metaheuristics} is a single-solution stochastic metaheuristic designed to escape sub-optimal local basins. While SA explores the same fundamental structural operation as ELS, namely boundary shifting, stratum splitting, and stratum merging, its execution architecture is fundamentally different. Rather than evaluating these operators within a strict, sequential layered hierarchy where subsequent operators act only as fallbacks, SA selects a single mutation type at each iteration by sampling from a categorical distribution governed by a set of operational probabilities, $P = \{p_{\text{shift}}, p_{\text{split}}, p_{\text{merge}}\}$, subject to $\sum p_i = 1$ (parameterized for this implementation as $0.5$, $0.3$, and $0.2$, respectively).

Furthermore, while ELS exhaustively searches an entire neighborhood to find the first valid improving move, SA utilizes a randomized neighborhood sampling scheme. Once an operator is selected, the algorithm stochastically draws a single candidate configuration vector, $\mathcal{B}_{\text{cand}}$, from that operator's specific transition space.

Like the baseline Local Search, the current baseline partition layout is defined as an ordered sequence of boundary indices
$\mathcal{B} = \{b_1, b_2, \dots, b_M\}$
where $M = k+1$. At each iteration, the algorithm tracks a global thermodynamic control parameter known as the system temperature, $T$. The temperature begins at a high initial value ($T_{\text{init}}$) and decays geometrically over time according to a cooling schedule:
\begin{equation}
T_{t+1} = \alpha \cdot T_t
\end{equation}
where $\alpha \in (0, 1)$ is the cooling rate.  That is, $T$ is decreased geometrically by a multiplicative factor $\alpha$ in each iteration.

\paragraph{The Acceptance Criterion}
Once a candidate partition vector $\mathcal{B}_{\text{cand}}$ is generated, its \textit{estimated} cost $f(\mathcal{B}_{\text{cand}})$ is evaluated against the current state cost $f(\mathcal{B})$. Moves that strictly improve the cost function are always accepted. Conversely, non-improving moves that degrade the estimated cost are accepted with some probability. Let $\Delta E$ be the difference in the estimated cost between the generated neighbor and the current solution:
$\Delta E = f(\mathcal{B}_{\text{cand}}) - f(\mathcal{B})$.
The probability $P(\Delta E, T)$ of accepting a non-improving state is a function of the temperature and the last difference and it is defined as:
$$P(\Delta E, T) = \exp\left(-\frac{\Delta E}{T}\right) = \mathrm{e}^{-\frac{f(\mathcal{B}_{\text{cand}}) - f(\mathcal{B})}{T}}. $$
This formulation ensures that when the system temperature $T$ is high, the algorithm maintains a high probability of accepting worse solutions to escape local minima. As the cooling schedule progresses and $T \to 0$, this probability decreases, causing the algorithm to freeze into a strictly greedy local exploitation phase. This acceptance criterion is known as the Metropolis criterion \citep{metropolis1953equation}. 

\paragraph{Stochastic Mutation Operators}
The rules sets for generating the candidate partition sequence $\mathcal{B}_{\text{cand}}$ depend on the mutation operator chosen:

\textbf{Random Boundary Shifting}: 
    This operator is structurally eligible for evaluation if and only if there is at least one interior boundary that can be shifted ($M > 2$, or equivalently, the number of strata $k > 1$).   
    If this condition is met, the algorithm randomly and uniformly selects a single interior boundary index position $m \in \{2, \dots, k\}$. It then samples a discrete displacement step temp-buffer uniformly from a bounded exploration window governed by a scaled function. The definition of $\textbf{temp-buffer}$ is the division of the current \textbf{system temperature}, $T$, by the $\textbf{temp-scale-factor}$ parameter: 
$\text{temp-buffer} = \left\lceil \frac{T}{\text{temp-scale-factor}} \right\rceil$.
A discrete displacement step, $s$, is then sampled uniformly from a bounded exploration window defined by the current mutation strength:
$s \in \{-\text{strength}, \dots, \text{strength}\}$, where  $\text{strength} = \max(1, \text{temp-buffer})$.
In our implementation, the initial system temperature value, $T_{\text{init}}$, is set to $50$. The scaling denominator variable, $\text{temp-scale-factor}$, which acts as a constant scalar parameter set to $10$, maps this initial high temperature down to a reasonable spatial search radius, while the ceiling operator ensures the search radius collapses to a minimum step size of $\pm 1$ index unit as $T \to 0$.  As a direct result of these initializations, the starting value of $\text{temp-buffer}$ is exactly $5$.

The proposed new position is then calculated directly as:
$b_{\text{prop}} = b_m + s$.
To preserve strict index sequence ordering and prevent boundaries from stepping over one another, the proposed index is clipped using its immediate neighbors as hard  constraints, namely, $b_{\text{new}} = \max(b_{m-1} + 1, \min(b_{m+1} - 1, b_{\text{prop}}))$.
    The candidate partition vector $\mathcal{B}_{\text{cand}}$ maintains a length of $M$ and is formed by substituting the original boundary $b_m$ with $b_{\text{new}}$, that is, $\mathcal{B}_{\text{cand}} = \{b_1, \dots, b_{m-1}, b_{\text{new}}, b_{m+1}, \dots, b_{k+1}\}$.

\textbf{Random Stratum Splitting}: 
    This structural expansion is only attempted if the current partition configuration has not yet reached the user-defined capacity threshold (i.e., $k < K_{\max}$).   
    If a split is permissible, the algorithm randomly and uniformly selects a target stratum index $m \in \{1, \dots, k\}$. To ensure that at least one unallocated unit index resides within the chosen interval to accommodate a new boundary, the bounding distance must satisfy $b_{m+1} - b_m > 1$. If this holds, a new cut-point position $b_{\text{new}}$ is sampled uniformly at random from the available discrete integer range within that stratum, that is, an integer in the interval in $\{b_m + 1, \dots, b_{m+1} - 1\}$.
    The candidate partition vector $\mathcal{B}_{\text{cand}}$ of expanded length $M+1$ is constructed by inserting $b_{\text{new}}$ into its ordered sequence position, that is, $ \mathcal{B}_{\text{cand}} = \{b_1, \dots, b_m, b_{\text{new}}, b_{m+1}, \dots, b_{k+1}\}$.

\textbf{Random Stratum Merging}: 
    This structural reduction collapses the partition landscape and is only evaluated if the current baseline vector contains more than two elements ($M > 2$, or equivalently, $k > 1$).     
    If a merge is valid, the algorithm randomly and uniformly selects a single interior boundary index position $m \in \{2, \dots, k\}$. The candidate partition vector $\mathcal{B}_{\text{cand}}$ of compressed length $M-1$ is constructed by dropping that specific boundary cut entirely:
$    \mathcal{B}_{\text{cand}} = \{b_1, \dots, b_{m-1}, b_{m+1}, \dots, b_{k+1}\}$.
    This operation fuses the units bounded between $b_{m-1}$ and $b_{m+1}$ into a single coarser stratum.
As in previous the deterministic heuristics (LS and ELS), estimated cost improvements at every stage are evaluated against a strict tolerance threshold, $\epsilon_{\text{cost}}$.

For the \texttt{Lindner} dataset, the final solution  after the operation of simulated annealing (SA) heuristic for a specific run (i.e., a realization of the random number used by the method) with $\delta = 0.2$, $K = 10$ is shown in Table~\ref{tab:sa_intervals}. In this solution SA resulted in $4$ sub-intervals. The \texttt{estimated} cost for this solution, using the standard Mahalanobis distance metric, rounded to four decimal places is $122.4996$. Hence, for this dataset, parameter values of $\delta$ and $K$, and this specific run, LS and ELS prove superior to SA, with respect to the estimated cost criterion. However, because SA is a randomized algorithm, across multiple runs it also identified a solution with the exact same \texttt{estimated} cost as LS and ELS.

\begin{table}[htbp]
  \centering
  \caption{Sub-interval Bounds At the End of Step 2 (SA)}
  \label{tab:sa_intervals}
  \begin{filecontents*}{chapter4_schemes/step2_heuristic_simulated_annealing_strict_partition_delta_0.20.csv}
start,end,treated,control
0.0417,0.2375,47,264
0.2375,0.4367,168,334
0.4367,0.6252,81,83
0.6252,0.6878,1,3
\end{filecontents*}
\pgfplotstabletypeset[
    col sep=comma,
    string type,
    every head row/.style={before row=\toprule, after row=\midrule},
    every last row/.style={after row=\bottomrule}
  ]{chapter4_schemes/step2_heuristic_simulated_annealing_strict_partition_delta_0.20.csv}
\end{table}

\subsubsection{Final Matching and Cost}
After the three search heuristics (LS, ELS, and SA) terminate, the structural partition layout with the lowest estimated cost is selected. If none of the heuristics improve upon the initial configuration found by the dynamic programming step, the baseline layout (i.e., the one obtained by the dynamic programming stage) is retained. 
Once these final stratum boundaries are established, the actual matching is performed using the \textbf{MatchIt} package in R \citep{ho2007matching}. Specifically, the treated and control units are strictly matched within their shared propensity score bin by invoking exact stratification. Within each isolated bin, an optimal pair matching configuration $1:1$ is computed by network flow optimization, using the strict or robust (multidimensional) Mahalanobis distance metric. Because the exact stratification restricts the matching space to within-bin candidate units, the overall optimization problem decomposes by partition. The global objective function minimized by the \textbf{MatchIt} package is the sum of the absolute pairwise distances aggregated across all $B$ distinct strata. The resulting total distance calculation serves as the definitive empirical cost of the OSIP solution:
\begin{equation}
\text{Cost}_{\text{OSIP}} = \sum_{b=1}^{B} \sum_{i \in \mathcal{M}_b} D(x_{1i}, x_{0i})
\end{equation}
where $B$ is the total number of optimized bins, $\mathcal{M}_b$ is the set of matched pairs within bin $b$, and $D(x_{1i}, x_{0i})$ represents the multidimensional Mahalanobis distance between the treated unit and its selected control matched unit.

\section{The comparison metrics}
\label{sec:metrics}
In this section, we present the primary performance metrics used for comparing OSIP against benchmark methodologies (and also comparing OSIP to itself, with different parameters). These metrics are generally categorized into three distinct classes: statistical significance measures, distance-based estimation measures (applicable strictly to outputs of the $1:1$ matching method), and covariate balancing comparisons.  

First, we present the metrics used for the comparison. Later, we show the results and analysis for the \texttt{Lindner} dataset.
For ease of presentation, the result of any metrics described here is rounded to four decimal digits.

Since some of the metrics are relevant only for $1:1$ matching methods, we partition our evaluated methods into $1:1$ and non-$1:1$ matching categories. The $1:1$ matching methods analyzed in this work are \textit{OSIP}, \textit{Quintile Subclassification}, \textit{Refined Quintile Subclassification}, \textit{Genetic Matching}, \textit{Optimal Matching}, and \textit{Cardinality Matching}. The non-$1:1$ methods are \textit{CEM} and \textit{Optimal Full Matching}. We note that while certain methods (such as \textit{Genetic Matching} and \textit{Cardinality Matching}) admit generalized variants that support non-$1:1$ matching ratios, we restrict our focus and analysis strictly to their $1:1$ matching variants in this paper.

\subsection{Statistical significance metrics}
To mathematically give our experimental comparisons a causal interpretation, we invoke the foundational unconfoundedness assumption (also frequently formalized as the \textit{conditional independence}, \textit{strong ignorability}, or \textit{selection-on-observables} assumption) \citep{rosenbaum1983central, imbens2015causal}. Formally, the assumption states that conditional on the observed baseline covariates $\mathbf{X}_i$, treatment assignment is statistically independent of the potential outcomes, that is, $(Y_i(1), Y_i(0)) \perp \!\!\! \perp Z_i \mid \mathbf{X}_i$.
Under this condition, we assume there are no unobserved confounding variables simultaneously influencing both the treatment assignment mechanism and the final outcome.  While we make this standard assumption, we try to bound the impact of such potential unobserved confounding variables. Within the OSIP framework, our multi-stage heuristic partitioning and matching pipeline operates directly on the observed covariate space $\mathbf{X}_i$ to ensure that this conditional independence holds empirically via robust structural balancing. Once balance is established across the joint distribution of $\mathbf{X}_i$, the unconfoundedness assumption allows us to unbiasedly isolate and estimate our target causal metrics defined as follows. 

    \textbf{1. Average Treatment Effect} (ATE). It measures the mean causal effect of the treatment across the entire population under evaluation \citep{imbens2015causal, rosenbaum1983central}. Formally, the true population parameter is defined as the difference in population expectations of the potential outcomes, 
    $\tau_{\text{ATE}} = E(Y(1)) - E(Y(0))$.
    Following the agnostic regression adjustment framework established by \citet{lin2013agnostic}, the operational sample estimation of this quantity is computed via Lin's covariate-adjusted linear regression estimator ($\hat{\tau}_{\text{adj}}$). Utilizing the post-match sample design weights ($w_i$) generated by the OSIP pipeline to anchor global balance, the data are fitted using a Weighted Least Squares (WLS) regression model:
    \begin{equation}
    \label{eq:ate_formula}
    Y_i = \alpha + \tau_{\text{ATE}} Z_i + \mathbf{\beta}^T(\mathbf{X}_i) + \epsilon_i
    \end{equation}
    where $Y_i$ represents the continuous outcome, $Z_i \in \{0,1\}$ is the binary treatment indicator, $\mathbf{X}_i$ is the vector of baseline covariates included for residual bias adjustment, and $\epsilon_i$ represents the random residual error term. By omitting treatment-covariate interaction terms, the model algebraically enforces uniform covariate slopes ($\mathbf{\beta}_1 = \mathbf{\beta}_0 = \mathbf{\beta}$) across both treatment assignments. As mathematically demonstrated by \citet{lin2013agnostic}, the extracted sample coefficient $\hat{\tau}_{\text{ATE}}$ consistently recovers the global population parameter $\tau_{\text{ATE}}$ under this agnostic framework, ensuring robustness to parametric model misspecification while absorbing residual variance to maximize statistical precision. 
    In our framework, the model in \eqref{eq:ate_formula} is implemented using the \texttt{lm()} function from the R \texttt{stats} package.

    \textbf{2. Standard Error of the Estimator} ($\text{SE}(\hat{\tau}_{\text{ATE}})$). This metric quantifies the sampling uncertainty and statistical precision of the estimated treatment effect. Extracted directly from the diagonal of the model's asymptotic variance-covariance matrix, it reflects the estimated standard deviation of $\hat{\tau}_{\text{ATE}}$'s sampling distribution conditional on the matched sample structure.

   \textbf{3. t-statistic}. The calculated test statistic used to evaluate the significance of the treatment coefficient against critical values:
    \begin{equation}
    t_{\text{stat}} = \frac{\hat{\tau}_{\text{ATE}}}{\text{SE}(\hat{\tau}_{\text{ATE}})}
    \end{equation}
    
    \textbf{4. Confidence Interval Lower Bound and Upper Bound} ($\text{ci-low}$ and $\text{ci-high}$). In this work we consider a $95\%$ confidence interval.   Its lower bound demarcating the lower limit of the causal effect precision boundary is computed as
$\text{ci-low} = \hat{\tau}_{\text{ATE}} - 1.96 \times \text{SE}(\hat{\tau}_{\text{ATE}})$, while its upper bound is similarly computed as $ \text{ci-high} = \hat{\tau}_{\text{ATE}} + 1.96 \times \text{SE}(\hat{\tau}_{\text{ATE}})$.
    Here, the critical multiplier $1.96$ and the symmetric interval limits assume asymptotic normality of the estimator's sampling distribution to isolate the central $95\%$ probability mass.

\textbf{5. p-value}. Evaluates the probability of observing the calculated treatment effect (or a more extreme value) under the null hypothesis of zero treatment effect ($H_0: \tau_{\text{ATE}} = 0$). It is calculated using the survival function of the $t$-distribution:
    \begin{equation}
    p\text{-value} = 2 \times \mathbb{P}\left( T_{df} \geq |t_{\text{stat}}| \right)
    \end{equation}
    where $df$ represents the residual degrees of freedom from the regression model. Namely, we compute the probability that a random variable with a $t$-distribution will have a value larger than the absolute value of $t_{\text{stat}}$. Twice this computed value is the $p$-value. Furthermore, the underlying sensitivity profile tabulates upper-bound $p$-values utilizing strict numerical inequalities; values dropping below $0.0001$ are explicitly expressed as $<0.0001$ rather than truncated to zero, ensuring transparent presentation of the asymptotic bounds.

\textbf{6. Rosenbaum Sensitivity Threshold} ($\Gamma_{\text{crit}}$). This metric quantifies the robustness of the calculated causal inferences against potential hidden selection bias from unobserved confounding variables \citep{rosenbaum2007sensitivity}. Operating under the framework of robust Huber-Maritz M-testing, the sensitivity parameter $\Gamma \ge 1$ bounds the maximum factor by which an unobserved covariate could alter the odds of treatment assignment between two matched units with identical observed baseline vectors $\mathbf{X}_i$ \citep{rosenbaum2007sensitivity, rosenbaum2014weighted}. Operational analysis evaluates the sharp null hypothesis of zero treatment effect across a spectrum of $\Gamma$ allocations using the \texttt{senmv} implementation in the R \texttt{sensitivitymv} package \citep{rosenbaum2015two}, setting the robust M-statistic configuration to Huber's $\psi$-function (\texttt{method = "h"}) to achieve resistance to outlier-induced variance. The critical threshold $\Gamma_{\text{crit}}$ is defined as the exact interpolated value at which the upper bound of the one-sided $p$-value crosses the standard significance limit:
$\Gamma_{\text{crit}} = \sup \left\{ \Gamma : p_{\text{upper}}(\Gamma) \le 0.05 \right\}$.
A higher value of $\Gamma_{\text{crit}}$ indicates a highly robust causal conclusion, requiring a larger magnitude of unmeasured confounding to alter the substantive conclusions of the statistical framework, whereas $\Gamma_{\text{crit}} = 1.0$ indicates an inference highly vulnerable to immediate unobserved selection bias. Empirical reporting of this sensitivity metric adheres to strict formatting and precision protocols. The critical threshold $\Gamma_{\text{crit}}$ is presented rounded to four decimal digits to preserve calibration accuracy. For methods utilizing non-$1:1$ matching designs (e.g., CEM or Full Matching) where structural pair-difference frameworks are mathematically undefined, the value is reported as not applicable (\texttt{NA}). In cases where a $1:1$ matching configuration yields no statistical significance even under the base condition of no hidden bias, $\Gamma_{\text{crit}}$ is bounded at a baseline value of $1.0000$.

\subsection{Distances}
\label{sub:dist}
Next, we consider distance measurements for the $1:1$ matching methods.  Let $r$ be the total number of matched pairs, and $d_{M}(i, j) = \sqrt{(\mathbf{x}_{1i} - \mathbf{x}_{0j})^T \mathbf{\Sigma}^{-1} (\mathbf{x}_{1i} - \mathbf{x}_{0j})}$ denote the Mahalanobis distance between the treated unit $i$ and its matched control unit $j$. Because these distance-based metrics assume a $1:1$ matching configuration, they are assigned a value of \texttt{NA} to any non-$1:1$ matching method. 

The \textbf{Mean Mahalanobis Distance} ($\text{Mean-Mahal}$) measures the average Mahalanobis distance across all matched treated-control pairs across the $B$ strata.  It is computed as $\text{Mean-Mahal} = \frac{1}{r} \sum_{p=1}^{r} d_{M}(p)$
    where $d_{M}(p)$ represents the distance corresponding to the $p$-th matched pair.    

The \textbf{Median Mahalanobis Distance} ($\text{Median-Mahal}$) represents the $50\text{th}$ percentile of the pair-specific Mahalanobis distance, and it serves as a robust alternative to the mean that is unaffected by extreme covariate outliers.

The \textbf{Maximum Mahalanobis Distance} ($\text{Max-Mahal}$) identifies the single worst multivariate match in the entire matched dataset, allowing investigators to audit the upper bound of pair mismatching.  It is computed as $\text{Max-Mahal} = \max_{p \in \{1, \dots, r\}} d_{M}(p)$.  

The \textbf{Standard Deviation of Mahalanobis Distances} ($\text{SD-Mahal}$) quantifies the spread and structural consistency of the Mahalanobis distance matching pairs, indicating whether the matching quality is uniform or highly variable. It is evaluated using
$\text{SD-Mahal} = \sqrt{\frac{1}{r-1} \sum_{p=1}^{r} \left( d_{M}(p) - \text{Mean-Mahal} \right)^2}$.

\subsection{Covariate balancing}
The next collection of metrics attempts to identify the impact of different covariates on the distances of the matched pairs. Let $\text{SMD}_{j}$ denote the standardized mean difference for covariate $j$, defined as:
    \begin{equation}
    \label{eq:SMD}
    \text{SMD}_{j} = \frac{\bar{X}_{1j} - \bar{X}_{0j}}{\sqrt{\frac{s_{1j}^2 + s_{0j}^2}{2}}}
    \end{equation}
    where $\bar{X}_{1j}, \bar{X}_{0j}$ and $s_{1j}^2, s_{0j}^2$ represent the sample means and variances of covariate $j$ for the matched treated and control pools, respectively. 
    This definition of $\text{SMD}_j$ assumes that the solution is an output of a $1:1$ matching method.  However, if we consider the sample means and variances of the subsets of retained treated units and retained control units (respectively) we get the definition for other methods like CEM or Full Matching.  These values are evaluated for every covariate, and based on these values the following metrics are defined.
    
    The \textbf{Maximum Standardized Mean Difference} ($\text{Max-SMD}$) isolates the single highest absolute standardized bias among all evaluated baseline $m$ covariates, namely it is $\text{Max-SMD} = \max_{j \in \{1, \dots, m\}} \left| \text{SMD}_j \right|$.

    The \textbf{Mean Standardized Mean Difference} ($\text{Mean-SMD}$) computes the arithmetic average of the absolute standardized mean differences across all $m$ baseline covariates, summarizing the overall success of the joint balancing effort, $\text{Mean-SMD} = \frac{1}{m} \sum_{j=1}^{m} \left| \text{SMD}_j \right|$.

The \textbf{Maximum Kolmogorov-Smirnov Statistic} ($\text{Max-KS}$) identifies the largest absolute vertical distance between the empirical cumulative distribution functions (eCDFs) of the treated and control groups across all baseline variables, signaling full-distribution mismatching. Let $F_{1j}(x)$ and $F_{0j}(x)$ represent the eCDFs of covariate $j$ for the retained treated units and retained control units, respectively. The metric is computed as:
    \begin{equation*}
    \text{Max-KS} = \max_{j \in \{1, \dots, m\}} \sup_{x} \left| F_{1j}(x) - F_{0j}(x) \right| .
\end{equation*}

The last balancing metrics we consider \textbf{N-Imbalanced-01 / N-Imbalanced-02} are the number of imbalanced covariates.  These metrics count how many baseline covariates fail to meet the standard strict academic threshold ($\left|\text{SMD}\right| \leq 0.1$) or the more relaxed, conservative balance boundary ($\left|\text{SMD}\right| \leq 0.2$), respectively. Formally calculated using indicator functions:
    \begin{equation*}
    \texttt{N-Imbalanced-01} = \sum_{j=1}^{m} \mathbb{I}\left( \left| \text{SMD}_j \right| > 0.1 \right)
   \quad  \text{and} \quad  \texttt{N-Imbalanced-02} = \sum_{j=1}^{m} \mathbb{I}\left( \left| \text{SMD}_j \right| > 0.2 \right)
    \end{equation*}
    where $\mathbb{I}(\cdot)$ evaluates to $1$ if the condition is satisfied and $0$ otherwise. In the tables below, \texttt{N-Imbalanced-01} and \texttt{N-Imbalanced-02} are combined into a single column, \texttt{Nib}, formatted as \texttt{N-Imbalanced-01; N-Imbalanced-02}.

As stated previously, whenever the framework evaluates a Mahalanobis distance metric ($d_{M}$), the calculation seamlessly accommodates either the standard formulation or the robust ranked variation depending on the underlying parameter matrices supplied to the pipeline.

\section{OSIP Analysis on the Lindner Dataset}
\label{sec:osip_lindner}

In this section, we present an internal comparative analysis of the \texttt{OSIP} method, with respect to the \texttt{Lindner} dataset, using the performance metrics detailed in Section~\ref{sec:metrics}. Specifically, we evaluate the sensitivity and performance of \texttt{OSIP} under various configurations across its two operational phases: the dynamic programming stage (Section~\ref{DP_Step}) and the heuristic search optimization stage (Section~\ref{Heuristics_Step}). In Section~\ref{sub:comp_delta}, we evaluate the performance of our approach across three distinct values of $\delta$ to assess its impact on the outcomes of Stage~1 and Stage~2. For Stage~2, this evaluation incorporates two distance metrics: strict Mahalanobis distance and robust Mahalanobis distance.

To analyze the mechanics of our heuristic search, we evaluate the intermediate solutions generated at each step (transition). By comparing this sequence of step-by-step solutions under a fixed $\delta$ across both strict and robust modes, we gain insight into the marginal contribution of each local heuristic change. Section~\ref{lindner-heuris-strict} presents this step-by-step evaluation for the strict mode.

Finally, we analyze the sensitivity of the Stage~1 dynamic programming algorithm (Step~\ref{DP_Step}) to the power parameter $p$ of the $L_p$ Minkowski distance metric. While the default setup uses the $L_1$ norm ($p = 1.0$), we evaluate performance across $p \in \{1.0, 1.2, \dots, 2.4\}$ in steps of $0.2$. Our results show that the resulting partition is virtually invariant to the choice of $p$. For instance, under strict mode with $\delta = 0.2$, all values of $p$ yield identical sub-intervals:$$[0.0417, 0.2389], (0.2389, 0.4312], (0.4312, 0.6252], (0.6252, 0.6878]$$ with the local search ($LS$) heuristic consistently producing the optimal outcome.

Similarly, in robust mode for this value of $\delta$, the local search ($LS$) heuristic almost universally produces the optimal outcome, yielding the intervals:
 $$[0.0417,0.2315], (0.2315,0.4312],(0.4312,0.6252],(0.6252,0.6878] .$$  Here the case of $p=2.2$ results in a unique exception where the solution is slightly modified and the output of OSIP is determined by the outcome of SA which is $$[0.0417,0.1857],(0.1857,0.3850],(0.3850,0.4312],(0.4312,0.6252],(0.6252,0.6878] .$$  Our conclusion is that the heuristic search in stage~2 serves as a safeguard for cases in which the stage~1 solution is not optimal. These heuristics allow convergence into a relatively good matching regardless of the exact starting point used (i.e., for different outcomes of the first step).

 Based on this empirical invariance across different inputs and choices of $\delta$, we fix $p = 1.0$ as the standard metric for the dynamic programming stage. This choice minimizes numerical computation error while exploiting the fact that varying $p$ has no material impact on the statistical output of \texttt{OSIP}.
 
{\bf Notation.} Our proposed framework variants are denoted as \texttt{OS} (\texttt{OSIP} strict mode) and \texttt{OR} (\texttt{OSIP} robust mode). Where relevant, the hyperparameter setting is directly attached without separation (e.g., \texttt{OS15} or \texttt{OR20} to indicate strict mode with $\delta = 0.15$ and robust mode with $\delta = 0.20$, respectively).

\subsection{A comparison using various \texorpdfstring{$\delta$}{Delta} values}
\label{sub:comp_delta}
In this section, we evaluate the performance of \texttt{OSIP} across three distinct values of $\delta$ ($0.15$, $0.20$, and $0.25$) on the \texttt{Lindner} dataset to assess their impact on Stage~1 and Stage~2 outcomes. For Stage~2, this evaluation incorporates both the standard Mahalanobis distance and its robust variant. For step $1$, we use the $L_1$ norm, that is $p=1$ as explained above. Since \texttt{OSIP} performs $1:1$ matching, we match $297$ treated units to $297$ control units (selected from $684$ available control units) under the constraint that matching occurs strictly within each assigned bin, as discussed earlier. 

This condition is satisfied across all solution designs generated by \texttt{OSIP} on the \texttt{Lindner} dataset, as well as across all other datasets evaluated in this study. Notably, researchers may select candidate solutions based on criteria beyond global cost minimization. For example, one may prioritize a \textbf{balanced} solution that minimizes the total estimated cost subject to exact balance constraints, namely $\text{N-Imbalanced-01} = \text{N-Imbalanced-02} = 0$ (as defined in Section~\ref{sub:dist}). 

For the \texttt{Lindner} dataset, every candidate solution generated by \texttt{OSIP} inherently satisfies this balance constraint, making further filtering unnecessary. However, as demonstrated across other benchmark datasets in the Appendix, this behavior does not always hold universally. Alternatively, users could select a balanced solution that minimizes the mean Mahalanobis distance rather than the estimated total cost. Ultimately, because trade-offs between covariate balance and distance metrics depend on domain context, we recommend that practitioners evaluate both the statistical significance and distance metrics of the candidate solutions generated by \texttt{OSIP} to select the optimal design for their specific experiment.

Table~\ref{tab:delta_all_est_conv} presents the estimated costs across all values of $\delta$ ($0.15$, $0.20$, and $0.25$) and both \texttt{OSIP} modes for each heuristic run. When a heuristic improves upon the initial solution from step $1$, the table reports the updated estimated cost along with the number of improving steps required. Otherwise, only the initial solution's cost is reported.  For $\delta=0.15$, neither the strict nor the robust mode yields any modifications across all three heuristics.  When $\delta$ increases this behavior does not repeat across all heuristics. In particular, when $\delta$ is $0.2$ or $0.25$ the search trajectories across the three heuristics diverge. Furthermore, for this dataset the numbers of improving iterations carried out by these heuristics are extremely small, with a maximum number of $9$ transitions carried out by LS and SA for $\delta=0.2$ in the robust mode. We will further explore the behavior of these heuristics on this dataset further in Section \ref{lindner-heuris-strict}. NC means no change with respect the step $1$ (DP) solution (i.e., no better solution was found by the heuristic).  
\begin{table}[htbp]
  \centering
  \caption{OSIP: Estimated cost and conversion rates for all $\delta$ values}
  \label{tab:delta_all_est_conv}
  \begin{filecontents*}{chapter6_schemes/cost_steps.csv}
method,DP,LS ,ELS ,SA
OS15,125.4073,NC,NC,NC
OR15,109.5982,NC,NC,NC
OS20,125.0851,121.6308; 7,121.6308; 4,122.4996; 1
OR20,113.9524,106.0668; 9,106.0668; 7,107.2759; 9
OS25,121.1525,120.9987; 2,120.9987; 2,NC
OR25,105.3259,NC,NC,103.4806; 1
\end{filecontents*}
\pgfplotstabletypeset[
    col sep=comma,
    string type,
    every head row/.style={before row=\toprule, after row=\midrule},
    every last row/.style={after row=\bottomrule}
  ]{chapter6_schemes/cost_steps.csv}
\end{table}

We now perform the actual bin-wise matching and compute the overall evaluation metrics (see Section~\ref{sec:metrics} for their definition). For readability, these metrics are structured across three separate tables: The significance table (Table \ref{tab:delta_all_sig_lindner}) where we report on the measures described as statistical significance metrics, and the distances table (Table \ref{tab:delta_all_dist_lindner}) where we report on the metrics described as distance metrics.  Furthermore, we take a closer look into the balance information via the covariate level standardized mean differences, that is, the $SMD_j$ values as a function of the covariate $j$.  We refer to Table \ref{tab:love_plot_tbl} for this information.

\begin{table}[htbp]
  \centering
  \caption{OSIP: Significance table for all values of $\delta$}
  \label{tab:delta_all_sig_lindner}
  \begin{filecontents*}{chapter6_schemes/delta_all_sig.csv}
method,ate,se(ate),p-value,gamma-shift,ci-low,ci-high,t-stat
OS15,-0.1766,0.0391,$<0.0001$,1.3753,-0.2532,-0.0999,-4.5166
OR15,-0.1697,0.0388,$<0.0001$,1.4764,-0.2457,-0.0937,-4.3737
OS20,-0.1644,0.039,$<0.0001$,1.2518,-0.2408,-0.0881,-4.2154
OR20,-0.1665,0.0389,$<0.0001$,1.3199,-0.2427,-0.0903,-4.2802
OS25,-0.1746,0.0389,$<0.0001$,1.4344,-0.2508,-0.0984,-4.4884
OR25,-0.1786,0.0391,$<0.0001$,1.5056,-0.2552,-0.1021,-4.5678
\end{filecontents*}
\pgfplotstabletypeset[
    col sep=comma,
    string type,
    every head row/.style={before row=\toprule, after row=\midrule},
    every last row/.style={after row=\bottomrule}
  ]{chapter6_schemes/delta_all_sig.csv}
\end{table}

Across all six reported solutions the output is statistical significant. The $p$-value is very small and in particular smaller than $0.0001$, the absolute value of the $t$-stat is always larger than $4$, the average treatment affect, i.e., the ATE value, is around $-0.17$ with a consistent standard errors near $0.039$.

Because the primary significance measures are virtually identical across all configurations, the sensitivity parameter ($\Gamma$-shift) serves as our primary tie-breaking criterion. Under this metric, both the choice of $\delta$ and the distance mode (robust versus strict) yield meaningful differences. Specifically, for any fixed $\delta$, the robust mode consistently achieves superior sensitivity to unobserved confounding, as indicated by higher $\Gamma$-shift values. Likewise, $\delta = 0.25$ provides the strongest robustness against unobserved confounding across both strict and robust modes.

Next, consider the distance measures and the covariance balance of these six solutions.
It is interesting to note that in \texttt{OS15} the $\text{SMD}$ measure of \texttt{ejecfrac} is $0$. Such exact balance occurs also with respect to \texttt{acutemi} for almost all configurations. When we consider the worse among the SMD values of the numerical covariates (\texttt{ejecfrac} and \texttt{ves1proc}), the better solution is for $\delta=0.25$ in the strict mode. For this value of $\delta$ in the strict mode the SMD value of \texttt{ves1proc} is minimized (among our $6$ runs) and it reached the value of $0.0355$ attaining the Max-SMD for this value of $\delta$ in the strict mode. 
It appears that the covariate balance achieved for \texttt{ves1proc} is noticeably worse than that for \texttt{ejecfrac}. Indeed, across all tested parameter configurations, the absolute standardized mean difference ($\text{SMD}$) for \texttt{ejecfrac} remains significantly lower than that of \texttt{ves1proc}. Setting aside the structural distinction between numerical and binary covariates, \texttt{OSIP} achieves its most pronounced balancing power on the \texttt{acutemi} covariate, where the SMD reaches exactly $0$ in five of the six solutions. 

Overall balance across the dataset is characterized by the Max-SMD and Mean-SMD metrics. Under the Max-SMD criterion, the optimal configuration is obtained when $\delta = 0.25$ in strict mode, whereas the $\delta = 0.2$ strict-mode solution performs best under the Mean-SMD metric.

Regarding the distribution of Mahalanobis distances, the strict mode consistently yields matchings containing a small number of paired outliers with high distance values. This is evidenced by lower median Mahalanobis distances compared to their robust counterparts, alongside higher mean and maximum Mahalanobis distance,   and the corresponding standard deviations.

\begin{table}[htbp]
  \centering
  \caption{OSIP: Distances table for all values of $\delta$}
  \label{tab:delta_all_dist_lindner}
  \begin{filecontents*}{chapter6_schemes/delta_all_dist.csv}
method,Mean-Mahal,Median-Mahal,Max-Mahal,SD-Mahal,Max-SMD,Mean-SMD,Max-KS
OS15,0.5715,0.2706,5.4913,0.8849,0.0576,0.0229,0.0539
OR15,0.4909,0.3348,3.8048,0.6932,0.0736,0.0248,0.064
OS20,0.5685,0.2473,4.4866,0.8385,0.0404,0.0223,0.0606
OR20,0.4795,0.3172,3.5798,0.6549,0.0539,0.0285,0.0572
OS25,0.5735,0.2706,4.4866,0.8674,0.0355,0.024,0.0572
OR25,0.4665,0.3245,3.2477,0.6314,0.0503,0.0242,0.0505
\end{filecontents*}
\pgfplotstabletypeset[
    col sep=comma,
    string type,
    every head row/.style={before row=\toprule, after row=\midrule},
    every last row/.style={after row=\bottomrule}
  ]{chapter6_schemes/delta_all_dist.csv}
\end{table}

\begin{table}[htbp]
  \centering
  \caption{OSIP Balance table. The numbers are the values of $SMD_j$ for the corresponding solution and covariate $j$.}
  \label{tab:love_plot_tbl}
  \begin{filecontents*}{chapter6_schemes/love_plot_table_combined.csv}
covariate,OS15,OR15,OS20,OR20,OS25,OR25
stent,-0.0337,-0.0303,-0.0303,-0.0269,-0.0303,-0.0202
height,0.0218,0.001,0.0289,-0.0063,0.0323,0.0089
female,0.0168,0.0269,0.0168,0.0303,0.0168,0.0269
diabetic,0.0303,0.0404,0.0337,0.0471,0.0337,0.0471
acutemi,0,0,0,0,0.0067,0
ejecfrac,0,-0.0011,0.0054,0.0064,0.0123,0.0161
ves1proc,0.0576,0.0736,0.0576,0.0812,0.0355,0.0503
\end{filecontents*}
\pgfplotstabletypeset[
    col sep=comma,
    string type,
    every head row/.style={before row=\toprule, after row=\midrule},
    every last row/.style={after row=\bottomrule}
  ]{chapter6_schemes/love_plot_table_combined.csv}
\end{table}

Given the importance of maximizing the covariate balance, we choose 
for the general comparison against other methods in Section~\ref{OSIP_vs_all_lindner}, to consider the strict solution for the \texttt{Lindner} dataset.  For the strict mode, the values of $\delta=0.2$ and $\delta=0.25$ offer a better tradeoff between the different measures than the tradeoff obtained for $\delta=0.15$.  However, before we make this general comparison we take a closer look on the behavior of the heuristics in stage~2 of \texttt{OSIP}. 

\subsection{Heuristics analysis step by step, strict mode} 
\label{lindner-heuris-strict}

In this section, we present the comprehensive results for $\delta = 0.20$ under the strict mode of \texttt{OSIP}, evaluating the Local Search ($LS$), Enhanced Local Search ($ELS$), and Simulated Annealing ($SA$) heuristics. For each heuristic, performance is reported across two companion tables: one reporting statistical significance and treatment effect estimation metrics, and the other detailing covariate balance and distance metrics (including SMD, KS statistics, and Mahalanobis distance distributions).

Our goal is to demonstrate that each local search operator (move) consistently yields meaningful improvements across established statistical evaluation metrics—even though \texttt{OSIP}'s optimization objective operates directly on the surrogate estimated cost. To evaluate this alignment, we inspect the intermediate solutions obtained by temporarily halting the heuristic search at the end of each step (transition). For each intermediate partition, we perform $1:1$ within-bin matching and evaluate the candidate design across all standard metrics before resuming the heuristic search based on the surrogate cost. Note that these intermediate evaluation checkpoints are used solely for diagnostic purposes in this analysis and are not part of the standard \texttt{OSIP} runtime procedure. Below, we focus on $\delta = 0.20$ in strict mode; the corresponding diagnostic analysis for robust mode is provided in Appendix~\ref{OR20-lindner}. Across all tables, ``Step~0'' denotes the initial partition produced at the conclusion of Stage~1 (dynamic programming). Any other row (labeled ``Step~$i$'') denotes the solution after $i$ iterations of the relevant heuristic.

Consider first the significance metrics. Now, we also specify the estimated cost of each solution with the purpose of showing the monotonicity of the estimated cost. This auxiliary goal function has no direct implication on the significance metric of the corresponding matching. We show the significance steps of LS (Table~\ref{tab:delta0.2_LS_steps_sig}), ELS (Table~\ref{tab:delta0.2_ELS_steps_sig}), and SA (Table~\ref{tab:delta0.2_SA_steps_sig}).
 Although the $p$-value remain uniformly significant throughout these heuristics, we can observe the behavior of the other metrics when a single step is carried out. For example, the magnitude of the $t$-stat is sometimes increased and sometimes decreased. Such increases occurs for example in step 4 in LS or in the last step of ELS.  On the other hand the magnitude of $t$-stat sometimes decreased like in step 2 of LS or in step 2 of ELS or in the unique step of SA.  The $\Gamma$-shift behavior during ELS also starts by decreasing after the first step and in the last iteration it increases to the maximum value obtained by \texttt{OSIP} for this value of $\delta$ and strict mode.  This type of non-monotonicity of the significance metrics occurs also for the robust mode, as can be seen in the appendix. 

\begin{table}[htbp]
  \centering
  \caption{LS significance metrics}
  \label{tab:delta0.2_LS_steps_sig}
  \begin{filecontents*}{chapter6_schemes/LS_delta0.2_sig.csv}
step,ate,se(ate),p-value,gamma-shift,ci-low,ci-high,t-stat,est-cost
0,-0.1774,0.0391,$< 0.0001$,1.3986,-0.2541,-0.1008,-4.5371,125.0851
1,-0.1673,0.0387,$< 0.0001$,1.2982,-0.2432,-0.0914,-4.3230,123.9023
2,-0.1678,0.0389,$< 0.0001$,1.3073,-0.2440,-0.0916,-4.3136,123.7556
3,-0.1680,0.0388,$< 0.0001$,1.3134,-0.2441,-0.0919,-4.3299,123.5722
4,-0.1714,0.0390,$< 0.0001$,1.3430,-0.2479,-0.0949,-4.3949,122.4996
5,-0.1677,0.0389,$< 0.0001$,1.2335,-0.2440,-0.0914,-4.3111,122.3761
6,-0.1697,0.0389,$< 0.0001$,1.2844,-0.2459,-0.0935,-4.3625,122.1213
7,-0.1644,0.0390,$< 0.0001$,1.2518,-0.2408,-0.0881,-4.2154,121.6308
\end{filecontents*}
\pgfplotstabletypeset[
    col sep=comma,
    string type,
    every head row/.style={before row=\toprule, after row=\midrule},
    every last row/.style={after row=\bottomrule}
  ]{chapter6_schemes/LS_delta0.2_sig.csv}
\end{table}

\begin{table}[htbp]
  \centering
  \caption{ELS significance metrics}
  \label{tab:delta0.2_ELS_steps_sig}
  \begin{filecontents*}{chapter6_schemes/ELS_delta0.2_sig.csv}
step,ate,se(ate),p-value,gamma-shift,ci-low,ci-high,t-stat,est-cost
0,-0.1774,0.0391,$< 0.0001$,1.3986,-0.2541,-0.1008,-4.5371,125.0851
1,-0.1745,0.0391,$< 0.0001$,1.4488,-0.2511,-0.0979,-4.4629,122.7543
2,-0.1675,0.0388,$< 0.0001$,1.3025,-0.2436,-0.0913,-4.3170,122.4996
3,-0.1678,0.0387,$< 0.0001$,1.2814,-0.2437,-0.0919,-4.3359,122.0070
4,-0.1745,0.0390,$< 0.0001$,1.4926,-0.2511,-0.0980,-4.4744,121.6308
\end{filecontents*}
\pgfplotstabletypeset[
    col sep=comma,
    string type,
    every head row/.style={before row=\toprule, after row=\midrule},
    every last row/.style={after row=\bottomrule}
  ]{chapter6_schemes/ELS_delta0.2_sig.csv}
\end{table}

\begin{table}[htbp]
  \centering
  \caption{SA significance metrics}
  \label{tab:delta0.2_SA_steps_sig}
  \begin{filecontents*}{chapter6_schemes/SA_delta0.2_sig.csv}
step,ate,se(ate),p-value,gamma-shift,ci-low,ci-high,t-stat,est-cost
0,-0.1774,0.0391,$< 0.0001$,1.3986,-0.2541,-0.1008,-4.5371,125.0851
1,-0.1675,0.0388,$< 0.0001$,1.3025,-0.2436,-0.0913,-4.317,122.4996
\end{filecontents*}
\pgfplotstabletypeset[
    col sep=comma,
    string type,
    every head row/.style={before row=\toprule, after row=\midrule},
    every last row/.style={after row=\bottomrule}
  ]{chapter6_schemes/SA_delta0.2_sig.csv}
\end{table}

We next examine the distance and covariate balance metrics for LS (Table~\ref{tab:delta0.2_LS_steps_dist}), ELS (Table~\ref{tab:delta0.2_ELS_steps_dist}), and SA (Table~\ref{tab:delta0.2_SA_steps_dist}). Although these metrics generally fluctuate across search iterations, the Mahalanobis distance parameters under LS, specifically Mean-Mahal, Median-Mahal, and SD-Mahal, exhibit an almost monotonic downward trend, ultimately attaining their minimum values at the final step. In contrast, ELS does not exhibit this steady convergence. For ELS, the intermediate design at Step~3 outperforms the final output (Step~4) across all distance and balance metrics except for the surrogate estimated cost. Nevertheless, this intermediate solution remains inferior to the final step in terms of the treatment effect sensitivity ($\Gamma$-shift).

In this work, we optimize the surrogate estimated cost during Stage~2 to significantly accelerate heuristic iterations. Evaluating the exact matching at every single heuristic step is computationally prohibitive, as computing the full bin-wise matching for each intermediate partition quickly becomes the primary execution bottleneck. While evaluating intermediate checkpoints yields a trajectory of candidate solutions, doing so requires substantial runtime. We argue that this extra computational budget is far better spent exploring additional global design parameters (e.g., evaluating a broader range of $\delta$ values). For instance, although the step-by-step diagnostic analysis yields a maximum ATE magnitude of $|-0.1774|$, a practitioner seeking a stronger effect estimate would achieve a superior outcome simply by configuring \texttt{OSIP} under $\delta = 0.25$ in robust mode \texttt{OR25} (See table \ref{tab:delta_all_sig_lindner}).

\begin{table}[htbp]
  \centering
  \caption{LS distances metrics}
  \label{tab:delta0.2_LS_steps_dist}
  \begin{filecontents*}{chapter6_schemes/LS_delta0.2_dist.csv}
step,Max-SMD,Mean-SMD,Max-KS,Mean-Mahal,Median-Mahal,Max-Mahal,SD-Mahal
0,0.0370,0.0262,0.0673,0.5866,0.2706,4.8626,0.8751
1,0.0370,0.0251,0.0707,0.5807,0.2632,4.4866,0.8640
2,0.0404,0.0226,0.0640,0.5760,0.2632,4.4866,0.8595
3,0.0404,0.0218,0.0572,0.5757,0.2632,4.4866,0.8596
4,0.0370,0.0227,0.0606,0.5723,0.2632,4.4866,0.8538
5,0.0370,0.0226,0.0640,0.5728,0.2632,4.4866,0.8539
6,0.0404,0.0219,0.0640,0.5762,0.2473,4.4866,0.8621
7,0.0404,0.0223,0.0606,0.5685,0.2473,4.4866,0.8385
\end{filecontents*}
\pgfplotstabletypeset[
    col sep=comma,
    string type,
    every head row/.style={before row=\toprule, after row=\midrule},
    every last row/.style={after row=\bottomrule}
  ]{chapter6_schemes/LS_delta0.2_dist.csv}
\end{table}

\begin{table}[htbp]
  \centering
  \caption{ELS distances metrics}
  \label{tab:delta0.2_ELS_steps_dist}
  \begin{filecontents*}{chapter6_schemes/ELS_delta0.2_dist.csv}
step,Max-SMD,Mean-SMD,Max-KS,Mean-Mahal,Median-Mahal,Max-Mahal,SD-Mahal
0,0.0370,0.0262,0.0673,0.5866,0.2706,4.8626,0.8751
1,0.0370,0.0222,0.0606,0.5721,0.2706,4.4866,0.8532
2,0.0370,0.0227,0.0606,0.5723,0.2632,4.4866,0.8538
3,0.0370,0.0208,0.0606,0.5633,0.2473,4.2636,0.8182
4,0.0404,0.0223,0.0606,0.5685,0.2473,4.4866,0.8385
\end{filecontents*}
\pgfplotstabletypeset[
    col sep=comma,
    string type,
    every head row/.style={before row=\toprule, after row=\midrule},
    every last row/.style={after row=\bottomrule}
  ]{chapter6_schemes/ELS_delta0.2_dist.csv}
\end{table}

\begin{table}[htbp]
  \centering
  \caption{SA distances metrics}
  \label{tab:delta0.2_SA_steps_dist}
  \begin{filecontents*}{chapter6_schemes/SA_delta0.2_dist.csv}
step,Max-SMD,Mean-SMD,Max-KS,Mean-Mahal,Median-Mahal,Max-Mahal,SD-Mahal
0,0.037,0.0262,0.0673,0.5866,0.2706,4.8626,0.8751
1,0.037,0.0227,0.0606,0.5723,0.2632,4.4866,0.8538
\end{filecontents*}
\pgfplotstabletypeset[
    col sep=comma,
    string type,
    every head row/.style={before row=\toprule, after row=\midrule},
    every last row/.style={after row=\bottomrule}
  ]{chapter6_schemes/SA_delta0.2_dist.csv}
\end{table}


\section{OSIP vs the Other Heuristics}
\label{OSIP_vs_all_lindner}
The goal of this section is to present a comprehensive comparison between the solutions generated by \texttt{OSIP} on the \texttt{Lindner} dataset and those produced by state-of-the-art methods. The benchmark methods included in this comparison are described in Section~\ref{Intro}.

\paragraph{Methodology and Naming Conventions}
To maintain conciseness across all empirical tables, matching algorithms are designated using standardized short keys. Unmatched baseline samples are denoted as \texttt{UM}, while standard benchmark algorithms are abbreviated as \texttt{Opt} (Optimal 1:1 Matching), \texttt{Card} (Cardinality 1:1 Matching), \texttt{CEM} (Coarsened Exact Matching), \texttt{FM} (Full Matching), \texttt{Gen} (Genetic 1:1 Matching), \texttt{Q} (Propensity Score Quintiles 1:1 Matching), and \texttt{RQ} (Refined Propensity Score Quintiles 1:1 Matching). 

Table~\ref{tab:all_sig} presents the statistical significance metrics.  We also report  the sample retention counts across all evaluated matching algorithms for the \texttt{Lindner} dataset.  The two outputs of \texttt{OSIP} (\texttt{OS20} and \texttt{OS25}) retain the complete set of treated units ($N_T = 297, N_C = 297$).  This is a straightforward corollary of the feasibility of the partition of the common support interval into sub-intervals satisfying the fine-balance constraint.  Thus, \texttt{OSIP} solutions avoid the substantial sample loss incurred by non-1:1 or trimmed benchmarks such as \texttt{CEM} ($N_T = 212$), \texttt{Q} ($N_T = 282$), and \texttt{RQ} ($N_T = 227$).  
Next, consider the significance metrics.  \texttt{OS25} yields the largest treatment effect estimate ($\text{ATE} = -0.1782$ with a moderate standard error of $0.0391$, achieving strong statistical significance ($t = -4.5575, p < 0.0001$).  \texttt{OS20} estimates a moderate effect ($\text{ATE} = -0.1644$ with a standard error of $0.0390$, and $t = -4.2154$), outperforming \texttt{Card} ($\text{ATE} = -0.1398$) and \texttt{Gen} ($\text{ATE} = -0.1573$) in treatment effect sensitivity. Furthermore, increasing the threshold parameter $\delta$ from $0.20$ to $0.25$ markedly improves sensitivity to unobserved confounding: the $\Gamma$-shift increases from $1.2518$ for \texttt{OS20} to $1.4999$ for \texttt{OS25}, closely matching the highest robustness observed in \texttt{Opt} ($\Gamma = 1.5145$) while far exceeding \texttt{Card} ($\Gamma = 1.1425$) and \texttt{RQ} ($\Gamma = 1.1044$).  While the solution of \texttt{FM} has the highest absolute value of $t$-stat ($-4.8946$) it has significantly smaller treatment effect estimate of $-0.1532$, and thus we find \texttt{OS25} is statistically more significant than \texttt{FM}.  Similar comparison with \texttt{UM} reveals that  \texttt{OS25} is even more significant than the unmatched baseline and it is our top priority solution with respect to statistical significance. 

\begin{table}[htbp]
  \centering
  \small
  \caption{Significance table for all models, Lindner}
  \label{tab:all_sig}
  \begin{filecontents*}{chapter7_schemes/consolidated_results_lindner_delta_0.25_power_1_sig_transposed.csv}
method,treated,control,ate,se(ate),p-value,gamma-shift,ci-low,ci-high,t-stat
UM,297,684,-0.1492,0.0321,$< 0.0001$,NA,-0.212,-0.0863,-4.648
OS20,297,297,-0.1644,0.039,$< 0.0001$,1.2518,-0.2408,-0.0881,-4.2154
OS25,297,297,-0.1782,0.0391,$< 0.0001$,1.4999,-0.2548,-0.1015,-4.5575
Opt ,297,297,-0.1723,0.0392,$< 0.0001$,1.5145,-0.2492,-0.0955,-4.3954
Card,297,297,-0.1398,0.0391,0.0004,1.1425,-0.2164,-0.0631,-3.5754
CEM,212,389,-0.1576,0.0404,0.0001,NA,-0.2367,-0.0785,-3.901
FM,297,684,-0.1532,0.0313,$< 0.0001$,NA,-0.2146,-0.0918,-4.8946
Gen ,297,297,-0.1573,0.0393,$< 0.0001$,1.4419,-0.2344,-0.0802,-4.0025
Q ,282,282,-0.1746,0.0398,$< 0.0001$,1.5059,-0.2526,-0.0966,-4.3869
RQ,227,227,-0.1458,0.0448,0.0012,1.1044,-0.2336,-0.0579,-3.2545
\end{filecontents*}
\pgfplotstabletypeset[
    col sep=comma,
    string type,
    every head row/.style={before row=\toprule, after row=\midrule},
    every last row/.style={after row=\bottomrule}
  ]{chapter7_schemes/consolidated_results_lindner_delta_0.25_power_1_sig_transposed.csv}
\end{table}

Table~\ref{tab:all_dist_lindner} presents the pairwise Mahalanobis distance metrics and overall balance summary statistics across all benchmark models for the \texttt{Lindner} dataset. Both \texttt{OSIP} configurations (\texttt{OS20} and \texttt{OS25}) achieve complete covariate balance ($\text{Nib} = 0;0$), maintaining maximum standardized mean differences well below the conventional $0.10$ threshold. In contrast, \texttt{Card} leaves two imbalanced covariates ($\text{Max-SMD} = 0.1468; \text{Nib} = 2;0$) and \texttt{FM} leaves one ($\text{Max-SMD} = 0.1096; \text{Nib} = 1;0$).

Focusing on methods that retain the full matched sample (excluding \texttt{CEM}, \texttt{Q}, and \texttt{RQ}, which suffer from sample attrition), \texttt{OS20} achieves the lowest \text{Max-SMD} ($0.0404$), followed by \texttt{Opt} ($0.0539$) and \texttt{OS25} ($0.0576$). Both \texttt{OSIP} variants substantially outperform \texttt{Gen} ($\text{Max-SMD} = 0.0925$). Under the Mean-SMD metric, \texttt{OS20} ($0.0223$) and \texttt{OS25} ($0.0252$) outperform all full-sample competitors. Finally, for the Max-KS criterion, \texttt{OS20} achieves the single best overall value ($0.0606$), while \texttt{OS25} ties with \texttt{Gen} for the second-best performance ($0.0640$).

The Mahalanobis distance measures are reported only for the $1:1$ matching methods.  Here, again we exclude from the comparison the solutions suffering sample loss, and analyze the resulting subset of solutions.  Here, \texttt{Opt} has the smallest Mean-Mahal value of $0.5075$ while \texttt{OS25} and \texttt{OS20} are second and third best options. The smallest Median-Mahal value is $0.2473$ obtained by \texttt{OS20} while the solutions of \texttt{OS25} and \texttt{Opt} have a common value of $0.2706$ ranked second.  For both metrics the solutions of \texttt{Card} and \texttt{Gen} are significantly worse than the ones achieved by these three solutions (\texttt{OS20}, \texttt{OS25}, and \texttt{Opt}).

\begin{table}[htbp]
  \centering
  \small
   \caption{Distances table for all models, Lindner}
  \label{tab:all_dist_lindner}
  \begin{filecontents*}{chapter7_schemes/consolidated_results_lindner_all_power_1_dist.csv}
method,Mean-Mahal,Median-Mahal,Max-Mahal,SD-Mahal,Max-SMD,Mean-SMD,Max-KS,Nib
UM,NA,NA,NA,NA,0.3817,0.1261,0.1749,4;1
OS20,0.5685,0.2473,4.4866,0.8385,0.0404,0.0223,0.0606,0;0
OS25,0.5598,0.2706,5.4913,0.8552,0.0576,0.0252,0.064,0;0
Opt ,0.5075,0.2706,3.8415,0.7055,0.0539,0.0264,0.0707,0;0
Card,1.0649,0.6107,5.4297,1.1288,0.1468,0.0722,0.0943,2;0
CEM,NA,NA,NA,NA,0.0494,0.0114,0.0827,0;0
FM,NA,NA,NA,NA,0.1096,0.0375,0.0954,1;0
Gen ,0.6725,0.3706,6.7185,0.9441,0.0925,0.0360,0.0640,0;0
Q ,0.5233,0.2473,4.0080,0.8061,0.0390,0.0193,0.0674,0;0
RQ,0.4391,0.2650,4.2275,0.6624,0.0213,0.0093,0.0705,0;0
\end{filecontents*}
\pgfplotstabletypeset[
    col sep=comma,
    string type,
    every head row/.style={before row=\toprule, after row=\midrule},
    every last row/.style={after row=\bottomrule}
  ]{chapter7_schemes/consolidated_results_lindner_all_power_1_dist.csv}
\end{table}

We next examine the covariate-level balance achieved by the different matching methods. Table~\ref{tab:bal_lindner_all} reports the standardized mean differences ($\text{SMD}$) across all seven covariates for the \texttt{Lindner} dataset. The baseline unmatched sample (\texttt{UM}) exhibits severe initial imbalance, primarily driven by \texttt{ves1proc} ($\text{SMD} = -0.3817$), \texttt{ejecfrac} ($\text{SMD} = 0.1642$), \texttt{stent} ($\text{SMD} = -0.1174$), and \texttt{acutemi} ($\text{SMD} = -0.1061$). 

Benchmark methods achieve mixed success in mitigating these initial disparities. \texttt{FM} corrects the imbalance across six covariates but leaves \texttt{ejecfrac} imbalanced ($\text{SMD} = 0.1096$). \texttt{Card} fails to satisfy the standard $0.10$ threshold on multiple covariates, including \texttt{height} ($-0.1468$) and \texttt{ves1proc} ($-0.1286$). Although \texttt{CEM} eliminates imbalance across all covariates, it suffers from severe sample loss. 

Conversely, both \texttt{OS20} and \texttt{OS25} reduce all seven covariate imbalances well below the $0.10$ threshold, with no individual covariate $j$ exceeding $|\text{SMD}_j| = 0.0576$. Both variants achieve exact balance ($\text{SMD} = 0$) on \texttt{acutemi}, matching the performance of \texttt{Opt}, \texttt{CEM}, and \texttt{RQ}. Specifically, \texttt{OS20} provides tighter control on \texttt{stent} ($-0.0269$) and \texttt{ves1proc} ($0.0214$), whereas \texttt{OS25} achieves better balance on \texttt{female} ($0.0202$), \texttt{diabetic} ($0.0269$), and \texttt{ejecfrac} ($0.0140$). On the problematic \texttt{ejecfrac} covariate, both \texttt{OSIP} variants ($\text{SMD} \le 0.0267$) outperform all other full-sample alternatives. Overall, both \texttt{OSIP} models deliver fine-grained, global covariate balance comparable to trimmed methods like \texttt{RQ} and \texttt{CEM} without discarding matched pairs or compromising full sample retention.

\begin{table}[htbp]
  \centering
  \small
   \caption{Balance table for all models, Lindner}
  \label{tab:bal_lindner_all}
  \begin{filecontents*}{chapter7_schemes/love_plot_table_delta_all.csv}
method,stent,height,female,diabetic,acutemi,ejecfrac,ves1proc
UM,-0.1174,-0.0027,0.0505,0.0599,-0.1061,0.1642,-0.3817
OS20,-0.0269,0.0136,0.0269,0.0404,0,0.0267,0.0214
OS25,-0.0404,0.0169,0.0202,0.0269,0,0.014,0.0576
Opt ,-0.0539,0.0455,0.0034,0.0034,0,0.0507,-0.0283
Card,-0.0707,-0.1468,0.0606,-0.037,-0.0337,0.028,-0.1286
CEM,0,0.0304,0,0,0,-0.0494,0
FM,-0.0152,0.078,0.0057,0.0098,-0.0084,0.1096,-0.0356
Gen ,-0.0135,0.0492,0.0135,0.0168,-0.0202,0.0463,-0.0925
Q ,-0.039,-0.0082,0.0035,0.0106,0.0071,0.0377,0.0292
RQ,-0.0088,0.0151,-0.0088,0,0,0.0154,0.0086
\end{filecontents*}
\pgfplotstabletypeset[
    col sep=comma,
    string type,
    every head row/.style={before row=\toprule, after row=\midrule},
    every last row/.style={after row=\bottomrule}
  ]{chapter7_schemes/love_plot_table_delta_all.csv}
\end{table}

Comparing the two \texttt{OSIP} configurations highlights a clear structural trade-off. \texttt{OS20} prioritizes covariate closeness and global balance, yielding lower overall balance error ($\text{Max-SMD} = 0.0404$, $\text{Mean-SMD} = 0.0223$) and tighter spatial matching ($\text{Median-Mahal} = 0.2473$, $\text{Max-Mahal} = 4.4866$). Conversely, \texttt{OS25} exhibits slightly higher distances ($\text{Mean-Mahal} = 0.5598$, $\text{Max-Mahal} = 5.4913$), trading minor spatial proximity to obtain the enhanced treatment effect sensitivity and $\Gamma$-shift robustness reported in Table~\ref{tab:all_sig}. Overall, both \texttt{OSIP} models, together with \texttt{Opt}, emerge as the clear winners regarding the distance measures and the covariate balance, outperforming the other reported methods.

\section{Conclusions}
In this work we present a novel matching method in observational study design and we demonstrate it on several standard benchmarks.  Our evaluation shows that our method provides an effective framework for causal design, competing favorably with state-of-the-art matching methodologies in the literature.  In the main part of the paper we tested and demonstrated the new method on the \texttt{Lindner} dataset, while in the appendix we consider several additional standard benchmark datasets. Below, we briefly summarize the results for these additional datasets.
\begin{itemize}

\item The {\bf NSW-Mixtape} dataset \citep{cunningham2021causal, lalonde1986evaluating, dehejia1999causal}.
Downloaded from R's \texttt{causaldata} package, it is derived from the National Supported Work (NSW) demonstration job-training program experiment, where treated participants were guaranteed employment for 9–18 months. The dataset comprises $445$ units and $11$ covariates. Trimming to common support $[0.2379, 0.6433]$ yields $181$ treated units and $249$ control units. We evaluated \texttt{OSIP} across both robust and strict modes for $\delta \in \{0.15, 0.2, 0.25\}$. Among these six \texttt{OSIP} runs, configurations with $\delta = 0.25$ perform best. For $\delta = 0.25$, the robust mode is preferred when Mahalanobis distance measures are paramount, whereas the strict mode is superior when prioritizing statistical significance. Incorporating standard benchmark methods into the comparison reveals an elite candidate set of three top-performing designs: \texttt{Gen} is preferred if achieving the most balanced solution is required, \texttt{OR25} dominates when prioritizing Mahalanobis distance, and \texttt{OS25} leads when prioritizing statistical significance. Thus, two of the three elite candidate solutions are produced by \texttt{OSIP}.

\item The {\bf IHDP} dataset originates from the Infant Health and Development Program. Downloaded from R's \texttt{bartcs} package, it derives from a randomized experiment conducted from $1985$ to $1988$ that evaluated the impact of home visits on infant cognitive test scores. First introduced to the causal inference literature by \citep{hill2011bayesian}, it was subsequently adopted as a standard benchmark by \citep{louizos2017causal} and \citep{shalit2017estimating}. The dataset comprises $134$ treated units, $556$ control units, and a larger covariate dimension ($25$ covariates). Trimming to common support $[0.0335, 0.5749]$ retains the matched sample. We evaluated \texttt{OSIP} across both strict and robust modes for $\delta \in \{0.1, 0.15, 0.2\}$. All six \texttt{OSIP} solutions achieved high statistical significance with minimal variation across significance metrics. However, when prioritizing Mahalanobis distance and overall covariate balance, \texttt{OR20} emerges as the clear best candidate design. The selection of \texttt{OR20} as the primary candidate solution remains robust even when comparing against all state-of-the-art alternative methods.

\item Our final dataset is the {\bf NHEFS} dataset (National Health and Nutrition Examination Survey) 
\newline
\citep{hernan2020causal, cdc_nhefs}. Specifically, we analyze its complete form, 
\texttt{nhefs\_complete}, provided by R's \texttt{causaldata} package, which originates from the Epidemiologic Follow-up Study initiated jointly by the National Center for Health Statistics (NCHS) and the National Institute on Aging. The dataset is widely utilized as a benchmark in causal inference literature \citep{hernan2020causal} to evaluate the effect of smoking cessation on weight gain. The \texttt{nhefs\_complete} subset contains $1556$ observations across $67$ variables. Trimming to the common support region $[0.0753, 0.6347]$ yields $402$ treated units and $1153$ control units. Here, we evaluated \texttt{OSIP} for $\delta \in \{0.15, 0.175, 0.2\}$ across both strict and robust modes. Because significance metrics are comparable across all \texttt{OSIP} configurations, we rely on distance and covariate balance measures as our tie-breaking criteria: \texttt{OR175} achieves the most refined covariate balance, whereas \texttt{OR20} minimizes Mahalanobis distance metrics. This candidate solution set remains optimal even when evaluating against state-of-the-art alternative methods from the literature.
\end{itemize}

We conclude that the solutions obtained by \texttt{OSIP} are highly competitive and, in many cases, provide significantly better designs than those existing in the scientific literature prior to our work.  Below, we provide practical guidelines for researchers and practitioners looking to apply our method and open-source implementation to their own datasets.

\begin{itemize}
    \item 
    \textbf{Selecting $\delta$:} Choose $\delta$ in accordance with the theoretical bounds established in Section~\ref{delta_choice}. Specifically, as established in Lemma 1, $\delta$ must satisfy $\delta \ge \max\{\delta_{\text{pair}}, \delta_{\text{cons}}\}$; selecting any $\delta < \max\{\delta_{\text{pair}}, \delta_{\text{cons}}\}$ guarantees that no valid post-stratification partition exists, rendering the Stage 1 optimization problem in \texttt{OSIP} infeasible. Across all real-world benchmark datasets tested, setting baseline defaults to $\delta = \frac{p_k - p_0}{4}$ or $\delta = \frac{p_k - p_0}{3}$ serves as an effective starting point, where $p_0$ and $p_k$ represent the lower ($L_S$) and upper ($R_S$) boundaries of the trimmed support interval. Once an initial feasible $\delta$ is identified, practitioners should evaluate a small candidate set around this value. In our experiments, a candidate set of three $\delta$ values provided sufficient flexibility to explore the balance-distance trade-off. While values close to the theoretical minimum maximize fine-grained balance control, setting $\delta$ excessively large relaxes interval bounds too much, reducing \texttt{OSIP}'s ability to leverage localized propensity score proximity for covariate alignment.

    \item \textbf{Setting $K$ 
    :} Given a value of $\delta$, the theoretical upper bound on the number of active strata ($K$) can be directly computed analytically via Equation~\eqref{eq:delta_k_bound} and is automatically evaluated by our implementation. For most datasets, using this theoretical upper bound guarantees global feasibility, while maintaining efficient computation. For larger datasets where solver runtimes may increase, practitioners can test smaller values of $K$ to reduce the search space and accelerate convergence without significantly compromising solution quality.

    \item \textbf{Log Files and Output Tracking:} Detailed execution logs and summary statistics are automatically saved to text files formatted as \texttt{main\_<dataset\_name>\_output\_delta\_<delta\_value>\_power\_1.txt} (for instance, \texttt{main\_lindner\_output\_delta\_0.20\_power\_1.txt}). We recommend evaluating paired solutions across both strict and robust modes for every $\delta$ in the candidate set. Once this family of \texttt{OSIP} outputs is generated, practitioners can rank the candidate designs according to domain-specific priorities (e.g., prioritizing spatial proximity, global balance, or statistical sensitivity) to select the optimal final matching structure.
    
    \item \textbf{Implementation and Usage:} For complete execution instructions and code documentation, please refer to the \texttt{README.md} file in our project repository at: 
    \newline
    \url{https://github.com/RonAd-1/Initial-release-OSIP-method-pipeline-and-experimental-code}.

\end{itemize}

\bibliography{references} 
\bibliographystyle{abbrvnat}

\clearpage

\appendix
\section{Omitted analysis for OSIP on the Lindner dataset - the careful look at OR20}
\label{OR20-lindner}

We now present the heuristic optimization results for $\delta = 0.20$ in robust mode on the \texttt{Lindner} dataset. Here, the local search heuristics yield clear improvements over the initial Stage 1 solution. Under the robust Mahalanobis distance metric, the baseline Stage 1 solution has a cost (rounded to $4$ decimal digits) of $113.9524$. Both the LS and ELS procedures (reported in Table~\ref{tab:ls_delta0.20}) converge to the same global optimum with an improved cost of $106.0668$. The principal difference lies in their convergence efficiency: ELS reaches this final solution in $7$ improving iterations, whereas LS requires $9$ iterations.

\begin{table}[htbp]
  \centering
  \caption{The intervals for $\delta = 0.20$ in the robust mode, output of LS and ELS}
  \label{tab:ls_delta0.20}
  \begin{filecontents*}{chapter6_schemes/step2_heuristic_local_search_robust_partition_delta_0.20.csv}
start,end,treated,control
0.0417,0.2315,40,255
0.2315,0.4312,172,340
0.4312,0.6252,84,86
0.6252,0.6878,1,3
\end{filecontents*}
\pgfplotstabletypeset[
    col sep=comma,
    string type,
    every head row/.style={before row=\toprule, after row=\midrule},
    every last row/.style={after row=\bottomrule}
  ]{chapter6_schemes/step2_heuristic_local_search_robust_partition_delta_0.20.csv}
\end{table}

Simulated annealing best solution is shown in Table~\ref{tab:sa_delta0.20}. It has an estimated cost of $107.2759$, using $9$ improving iterations. Down below we would analyze these improving iterations in details for all three heuristics.

\begin{table}[htbp]
  \centering
  \caption{Optimal intervals for $\delta = 0.20$ robust, SA}
  \label{tab:sa_delta0.20}
  \begin{filecontents*}{chapter6_schemes/step2_heuristic_simulated_annealing_robust_partition_delta_0.20.csv}
start,end,treated,control
0.0417,0.1713,18,155
0.1713,0.2973,67,213
0.2973,0.4312,127,227
0.4312,0.6252,84,86
0.6252,0.6878,1,3
\end{filecontents*}
\pgfplotstabletypeset[
    col sep=comma,
    string type,
    every head row/.style={before row=\toprule, after row=\midrule},
    every last row/.style={after row=\bottomrule}
  ]{chapter6_schemes/step2_heuristic_simulated_annealing_robust_partition_delta_0.20.csv}
\end{table}

Observe that the final outputs obtained by the heuristics use a very small number of intervals: $4$ for the LS/ELS solution and $5$ for the SA output. Furthermore, because the last two intervals of the SA solution match those of the LS/ELS solution, the distinction between them lies entirely in how the sub-interval $[0.0417, 0.4312]$ is partitioned.

Consider the significance metrics tables of the three heuristics in Tables \ref{tab:delta0.2_LS_steps_sig_rob}, \ref{tab:delta0.2_ELS_steps_sig_rob}, and \ref{tab:delta0.2_SA_steps_sig_rob}.  The non-monotonicity we saw in the strict mode holds also in the robust mode.  Here, the SA table reveals that although the SA heuristic is randomized and allow to move into solutions with a higher estimated cost, all changes carried out by SA strictly improve the estimated cost, so this random acceptance of weaker solutions ability of SA was not used, and it still provides a different solution than the one of ELS.  This structural divergence stems from differing solution-evaluation orderings between the deterministic ELS procedure and the randomized SA search.  Regarding the ATE value, observe that the best value of this measure in the LS or ELS heuristics is obtained by their final step, but this is not the case for SA that achieves its best value at step $5$.

\begin{table}[htbp]
  \centering
  \caption{LS significance metrics}
  \label{tab:delta0.2_LS_steps_sig_rob}
  \begin{filecontents*}{appendix_schemes/lindner/LS_delta0.2_sig_rob.csv}
step,ate,se(ate),p-value,gamma-shift,ci-low,ci-high,t-stat,est-cost
0,-0.1647,0.0391,$< 0.0001$,1.3695,-0.2414,-0.0880,-4.2123,113.9524
1,-0.1657,0.0389,$< 0.0001$,1.4067,-0.2420,-0.0893,-4.2596,112.0846
2,-0.1635,0.0390,$< 0.0001$,1.2954,-0.2399,-0.0871,-4.1923,111.8522
3,-0.1652,0.0389,$< 0.0001$,1.3592,-0.2415,-0.0889,-4.2468,108.7426
4,-0.1644,0.0389,$< 0.0001$,1.3049,-0.2407,-0.0882,-4.2262,108.6861
5,-0.1629,0.0390,$< 0.0001$,1.3892,-0.2394,-0.0865,-4.1769,108.4619
6,-0.1655,0.0389,$< 0.0001$,1.5074,-0.2417,-0.0892,-4.2545,108.4159
7,-0.1640,0.0390,$< 0.0001$,1.5142,-0.2404,-0.0877,-4.2051,107.8283
8,-0.1591,0.0389,$< 0.0001$,1.2697,-0.2353,-0.0829,-4.0900,106.9268
9,-0.1665,0.0389,$< 0.0001$,1.3199,-0.2427,-0.0903,-4.2802,106.0668
\end{filecontents*}
\pgfplotstabletypeset[
    col sep=comma,
    string type,
    every head row/.style={before row=\toprule, after row=\midrule},
    every last row/.style={after row=\bottomrule}
  ]{appendix_schemes/lindner/LS_delta0.2_sig_rob.csv}
\end{table}

\begin{table}[htbp]
  \centering
  \caption{ELS significance metrics}
  \label{tab:delta0.2_ELS_steps_sig_rob}
  \begin{filecontents*}{appendix_schemes/lindner/ELS_delta0.2_sig_rob.csv}
step,ate,se(ate),p-value,gamma-shift,ci-low,ci-high,t-stat,est-cost
0,-0.1647,0.0391,$< 0.0001$,1.3695,-0.2414,-0.088,-4.2123,113.9524
1,-0.1644,0.0389,$< 0.0001$,1.3049,-0.2407,-0.0882,-4.2262,108.6861
2,-0.1579,0.0388,$< 0.0001$,1.2618,-0.2339,-0.0819,-4.0696,108.4026
3,-0.1553,0.0387,$< 0.0001$,1.2528,-0.2312,-0.0793,-4.0129,108.3566
4,-0.1644,0.0389,$< 0.0001$,1.3431,-0.2406,-0.0883,-4.2262,108.0987
5,-0.1653,0.039,$< 0.0001$,1.2855,-0.2417,-0.089,-4.2385,107.4966
6,-0.1572,0.0389,$< 0.0001$,1.2579,-0.2334,-0.0809,-4.0411,106.491
7,-0.1665,0.0389,$< 0.0001$,1.3199,-0.2427,-0.0903,-4.2802,106.0668
\end{filecontents*}
\pgfplotstabletypeset[
    col sep=comma,
    string type,
    every head row/.style={before row=\toprule, after row=\midrule},
    every last row/.style={after row=\bottomrule}
  ]{appendix_schemes/lindner/ELS_delta0.2_sig_rob.csv}
\end{table}

\begin{table}[htbp]
  \centering
  \caption{SA significance metrics}
  \label{tab:delta0.2_SA_steps_sig_rob}
  \begin{filecontents*}{appendix_schemes/lindner/SA_delta0.2_sig_rob.csv}
step,ate,se(ate),p-value,gamma-shift,ci-low,ci-high,t-stat,est-cost
0,-0.1647,0.0391,$< 0.0001$,1.3695,-0.2414,-0.088,-4.2123,113.9524
1,-0.1593,0.0389,$< 0.0001$,1.3278,-0.2356,-0.083,-4.0951,113.8026
2,-0.1627,0.0389,$< 0.0001$,1.3215,-0.2389,-0.0864,-4.1825,110.4605
3,-0.1711,0.0387,$< 0.0001$,1.5227,-0.247,-0.0952,-4.4212,108.8631
4,-0.1678,0.0388,$< 0.0001$,1.5077,-0.2437,-0.0918,-4.3247,108.4884
5,-0.1737,0.0387,$< 0.0001$,1.5137,-0.2496,-0.0978,-4.4884,108.171
6,-0.1648,0.0386,$< 0.0001$,1.3539,-0.2404,-0.0891,-4.2694,107.7145
7,-0.1609,0.0386,$< 0.0001$,1.3076,-0.2365,-0.0852,-4.1684,107.4288
8,-0.1673,0.0385,$< 0.0001$,1.4041,-0.2429,-0.0918,-4.3455,107.3314
9,-0.1653,0.0385,$< 0.0001$,1.3619,-0.2408,-0.0898,-4.2935,107.2759
\end{filecontents*}
\pgfplotstabletypeset[
    col sep=comma,
    string type,
    every head row/.style={before row=\toprule, after row=\midrule},
    every last row/.style={after row=\bottomrule}
  ]{appendix_schemes/lindner/SA_delta0.2_sig_rob.csv}
\end{table}

Next, consider the distance and balance metrics across the heuristics in Tables~\ref{tab:delta0.2_LS_steps_dist_rob}, \ref{tab:delta0.2_ELS_steps_dist_rob}, and \ref{tab:delta0.2_SA_steps_dist_rob}. Here, we again observe a strong alignment between lower estimated cost and reduced mean Mahalanobis distance. Across all three heuristics, Max-Mahal improves from an initial value of $4.8636$ to a common final value of $3.5798$. However, while LS and ELS achieve this terminal value in a single step (step 1), SA's Max-Mahal initially increases to $5.5436$ at step 1 before decreasing to $3.5798$ at step 3. Neither Median-Mahal nor SD-Mahal undergoes a net increase after executing any heuristic to completion. Regarding covariate balance, the $\text{Max-SMD}$ metric reaches $0.0505$ during intermediate steps of ELS and throughout the final iterations of SA. Although the baseline \texttt{OR20} solution exhibits a higher $\text{Max-SMD}$, transitioning to \texttt{OR25} (adjusting $\delta$) or \texttt{OS20} (changing the OSIP mode) yields even smaller $\text{Max-SMD}$ values (See Table~\ref{tab:delta_all_dist_lindner}).

\begin{table}[htbp]
  \centering
  \caption{LS distances metrics}
  \label{tab:delta0.2_LS_steps_dist_rob}
  \begin{filecontents*}{appendix_schemes/lindner/LS_delta0.2_dist_rob.csv}
step,Max-SMD,Mean-SMD,Max-KS,Mean-Mahal,Median-Mahal,Max-Mahal,SD-Mahal
0,0.0539,0.0261,0.0606,0.4979,0.3245,4.8636,0.6975
1,0.0505,0.0225,0.0606,0.4905,0.3245,3.5798,0.6785
2,0.0539,0.0233,0.0539,0.4898,0.3245,3.5798,0.6761
3,0.0539,0.0239,0.0539,0.4843,0.3245,3.5798,0.6654
4,0.0539,0.0234,0.0539,0.4842,0.3173,3.5798,0.6664
5,0.0572,0.0260,0.0572,0.4875,0.3245,3.5798,0.6709
6,0.0572,0.0266,0.0572,0.4873,0.3348,3.5798,0.6709
7,0.0572,0.0272,0.0572,0.4816,0.3245,3.5798,0.6581
8,0.0539,0.0268,0.0572,0.4810,0.3173,3.5798,0.6495
9,0.0539,0.0285,0.0572,0.4795,0.3172,3.5798,0.6549
\end{filecontents*}
\pgfplotstabletypeset[
    col sep=comma,
    string type,
    every head row/.style={before row=\toprule, after row=\midrule},
    every last row/.style={after row=\bottomrule}
  ]{appendix_schemes/lindner/LS_delta0.2_dist_rob.csv}
\end{table}

\begin{table}[htbp]
  \centering
  \caption{ELS distances metrics}
  \label{tab:delta0.2_ELS_steps_dist_rob}
  \begin{filecontents*}{appendix_schemes/lindner/ELS_delta0.2_dist_rob.csv}
step,Max-SMD,Mean-SMD,Max-KS,Mean-Mahal,Median-Mahal,Max-Mahal,SD-Mahal
0,0.0539,0.0261,0.0606,0.4979,0.3245,4.8636,0.6975
1,0.0539,0.0234,0.0539,0.4842,0.3173,3.5798,0.6664
2,0.0505,0.0252,0.0539,0.4861,0.3245,3.5798,0.6596
3,0.0505,0.0259,0.0539,0.4858,0.3245,3.5798,0.6595
4,0.0505,0.0267,0.0539,0.4887,0.3245,3.5798,0.6676
5,0.0515,0.0276,0.0539,0.4843,0.3245,3.5798,0.664
6,0.0515,0.0288,0.0572,0.4794,0.3173,3.5798,0.6504
7,0.0539,0.0285,0.0572,0.4795,0.3172,3.5798,0.6549
\end{filecontents*}
\pgfplotstabletypeset[
    col sep=comma,
    string type,
    every head row/.style={before row=\toprule, after row=\midrule},
    every last row/.style={after row=\bottomrule}
  ]{appendix_schemes/lindner/ELS_delta0.2_dist_rob.csv}
\end{table}

\begin{table}[htbp]
  \centering
  \caption{SA distances metrics}
  \label{tab:delta0.2_SA_steps_dist_rob}
  \begin{filecontents*}{appendix_schemes/lindner/SA_delta0.2_dist_rob.csv}
step,Max-SMD,Mean-SMD,Max-KS,Mean-Mahal,Median-Mahal,Max-Mahal,SD-Mahal
0,0.0539,0.0261,0.0606,0.4979,0.3245,4.8636,0.6975
1,0.0471,0.0224,0.0606,0.4971,0.3245,5.5436,0.7172
2,0.0505,0.0228,0.0539,0.4909,0.3245,5.5436,0.7049
3,0.0539,0.026,0.0572,0.4889,0.3348,3.5798,0.6649
4,0.0539,0.0252,0.0572,0.4883,0.3348,3.5798,0.6645
5,0.0539,0.0253,0.0572,0.4872,0.3348,3.5798,0.6643
6,0.0505,0.0252,0.0572,0.4881,0.3245,3.5798,0.6569
7,0.0505,0.0249,0.0572,0.486,0.3245,3.5798,0.6552
8,0.0505,0.0257,0.0572,0.4863,0.3348,3.5798,0.6556
9,0.0505,0.0255,0.0572,0.4861,0.3245,3.5798,0.6556
\end{filecontents*}
\pgfplotstabletypeset[
    col sep=comma,
    string type,
    every head row/.style={before row=\toprule, after row=\midrule},
    every last row/.style={after row=\bottomrule}
  ]{appendix_schemes/lindner/SA_delta0.2_dist_rob.csv}
\end{table}

\clearpage

\section{The NSW-MIXTAPE dataset}
\label{app:mixtape}

We now describe the \texttt{OSIP} comparison on another dataset, the  \texttt{NSW-Mixtape} dataset \citep{cunningham2021causal, lalonde1986evaluating, dehejia1999causal}.
Downloaded from R's causaldata package, it is related to the data national supported work demonstration (NSW) job-training program experiment, where those treated were guaranteed a job for 9-18 months. It has $445$ units and $11$ covariates.  First, consider in more details the set of variables characterizing each unit in Table \ref{tab:nsw_variables}.

\begin{table}[htbp]
\centering
\caption{Description of Variables in the NSW-Mixtape Dataset ($N = 445$)}
\label{tab:nsw_variables}
\begin{tabularx}{\textwidth}{l l X}
\toprule
\textbf{Variable Name} & \textbf{Data Type} & \textbf{Description} \\
\midrule
\texttt{data\_id} & Numeric Identifier & Individual ID. \\
\addlinespace
\texttt{treat} & Binary (0/1) & Treatment Indicator: 1 if assigned to the National Supported Work Demonstration Job Training Program; 0 otherwise (185 treated, 260 control). \\
\addlinespace
\texttt{age} & Numeric integer & Covariate: Age in years. \\
\addlinespace
\texttt{educ} & Numeric integer & Covariate: Years of completed education. \\
\addlinespace
\texttt{black} & Binary (0/1) & Covariate: Race (1 = Black, 0 = Otherwise). \\
\addlinespace
\texttt{hisp} & Binary (0/1) & Covariate: Ethnicity (1 = Hispanic, 0 = Otherwise). \\
\addlinespace
\texttt{marr} & Binary (0/1) & Covariate: Marital status (1 = Married, 0 = Otherwise). \\
\addlinespace
\texttt{nodegree} & Binary (0/1) & Covariate: High school degree indicator (1 = Has no high school diploma, 0 = Has diploma). \\
\addlinespace
\texttt{re74} & Numeric value & Baseline Covariate: Real earnings in 1974 (US dollars). \\
\addlinespace
\texttt{re75} & Numeric value & Baseline Covariate: Real earnings in 1975 (US dollars). \\
\addlinespace
\texttt{re78} & Numeric value & Primary Outcome: Real earnings in 1978 post-intervention (US dollars). \\
\bottomrule
\end{tabularx}
\end{table}

As with the \texttt{Lindner} dataset, we first perform the trimming to common support step. This step results in $181$ treated units and $249$ control units. Following this step the common support of the propensity scores becomes $[0.2379, 0.6433]$.

We present the \texttt{OSIP} results across all tested values of $\delta \in \{ 0.15,0.2, 0.25\}$. 
Similarly to the layout in Section~\ref{sub:comp_delta},
we now present the corresponding evaluation tables for the \texttt{NSW-Mixtape} dataset. 

First, Table~\ref{tab:delta_all_est_conv_nsw} reports heuristic performance regarding the auxiliary objective of minimizing the estimated solution cost alongside the number of improving steps. Entries designated as ``NC'' indicate no change from the Stage~1 baseline solution. At $\delta = 0.25$, the heuristics fail to improve upon the initial Stage~1 output for either mode. For $\delta = 0.15$ and $\delta = 0.20$, Simulated Annealing (SA) consistently yields the lowest estimated costs across both modes. Across all values of $\delta$, the estimated cost in robust mode is substantially lower than in strict mode; for instance, at $\delta = 0.25$, robust mode achieves an estimated cost of $107.4958$ compared to $130.1412$ in strict mode. Overall, the number of improving steps executed remains modest, peaking at $4$ steps for SA under \texttt{OS15} and \texttt{OR20}. 

\begin{table}[htbp]
  \centering
  \small
  \caption{OSIP: Estimated cost and conversion rates for all $\delta$ values, NSW-Mixtape}
  \label{tab:delta_all_est_conv_nsw}
  \begin{filecontents*}{appendix_schemes/mixtape/cost_steps_mixtape.csv}
method,step 1,LS ,ELS ,SA
OS15,138.7039,136.3921; 2,136.3921; 1,133.8031; 4
OR15,120.668,119.2838; 2,119.2838; 2,119.1758; 1
OS20,131.1793,NC ,NC ,130.6109; 2
OR20,109.4921,108.8462;3,108.8462; 2,107.3594; 4
OS25,130.1412,NC ,NC ,NC 
OR25,107.4958,NC ,NC ,NC 
\end{filecontents*}
\pgfplotstabletypeset[
    col sep=comma,
    string type,
    every head row/.style={before row=\toprule, after row=\midrule},
    every last row/.style={after row=\bottomrule}
  ]{appendix_schemes/mixtape/cost_steps_mixtape.csv}
\end{table}

Next, Table~\ref{tab:delta_all_sig_nsw} summarizes the statistical significance metrics across the six resulting matched samples. Estimated ATE values range from $1530.991$ (under \texttt{OR25}) to $1858.7901$ (under \texttt{OS25}). The \texttt{OS25} configuration achieves the lowest $p$-value ($0.0094$) and highest $t$-statistic ($2.613$) among all candidate solutions. Conversely, its robust counterpart (\texttt{OR25}) exhibits the highest $p$-value ($0.0345$) and lowest $t$-statistic ($2.1222$), though it remains statistically significant at the $5\%$ level. Sensitivity to unobserved confounding, measured by the $\Gamma$-shift, peaks under strict mode with $\texttt{OS20}$ ($\Gamma = 1.0846$), followed closely by \texttt{OS15} ($\Gamma = 1.0740$) and \texttt{OS25} ($\Gamma = 1.0720$), whereas \texttt{OR25} yields the lowest bound ($\Gamma = 1.0260$). Consequently, when prioritizing statistical power and unobserved confounding robustness, \texttt{OS25} represents the primary design candidate, while \texttt{OR25} serves as a secondary alternative.

\begin{table}[htbp]
  \centering
  \small
  \caption{OSIP: Significance table for all values of $\delta$, NSW-Mixtape}
  \label{tab:delta_all_sig_nsw}
  \begin{filecontents*}{appendix_schemes/mixtape/nsw_mixtape_delta_all_sig.csv}
method,ate,se(ate),p-value,gamma-shift,ci-low,ci-high,t-stat
OS15,1836.8409,714.0825,0.0105,1.074,437.2391,3236.4426,2.5723
OR15,1660.3614,713.3819,0.0205,1.0447,262.1329,3058.5898,2.3275
OS20,1578.0865,719.0185,0.0288,1.0846,168.8102,2987.3628,2.1948
OR20,1643.0731,712.8842,0.0218,1.0549,245.8201,3040.3262,2.3048
OS25,1858.7901,711.361,0.0094,1.072,464.5226,3253.0576,2.613
OR25,1530.991,721.4283,0.0345,1.026,116.9916,2944.9905,2.1222
\end{filecontents*}
\pgfplotstabletypeset[
    col sep=comma,
    string type,
    every head row/.style={before row=\toprule, after row=\midrule},
    every last row/.style={after row=\bottomrule}
  ]{appendix_schemes/mixtape/nsw_mixtape_delta_all_sig.csv}
\end{table}

Turning to spatial proximity and covariate balance, Tables~\ref{tab:delta_all_dist_nsw} and \ref{tab:love_plot_tbl_mixtape} report the corresponding distance and imbalance metrics across all configurations. First, examining Mahalanobis distance measures ($\text{Mean-Mahal}$, $\text{Median-Mahal}$, $\text{Max-Mahal}$, and $\text{SD-Mahal}$), robust solutions consistently achieve superior spatial matching relative to their strict counterparts across all values of $\delta$. In robust mode, distance metrics improve monotonically with $\delta$, reaching their global minimum at $\delta = 0.25$ (\texttt{OR25}). In contrast, strict mode lacks monotonicity, as $\delta = 0.20$ (\texttt{OS20}) yields lower distance metrics than either $\delta = 0.15$ or $\delta = 0.25$. 

Covariate balance, as measured by $\text{Max-SMD}$, presents the defining trade-off for this dataset. Solutions configured with $\delta = 0.20$ exhibit substantial imbalance, with $\text{Max-SMD}$ exceeding $0.1188$. Notably, \texttt{OS20} leaves $3$ out of $8$ thresholded covariates imbalanced ($\text{Nib} = 3;0$). Because balance preservation is prioritized in observational design, $\delta = 0.20$ should be avoided. Granular examination of individual covariates in Table~\ref{tab:love_plot_tbl_mixtape} reveals that the \texttt{educ} covariate drives maximum imbalance in most configurations, with \texttt{age} representing the second most problematic variable (and primary driver of $\text{Max-SMD}$ in \texttt{OS25}). Conversely, the indicator covariate \texttt{hisp} achieves perfect balance ($\text{SMD} = 0$) in half of the six evaluated configurations.

\begin{table}[htbp]
  \centering
  \small
  \caption{OSIP: Distances table for all values of $\delta$, NSW-Mixtape}
  \label{tab:delta_all_dist_nsw}
  \begin{filecontents*}{appendix_schemes/mixtape/nsw_mixtape_delta_all_dist.csv}
method,Mean-Mahal,Median-Mahal,Max-Mahal,SD-Mahal,Max-SMD,Mean-SMD,Max-KS,Nib
OS15,1.0755,0.6443,5.8313,1.307,0.0737,0.0395,0.0773,0;0
OR15,0.8788,0.525,3.481,0.977,0.0936,0.0381,0.0718,0;0
OS20,0.9467,0.5476,4.9690,1.0626,0.1188,0.0627,0.0939,3;0
OR20,0.8315,0.4875,3.8280,0.8870,0.1339,0.0454,0.0994,1;0
OS25,1.0318,0.6049,5.457,1.2222,0.075,0.0389,0.0663,0;0
OR25,0.81,0.4875,3.105,0.8827,0.0876,0.0349,0.0718,0;0
\end{filecontents*}
\pgfplotstabletypeset[
    col sep=comma,
    string type,
    every head row/.style={before row=\toprule, after row=\midrule},
    every last row/.style={after row=\bottomrule}
  ]{appendix_schemes/mixtape/nsw_mixtape_delta_all_dist.csv}
\end{table}

\begin{table}[htbp]
  \centering
  \small
  \caption{OSIP Balance table for NSW-Mixtape}
  \label{tab:love_plot_tbl_mixtape}
  \begin{filecontents*}{appendix_schemes/mixtape/love_plot_table_combined.csv}
Covariate,OS15,OR15,OS20,OR20,OS25,OR25
age,0.0734,0.0581,0.1049,0.0435,0.075,0.0509
educ,0.0737,0.0936,0.0834,0.1339,0.0678,0.0876
black,0.0166,0.011,0.0055,-0.011,0.0166,0.0055
hisp,-0.011,0,0,0.0055,-0.011,0
marr,-0.0331,-0.0221,0.0166,0.0166,-0.0331,-0.0221
nodegree,-0.0442,-0.0497,-0.0663,-0.0994,-0.0442,-0.0442
re74,0.0367,-0.0122,0.1188,0.0277,0.0307,0.0187
re75,-0.027,-0.0582,0.1061,-0.0258,-0.0328,-0.0499
\end{filecontents*}
\pgfplotstabletypeset[
    col sep=comma,
    string type,
    every head row/.style={before row=\toprule, after row=\midrule},
    every last row/.style={after row=\bottomrule}
  ]{appendix_schemes/mixtape/love_plot_table_combined.csv}
\end{table}

In summary, because solutions at $\delta=0.20$ fail covariate balance requirements, they are discarded.  Among the remaining configurations, the ones corresponding to $\delta=0.25$ are the optimal candidate set to be evaluated in the next step in comparison to the solutions obtained by other methods. The robust mode is preferred if the Mahalanobis distance measures are paramount, while the strict mode is preferred when prioritizing statistical significance.  Both solutions are retained for the next step when we evaluate \texttt{OSIP} versus other methods with respect to the \texttt{NSW-Mixtape} dataset.

We now compare the two selected \texttt{OSIP} configurations (\texttt{OS25} and \texttt{OR25}) against state-of-the-art alternative methods on the \texttt{NSW-Mixtape} dataset, following the comparative structure outlined in Section~\ref{OSIP_vs_all_lindner}. Table~\ref{tab:delta0.25_all_mixtape_sig} summarizes the statistical significance metrics across all evaluated methods (see also Figure~\ref{fig:ate_plot_mixtape} for a graphical representation of the estimated ATE values and their corresponding confidence intervals).

Both \texttt{OSIP} variants retain all $181$ treated units in the trimmed common support along with $181$ matched control units under $1:1$ matching. While benchmark algorithms such as \texttt{Opt}, \texttt{Card}, \texttt{FM}, and \texttt{Gen} also preserve the full set of $181$ treated units, methods such as \texttt{CEM}, \texttt{Q}, and \texttt{RQ} suffer sample loss by discarding treated observations. 

With respect to estimated treatment effects, \texttt{CEM} and \texttt{FM} report the highest ATE estimates. \texttt{OS25} achieves the third-largest ATE overall ($1858.7901$), paired with a strong $p$-value ($0.0094$), a $t$-statistic of $2.6130$, and the highest sensitivity bound against unobserved confounding ($\Gamma = 1.0720$). Meanwhile, \texttt{OR25} outperforms \texttt{Opt} across every significance metric except $\Gamma$-shift, strictly dominates \texttt{Card} across all significance metrics, and presents a competitive trade-off relative to \texttt{Gen}. 

Evaluating purely on statistical significance, we exclude \texttt{CEM}, \texttt{Q}, and \texttt{RQ} due to sample attrition, identify \texttt{FM} as the top unconstrained benchmark, rank \texttt{OS25} as our primary 1:1 matching solution, and jointly rank \texttt{OR25} and \texttt{Gen} as secondary choices superior to remaining 1:1 matching benchmarks.

\begin{table}[htbp]
  \centering
  \small
  \caption{Significance table for all models on NSW-Mixtape}
  \label{tab:delta0.25_all_mixtape_sig}
  \begin{filecontents*}{appendix_schemes/mixtape/consolidated_results_nsw_mixtape_delta_0.25_power_1_sig.csv}
method,treated,control,ate,se(ate),p-value,gamma-shift,ci-low,ci-high,t-stat
UM,181,249,1673.2903,643.0512,0.0096,NA,412.9099,2933.6707,2.6021
OS25,181,181,1858.7901,711.3610,0.0094,1.0720,464.5226,3253.0576,2.6130
OR25,181,181,1530.9910,721.4283,0.0345,1.0260,116.9916,2944.9905,2.1222
Opt ,181,181,1493.3520,724.7933,0.0401,1.0498,72.7571,2913.9469,2.0604
Card,181,181,1357.8293,723.6923,0.0614,1.0061,-60.6076,2776.2662,1.8763
CEM,103,147,2571.7431,824.7806,0.0020,NA,955.1731,4188.3131,3.1181
FM,181,249,2072.9004,637.9304,0.0012,NA,822.5568,3323.2439,3.2494
Gen ,181,181,1621.2878,719.3277,0.0248,1.0149,211.4055,3031.1700,2.2539
Q ,169,169,1587.4455,705.0883,0.0250,1.0544,205.4725,2969.4185,2.2514
RQ,118,118,1857.8031,891.9752,0.0384,1.0702,109.5317,3606.0744,2.0828
\end{filecontents*}
\pgfplotstabletypeset[
    col sep=comma,
    string type,
    every head row/.style={before row=\toprule, after row=\midrule},
    every last row/.style={after row=\bottomrule}
  ]{appendix_schemes/mixtape/consolidated_results_nsw_mixtape_delta_0.25_power_1_sig.csv}
\end{table}

\begin{figure}[htbp]
  \centering
  \small
  \includegraphics[width=0.85\textwidth]{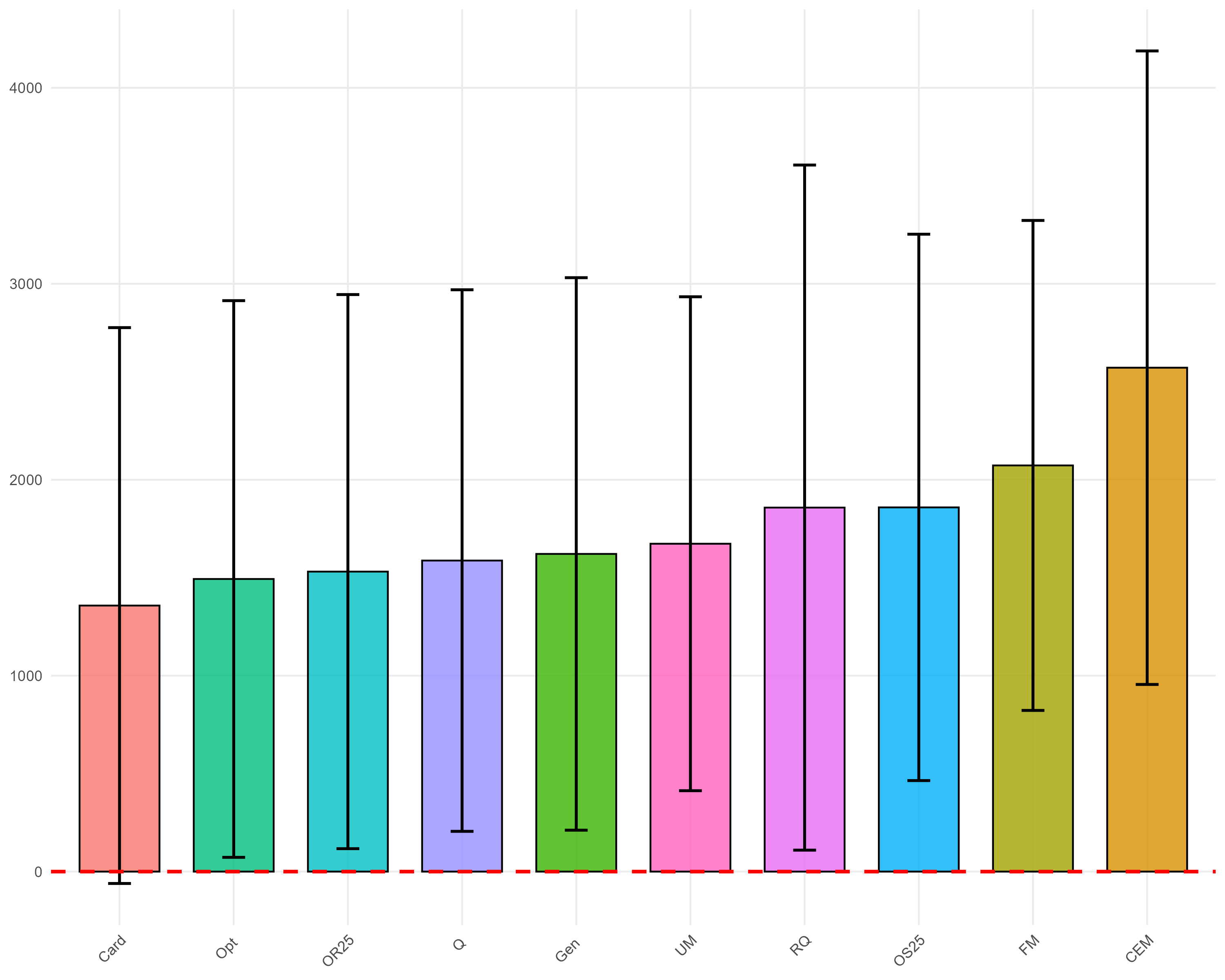} 
  \caption{ATE plot for all models, NSW-Mixtape}
  \label{fig:ate_plot_mixtape}
\end{figure}

Table~\ref{tab:delta0.25_all_mixtape_dist} summarizes the comparative matching performance across all methods on the \texttt{NSW-Mixtape} dataset in terms of distance and balance measures. First, consider $\text{Max-SMD}$ and the number of imbalanced covariates (encoded by $\text{Nib}$). Here, the unmatched baseline (\texttt{UM}) admits two imbalanced covariates, and two competing methods, namely \texttt{Opt} and \texttt{Card}, fail to decrease this count. With respect to $\text{Max-SMD}$, \texttt{Gen} is the clear winner with a very low value of $0.0497$, while \texttt{OS25} and \texttt{OR25} rank second and third, respectively, among solutions retaining the entire set of treated units. With respect to other balance measures, such as $\text{Mean-SMD}$, \texttt{FM} performs better than \texttt{OS25}, but \texttt{OR25} achieves an even lower value, with \texttt{Gen} remaining the top performer. As for the $\text{Max-KS}$ metric, \texttt{FM} achieves the best value, while \texttt{OS25} and \texttt{Gen} tie for second best, followed by the \texttt{OR25} solution.


The distance measures for \texttt{OR25} and \texttt{OS25} are substantially superior to those of competing $1:1$ matching methods that retain the full set of treated units while achieving complete covariate balance ($\text{Nib} = 0;0$). For example, the $\text{Mean-Mahal}$ value for \texttt{OR25} is $0.8100$ and for \texttt{OS25} is $1.0318$, whereas \texttt{Gen} yields $1.4815$. Direct comparison with \texttt{Gen} demonstrates similar dominance across all other Mahalanobis distance measures reported in Table~\ref{tab:delta0.25_all_mixtape_dist}. In fact, \texttt{OR25} achieves the lowest Mahalanobis distance metrics across all full-sample $1:1$ matching benchmarks. Among these methods, \texttt{Opt} ranks second-best in distance metrics, while \texttt{OS25} ranks third; however, unlike both \texttt{OSIP} solutions, \texttt{Opt} fails to balance all covariates ($\text{Nib} = 2;0$).


\begin{table}[htbp]
  \centering
  \small
   \caption{Distances table for $\delta=0.25$ all models, NSW-Mixtape}
  \label{tab:delta0.25_all_mixtape_dist}
  \begin{filecontents*}{appendix_schemes/mixtape/consolidated_results_nsw_mixtape_delta_0.25_power_1_dist.csv}
method,Mean-Mahal,Median-Mahal,Max-Mahal,SD-Mahal,Max-SMD,Mean-SMD,Max-KS,Nib
UM,NA,NA,NA,NA,0.1169,0.0508,0.1091,2;0
OS25,1.0318,0.6049,5.457,1.2222,0.075,0.0389,0.0663,0;0
OR25,0.81,0.4875,3.105,0.8827,0.0876,0.0349,0.0718,0;0
Opt ,0.9128,0.5811,4.0873,1.0008,0.1297,0.065,0.0939,2;0
Card,1.1431,0.7408,5.7514,1.2293,0.1578,0.0592,0.0994,2;0
CEM,NA,NA,NA,NA,0.0547,0.0113,0.0453,0;0
FM,NA,NA,NA,NA,0.0967,0.0363,0.0497,0;0
Gen ,1.4815,0.8912,7.6683,1.6207,0.0497,0.0198,0.0663,0;0
Q ,0.9878,0.473,5.7306,1.2529,0.0658,0.0233,0.0769,0;0
RQ,0.7902,0.4621,5.18,0.9051,0.0579,0.0159,0.0339,0;0
\end{filecontents*}
\pgfplotstabletypeset[
    col sep=comma,
    string type,
    every head row/.style={before row=\toprule, after row=\midrule},
    every last row/.style={after row=\bottomrule}
  ]{appendix_schemes/mixtape/consolidated_results_nsw_mixtape_delta_0.25_power_1_dist.csv}
\end{table}

Table~\ref{tab:delta0.25_bal_mixtape} presents the covariate-level standardized mean differences ($\text{SMD}$) across all evaluated models for the \texttt{NSW-Mixtape} dataset (see also Figure~\ref{fig:love_plot_mixtape} for a Love plot illustration). The baseline unmatched sample (\texttt{UM}) exhibits notable initial imbalance, particularly regarding \texttt{educ} ($\text{SMD} = 0.1169$) and \texttt{nodegree} ($\text{SMD} = -0.1091$). Standard benchmark methods struggle to reconcile the joint covariate distribution without shifting imbalance onto other variables. Specifically, \texttt{FM} overcomes baseline imbalances only at the cost of inducing substantial imbalance in the financial covariates \texttt{re74} ($\text{SMD} = -0.0806$) and \texttt{re75} ($\text{SMD} = -0.0967$, nearly three times the initial baseline imbalance of \texttt{re75}). \texttt{Opt} resolves initial baseline issues but severely unbalances \texttt{age} ($0.1220$) and \texttt{re75} ($0.1297$). Similarly, while \texttt{Card} achieves exact balance on \texttt{age} ($0$) and rectifies baseline deficiencies, it causes severe imbalance in both earnings variables, \texttt{re74} ($0.1048$) and \texttt{re75} ($0.1578$).

In contrast, \texttt{Gen} demonstrates strong balancing capability, resolving baseline issues and achieving exact balance on \texttt{black}, \texttt{hisp}, and \texttt{marr} without causing severe variance spillover to remaining covariates, all while preserving the full treated sample. In particular, \texttt{Gen} outperforms \texttt{FM} in balancing every covariate except \texttt{age} and \texttt{nodegree}. Finally, both \texttt{OS25} and \texttt{OR25} successfully keep all eight covariates well within the standard $0.10$ threshold, effectively managing the trade-off between demographic and financial variables. \texttt{OS25} records its maximum imbalance on \texttt{age} ($0.0750$), whereas \texttt{OR25} records its maximum on \texttt{educ} ($0.0876$) alongside achieving exact balance on \texttt{hisp} ($0$).

\begin{table}[htbp]
  \centering
  \small
   \caption{Balance table all models, NSW-Mixtape}
  \label{tab:delta0.25_bal_mixtape}
  \begin{filecontents*}{appendix_schemes/mixtape/love_plot_table_delta_0.25_mixtape_transposed.csv}
method,age,educ,black,hisp,marr,nodegree,re74,re75
UM,0.0703,0.1169,0.0009,-0.0155,0.0257,-0.1091,0.0367,0.0313
OS25,0.075,0.0678,0.0166,-0.011,-0.0331,-0.0442,0.0307,-0.0328
OR25,0.0509,0.0876,0.0055,0,-0.0221,-0.0442,0.0187,-0.0499
Opt ,0.122,0.0806,0,0,0.0331,-0.0773,0.0774,0.1297
Card,0,0.0618,0.011,-0.0221,0.0497,-0.0663,0.1048,0.1578
CEM,-0.0201,-0.0044,0.0000,0.0000,0.0000,0.0000,0.0114,-0.0547
FM,0.0377,0.0257,0.0110,-0.0110,0.0055,-0.0221,-0.0806,-0.0967
Gen ,0.0411,-0.0213,0,0,0,-0.0497,-0.005,-0.0414
Q ,0.0658,0.0373,0.0059,-0.0059,-0.0237,-0.0118,0.0308,-0.005
RQ,0.0371,-0.0579,0.0085,0,0,0,-0.0128,-0.0111
\end{filecontents*}
\pgfplotstabletypeset[
    col sep=comma,
    string type,
    every head row/.style={before row=\toprule, after row=\midrule},
    every last row/.style={after row=\bottomrule}
  ]{appendix_schemes/mixtape/love_plot_table_delta_0.25_mixtape_transposed.csv}
\end{table}

\begin{figure}[htbp]
  \centering
  \small
  \includegraphics[width=0.85\textwidth]{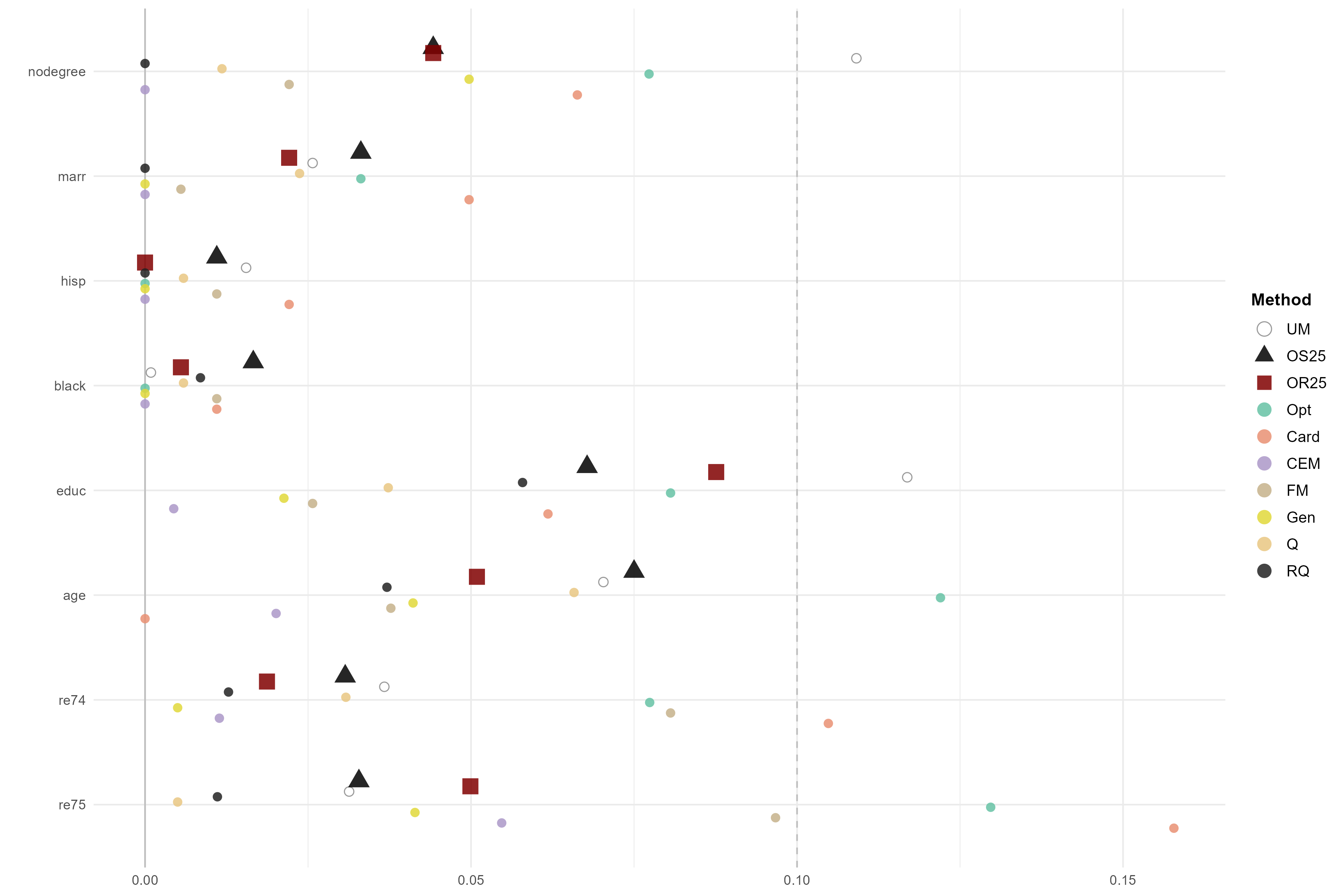} 
  \caption{Love plot for all models, NSW-Mixtape}
  \label{fig:love_plot_mixtape}
\end{figure}

With respect to the balance measures, \texttt{Gen} shows a good alternative to our proposed solutions \texttt{OR25} and \texttt{OS25} but it has significantly larger Mahalanobis distance measures. The other alternatives retaining all treated units perform worse with respect to this crucial balancing requirement. Thus, for this dataset, our proposed solution is to use \texttt{OR25} if Mahalanobis distance measures are prioritized, \texttt{OS25} is preferred when prioritizing statistical significance, while we suggest keeping \texttt{Gen} as a reasonable alternative to this pair of new solutions obtained by \texttt{OSIP} if balancing measures are prioritized.

\section{IHDP Dataset}
\label{ihdp}
We next evaluate our framework on the Infant Health and Development Program (\texttt{IHDP}) dataset, obtained via the \texttt{bartcs} R package. Originating from an intensive randomized targeted intervention conducted between $1985$ and $1988$, the study evaluated the effect of home visits and specialized educational care on cognitive test scores of low birth-weight infants. Initially introduced to the causal inference literature by \citep{hill2011bayesian}, this dataset has served as a standard benchmark in subsequent methodological studies \citep[e.g.,][]{shalit2017estimating, louizos2017causal}. 
Because the
IHDP dataset originates from a randomized trial with a strong treatment signal, all matching configurations exhibit
overwhelming statistical significance.
The dataset comprises $134$ treated units, $556$ control units, and $25$ covariates. Trimming the sample to common support yields a propensity score interval of $[0.0335, 0.5749]$. The dataset includes the following covariates:


\begin{itemize}
    \item \textbf{treatment}: Binary indicator of treatment assignment ($1 =$ treated, $0 =$ control).
    \item \textbf{y\_factual}: Observed (factual) outcome.
    \item \textbf{y\_cfactual}: Counterfactual (potential) outcome given the opposite treatment assignment.
    \item \textbf{mu0}: Control conditional mean ($\mathbb{E}[Y(0) \mid X]$).
    \item \textbf{mu1}: Treated conditional mean ($\mathbb{E}[Y(1) \mid X]$).
    \item \textbf{X1 -- X6}: Continuous confounders/covariates.
    \item \textbf{X7 -- X25}: Binary confounders/covariates.
\end{itemize}

For the purpose of matching, only $X_{1}\text{--}X_{25}$ and the treatment indicator are used. First, we evaluate \texttt{OSIP} solutions for $\delta \in \{0.1, 0.15, 0.2\}$ under both strict and robust modes. For this evaluation, consider the heuristic performance regarding the auxiliary goal of minimizing estimated cost and the number of improving steps carried out by each heuristic in Table~\ref{tab:delta_all_est_conv_ihdp}. Here, "NC" denotes no change from the solution obtained at the end of Step~1. Regarding the estimated cost of the resulting solutions, for every heuristic and every value of $\delta$, the solution obtained under strict mode yields a smaller estimated cost than the corresponding solution under robust mode. Furthermore, for every heuristic, increasing $\delta$ decreases the estimated cost for the same operational mode. 

Note that the solutions for $\delta = 0.2$ obtained at Step~1 were not improved by any heuristic, whereas for smaller values of $\delta$, every heuristic achieved further cost reductions. The maximum number of transition moves executed by any heuristic is $9$, performed by SA when constructing the solution for $\delta = 0.15$ in robust mode. Note that for this parameter configuration (robust mode and $\delta = 0.15$), the final \texttt{OR15} solution is obtained by LS in only $4$ iterations while achieving a lower cost ($495.6527$) than SA ($495.6673$).


\begin{table}[htbp]
  \centering
  \small
  \caption{OSIP: Estimated cost and conversion rates for all $\delta$ values, IHDP}
  \label{tab:delta_all_est_conv_ihdp}
  \begin{filecontents*}{appendix_schemes/ihdp/cost_steps_ihdp.csv}
method,step 1,LS ,ELS ,SA
OS10,511.7039,508.8898; 4,508.8898; 2,506.7678; 8
OR10,524.2019,522.4704; 2,522.4704; 2,518.0531; 5
OS15,486.4641,483.7097; 1,483.7097; 2,483.7097; 2
OR15,499.4293,495.6527; 4,496.5762; 4,495.6673; 9
OS20,466.2632,NC,NC,NC
OR20,478.4633,NC,NC,NC
\end{filecontents*}
\pgfplotstabletypeset[
    col sep=comma,
    string type,
    every head row/.style={before row=\toprule, after row=\midrule},
    every last row/.style={after row=\bottomrule}
  ]{appendix_schemes/ihdp/cost_steps_ihdp.csv}
\end{table}

Next, consider the statistical significance metrics of the OSIP solutions presented in Table~\ref{tab:delta_all_sig_ihdp}. Our results demonstrate strong stability in causal effect estimation across all model variants. The estimated ATE remains strictly within the narrow range of $3.8779$ to $3.9682$, the $p$-value is consistently $< 0.0001$, and all $t$-statistics comfortably exceed $25$. Furthermore, every solution achieves a $\Gamma$-shift value greater than $5$, confirming that the estimated treatment effect is exceptionally resilient against potential unobserved selection bias. Because all configurations display such overwhelming statistical significance, 
differentiating between them based on significance metrics alone is uninformative. Given these metrics, we instead prioritize and evaluate the solutions using covariate balance and distance measures.


\begin{table}[htbp]
  \centering
  \small
  \caption{OSIP: Significance table for all values of $\delta$, IHDP}
  \label{tab:delta_all_sig_ihdp}
  \begin{filecontents*}{appendix_schemes/ihdp/ihdp_delta_all_sig.csv}
method,ate,se(ate),p-value,gamma-shift,ci-low,ci-high,t-stat
OS10,3.8779,0.1359,$< 0.0001$,$>5$,3.6116,4.1443,28.5350
OR10,3.9269,0.1438,$< 0.0001$,$>5$,3.6451,4.2086,27.3081
OS15,3.8900,0.1423,$< 0.0001$,$>5$,3.6111,4.1689,27.3366
OR15,3.9507,0.1481,$< 0.0001$,$>5$,3.6605,4.2409,26.6759
OS20,3.9682,0.1404,$< 0.0001$,$>5$,3.6929,4.2434,28.2635
OR20,3.9274,0.1446,$< 0.0001$,$>5$,3.6440,4.2107,27.1604
\end{filecontents*}
\pgfplotstabletypeset[
    col sep=comma,
    string type,
    every head row/.style={before row=\toprule, after row=\midrule},
    every last row/.style={after row=\bottomrule}
  ]{appendix_schemes/ihdp/ihdp_delta_all_sig.csv}
\end{table}

\begin{figure}[htbp]
  \centering
  \small
  \includegraphics[width=0.85\textwidth]{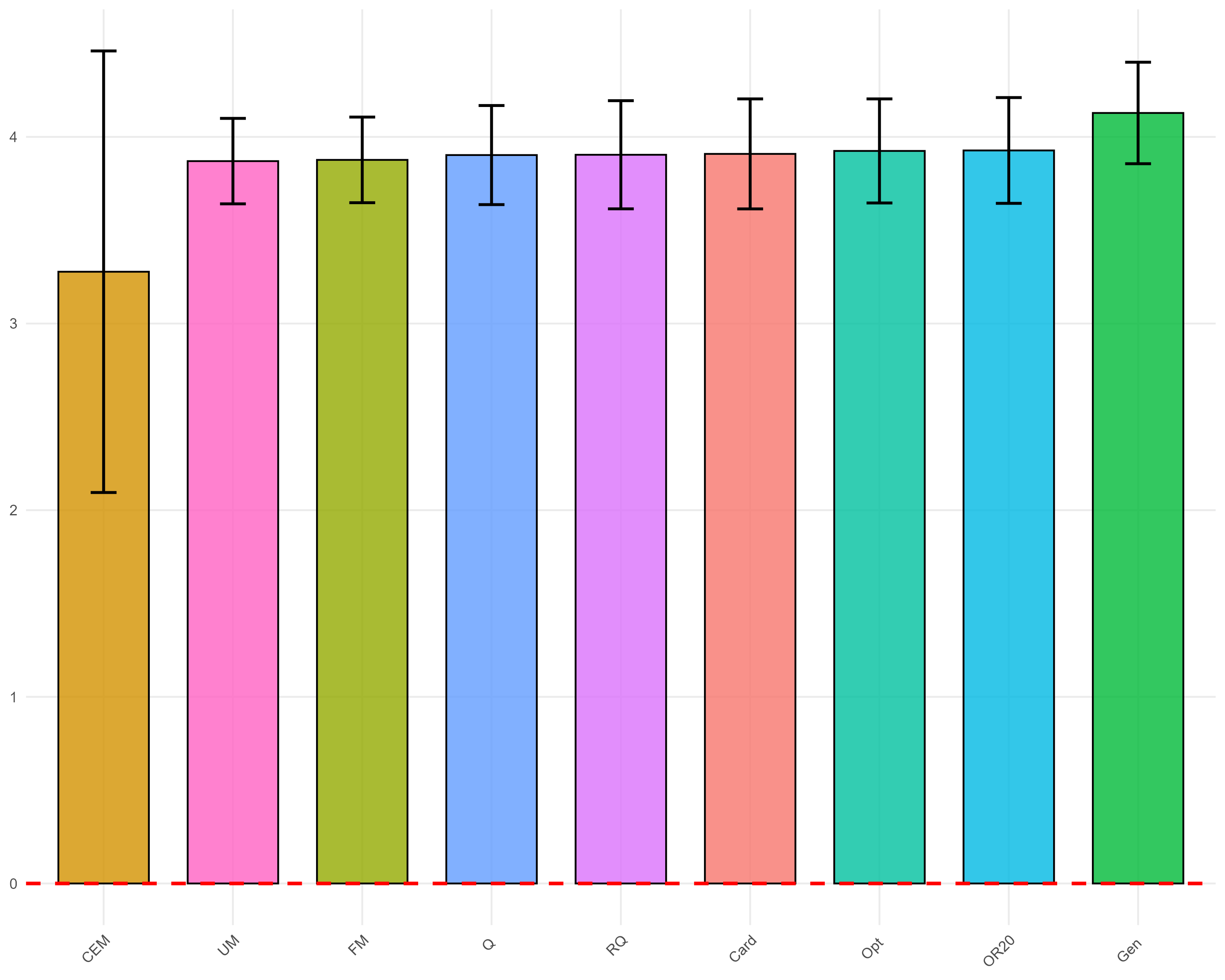} 
  \caption{ATE plot for all models, IHDP}
  \label{fig:ate_plot_ihdp}
\end{figure}

We now evaluate the distance and balance metrics reported in Table~\ref{tab:delta_all_dist_ihdp}. First, consider the Mahalanobis distance measures ($\text{Mean-Mahal}$, $\text{Median-Mahal}$, $\text{Max-Mahal}$, and $\text{SD-Mahal}$). For every value of $\delta$, the robust mode solution consistently yields lower Mahalanobis distance values than its strict mode counterpart. As $\delta$ increases, the first three distance metrics are monotonically non-increasing within each mode. In contrast, this monotonicity does not hold for $\text{SD-Mahal}$, which attains its minimum value under \texttt{OR15} ($0.6046$) and its maximum under \texttt{OS15} ($1.0865$). 

Across all four distance metrics, \texttt{OR20} emerges as the overall best \texttt{OSIP} solution: it achieves the minimum overall $\text{Mean-Mahal}$ ($3.0554$), minimum overall $\text{Median-Mahal}$ ($3.0934$), ties for the minimum $\text{Max-Mahal}$ ($4.4960$), and ranks second in $\text{SD-Mahal}$ ($0.6174$). Regarding covariate balance, only two solutions achieve full covariate balance ($\text{Nib} = 0;0$), namely \texttt{OR15} and \texttt{OR20}. Comparing the two, \texttt{OR20} outperforms \texttt{OR15} in $\text{Max-SMD}$ ($0.0897$ vs. $0.0968$) and $\text{Mean-SMD}$ ($0.0291$ vs. $0.0358$), while tying for the best $\text{Max-KS}$ value ($0.0896$). Consequently, we select \texttt{OR20} as our primary \texttt{OSIP} representative for comparison against external benchmark methods.


\begin{table}[htbp]
  \centering
  \small
  \caption{OSIP: Distances table for all values of $\delta$, IHDP}
  \label{tab:delta_all_dist_ihdp}
  \begin{filecontents*}{appendix_schemes/ihdp/ihdp_delta_all_dist.csv}
method,Mean-Mahal,Median-Mahal,Max-Mahal,SD-Mahal,Max-SMD,Mean-SMD,Max-KS,Nib
OS10,3.8932,3.9010,7.5043,1.0305,0.1094,0.0240,0.1045,1;0
OR10,3.2714,3.3242,4.7978,0.6381,0.1399,0.0412,0.1269,2;0
OS15,3.7253,3.6748,7.3125,1.0865,0.1369,0.0313,0.1194,2;0
OR15,3.1580,3.1456,4.4960,0.6046,0.0968,0.0358,0.0896,0;0
OS20,3.5805,3.5545,6.6828,0.9709,0.1279,0.0306,0.1119,3;0
OR20,3.0554,3.0934,4.4960,0.6174,0.0897,0.0291,0.0896,0;0
\end{filecontents*}
\pgfplotstabletypeset[
    col sep=comma,
    string type,
    every head row/.style={before row=\toprule, after row=\midrule},
    every last row/.style={after row=\bottomrule}
  ]{appendix_schemes/ihdp/ihdp_delta_all_dist.csv}
\end{table}

We now compare the \texttt{OR20} solution against all baseline methods on the \texttt{IHDP} dataset, consistent with our evaluation on the \texttt{Lindner} and \texttt{NSW-Mixtape} datasets. Because using the full set of variables causes severe cell sparsity, \texttt{CEM} is restricted to the continuous covariates $X_{1}\text{--}X_{6}$. Even under this relaxed specification, \texttt{CEM} retains only $9$ treated units, and $11$ control units. 

As detailed in Table~\ref{tab:delta0.20_all_ihdp_dist} \texttt{CEM} and \texttt{RQ} do not retain the full set of treated units.  This is a severe feature of the solution of \texttt{CEM} and for \texttt{RQ} it is a major disadvantage.  \texttt{CEM} suffers from this sample attrition due to high-dimensional cell sparsity. This extreme sample reduction substantially destabilizes its point estimate ($\hat{\tau}_{\text{ATE}} = 3.2776$) and drastically inflates the standard error to $0.6035$, underscoring the necessity of other frameworks like \texttt{OSIP} in moderate-to-high dimensional settings.  \texttt{RQ} also loses significant amount of data regarding treated units and it retains only $105$ treated units (and $105$ control units), and thus if possible we would like to avoid such bad behavior.  Focusing exclusively on the methods that retain the complete sample of $134$ treated units (\texttt{OR20}, \texttt{Opt}, \texttt{Card}, \texttt{FM}, \texttt{Gen}, \texttt{Q}), the estimated Average Treatment Effect remains remarkably consistent within the range of $3.8768$ to $4.1282$ where the unique maximizer is  of \texttt{Gen} and all other such methods have this value at most $3.9274$ (the $p$-value is always at most $0.0001$, the $t$-stat value is always larger than $25$). All methods evaluate to a Rosenbaum sensitivity threshold of $\Gamma > 5$, confirming that the estimated causal effect is robust against hidden confounding.  Thus, the data of this table suggest that we should not prefer \texttt{CEM} or \texttt{RQ}, and use other criteria to choose among the remaining solutions.

\begin{table}[htbp]
  \centering
  \small
  \caption{Significance table for all models, IHDP}
  \label{tab:delta0.20_all_ihdp_sig}
  \begin{filecontents*}{appendix_schemes/ihdp/consolidated_results_ihdp_delta_0.20_power_1_sig.csv}
method,treated,control,ate,se(ate),p-value,gamma-shift,ci-low,ci-high,t-stat
UM,134,556,3.8702,0.1168,$< 0.0001$,NA,3.6413,4.0992,33.1353
OR20,134,134,3.9274,0.1446,$< 0.0001$,$>5$,3.6440,4.2107,27.1604
Opt ,134,134,3.9248,0.1422,$< 0.0001$,$>5$,3.6460,4.2036,27.6006
Card,134,134,3.9089,0.1504,$< 0.0001$,$>5$,3.6142,4.2036,25.9900
CEM,9,11,3.2776,0.6035,0.0002,NA,2.0947,4.4605,5.4310
FM,134,556,3.8768,0.1170,$< 0.0001$,NA,3.6474,4.1062,33.1350
Gen ,134,134,4.1282,0.1387,$< 0.0001$,$>5$,3.8563,4.4002,29.7635
Q ,134,134,3.9027,0.1354,$< 0.0001$,$>5$,3.6372,4.1681,28.8235
RQ,105,105,3.9042,0.1478,$< 0.0001$,$>5$,3.6145,4.1940,26.4154
\end{filecontents*}
\pgfplotstabletypeset[
    col sep=comma,
    string type,
    every head row/.style={before row=\toprule, after row=\midrule},
    every last row/.style={after row=\bottomrule}
  ]{appendix_schemes/ihdp/consolidated_results_ihdp_delta_0.20_power_1_sig.csv}
\end{table}

As detailed in Table~\ref{tab:delta0.20_all_ihdp_dist}, \texttt{OR20} outperforms all baseline matching approaches in distance minimization and overall balance.  First consider the Mahalanobis distance metrics.  Here, \texttt{OR20} achieves a unique minimum value of the mean Mahalanobis distance ($3.0554$), a unique minimizer of the median Mahalanobis distance ($3.0934$), a unique minimizer of the maximum Mahalanobis distance ($4.496$) and a unique minimizer of the standard deviation of the Mahalanobis distance ($0.6174$).  Thus, with respect to distance Mahalanobis measures \texttt{OR20} is clearly our best candidate solution (among matching methods).   
 Furthermore, \texttt{OR20} achieves complete covariate balance ($N_{\text{ib}} = 0;0$) with a $\text{Max-SMD}$ of $0.0897$ and it is the unique method achieving this standard requirement. \texttt{RQ}  achieves only one imbalanced covariates and a $\text{Max-SMD}$ value of $0.1111$, but it requires discarding $29$ treated units ($N_{\text{treated}} = 105$). In contrast, \texttt{OR20} achieves full balance while retaining the complete treated cohort ($N_{\text{treated}} = 134$).  All other alternatives do not even get close to balance all covariate, and in particular their $\text{Max-SMD}$ value is at least $0.1293$. The $\text{Mean-SMD}$ of \texttt{OR20} is $0.0291$, that is second best following only the solution of \texttt{RQ}.  For the Max-KS criterion the solution \texttt{OR20} is our best solutions achieving a Max-KS value of $0.0896$.  Thus, \texttt{OR20} is both best solution with respect to Mahalanobis distance as well as regarding balancing the covariates.  We would like to understand better the reasons for this behavior of the methods with respect to the $\text{Max-SMD}$ metric, and thus turns our attention to the detailed balance table. 

\begin{table}[htbp]
  \centering
  \small
   \caption{Distances table for all models, IHDP}
  \label{tab:delta0.20_all_ihdp_dist}
  \begin{filecontents*}{appendix_schemes/ihdp/consolidated_results_ihdp_delta_0.20_power_1_dist.csv}
method,Mean-Mahal,Median-Mahal,Max-Mahal,SD-Mahal,Max-SMD,Mean-SMD,Max-KS,Nib
UM,NA,NA,NA,NA,0.228,0.0783,0.1626,9;3
OR20,3.0554,3.0934,4.496,0.6174,0.0897,0.0291,0.0896,0;0
Opt ,3.3812,3.4175,5.9534,0.8879,0.2286,0.0363,0.1269,3;1
Card,4.3654,4.3054,10.8213,1.3053,0.1476,0.0472,0.1194,3;0
CEM,NA,NA,NA,NA,0.3889,0.1043,0.3889,10;5
FM,NA,NA,NA,NA,0.2389,0.0515,0.1167,4;2
Gen ,4.052,4.1148,6.5856,1.2014,0.16,0.0299,0.097,2;0
Q ,3.9557,3.9138,7.5106,1.0375,0.1293,0.0315,0.1119,1;0
RQ,4.2355,4.152,6.8551,0.9307,0.1111,0.027,0.1143,1;0
\end{filecontents*}
\pgfplotstabletypeset[
    col sep=comma,
    string type,
    every head row/.style={before row=\toprule, after row=\midrule},
    every last row/.style={after row=\bottomrule}
  ]{appendix_schemes/ihdp/consolidated_results_ihdp_delta_0.20_power_1_dist.csv}
\end{table}

Table~\ref{tab:delta0.20_bal_ihdp} presents the covariate-level balance comparison between \texttt{OR20} and all benchmark models (excluding \texttt{CEM} due to its extreme sample attrition). The corresponding Love plot is illustrated in Figure~\ref{fig:love_plot_ihdp}. Among the continuous covariates ($X_{1}\text{--}X_{6}$), the unmatched baseline (\texttt{UM}) exhibits severe imbalance in five out of six variables ($X_1, X_2, X_4, X_5, X_6$). In addition, four binary covariates ($X_9, X_{14}, X_{17}, X_{25}$) are imbalanced. In total, $9$ out of $25$ covariates fail standard balance thresholds ($|\text{SMD}| > 0.10$) in the raw data, posing a challenging high-dimensional balancing task. 
The baseline matching methods struggle to overcome these baseline imbalances to varying degrees:
\begin{itemize}
    \item \texttt{Opt} successfully balances $X_1$ and $X_2$ as well as all binary covariates, but fails on $X_4$, $X_5$, and $X_6$.
    \item \texttt{Card} similarly eliminates binary covariate imbalances but fails to balance three continuous covariates ($X_1, X_5, X_6$).
    \item \texttt{FM} inherits much of \texttt{UM}'s continuous imbalance ($X_1, X_2, X_4$) and remains imbalanced on the binary covariate $X_{14}$.
    \item \texttt{Gen} fails on $X_3$ and $X_5$, notably producing \emph{worse} imbalance on both covariates than the original unmatched sample.
    \item \texttt{Q} fails on a single continuous covariate ($X_6$), which drives its maximum imbalance to $\text{Max-SMD} = 0.1293$.
    \item \texttt{RQ} likewise fails to balance $X_6$ ($\text{Max-SMD} = 0.1111$), demonstrating that even after discarding $29$ treated units, it cannot achieve full balance.
\end{itemize}
In sharp contrast, \texttt{OR20} manages to bring every single covariate within the $|\text{SMD}| \le 0.10$ threshold, achieving complete covariate balance across the full treated sample ($N_{\text{treated}} = 134$).

\begin{table}[htbp]
  \centering
  \small
   \caption{Balance table for all models, IHDP}
  \label{tab:delta0.20_bal_ihdp}
  \begin{filecontents*}{appendix_schemes/ihdp/love_plot_table_delta_0.20_ihdp.csv}
covariate,UM,OR20,Opt ,Card,FM,Gen ,Q ,RQ
X1,0.2018,0.0235,-0.0268,0.1369,0.2093,-0.0498,-0.0605,-0.058
X2,0.1694,0.0599,-0.09,0.0954,0.1806,-0.0982,-0.0763,-0.0626
X3,-0.0579,-0.0512,0.0408,0.0379,-0.0886,0.1422,0.0499,-0.0456
X4,-0.2114,0.0345,0.1144,-0.0802,-0.2389,-0.0651,0.0891,-0.0221
X5,-0.1499,-0.0897,-0.2286,-0.1476,-0.051,-0.16,-0.0991,0.0522
X6,0.228,-0.029,0.1449,0.1303,0.003,0.0766,0.1293,0.1111
X7,0.0005,0.0299,-0.0149,0.0075,0.0095,0.0149,0.0075,-0.0381
X8,-0.0058,0.0448,0,0,-0.0274,0,0.0149,-0.0095
X9,0.1626,-0.0299,0.0299,0.0597,0.0851,0,0,0.0286
X10,-0.0684,0,0,-0.0448,-0.0078,0.0224,0.0075,-0.0286
X11,-0.0217,0.0448,0.0299,0.0299,-0.0138,-0.0149,0.0075,0.0476
X12,-0.0045,-0.0448,-0.0224,-0.0448,-0.0078,-0.0149,-0.0299,-0.0286
X13,0.0227,0.0448,0.0522,-0.0075,-0.0181,0.0075,0.0373,0
X14,0.1047,-0.0149,-0.0522,0.0075,0.1101,0.0075,-0.0224,0.0476
X15,0.0105,0.0373,0.0224,-0.0373,-0.0417,0.0373,0.0448,-0.0095
X16,-0.0255,-0.0299,0,-0.0149,0.0008,0,-0.0224,-0.0095
X17,0.1002,0.0224,0,-0.0522,0.0375,-0.0075,0.0373,0.0381
X18,0.0087,0,0,0.0075,0,0,0,0
X19,0.0128,0.0075,0.0075,-0.0448,0.0041,0,0.0075,-0.019
X20,-0.0911,-0.0299,-0.0075,-0.0373,-0.0385,0,-0.0149,-0.0095
X21,0.0118,-0.0373,-0.0075,0.0448,0.0151,0.0149,0,0
X22,-0.0562,0,0,-0.0224,-0.0023,0,0,0
X23,-0.0174,0.0075,0,0,-0.0075,0,0.0075,0
X24,-0.0824,0,-0.0075,-0.0299,-0.0385,0,0,-0.0095
X25,0.1322,0.0149,0.0075,0.0597,0.0496,0.0149,-0.0224,0
\end{filecontents*}
\pgfplotstabletypeset[
    col sep=comma,
    string type,
    every head row/.style={before row=\toprule, after row=\midrule},
    every last row/.style={after row=\bottomrule}
  ]{appendix_schemes/ihdp/love_plot_table_delta_0.20_ihdp.csv}
\end{table}

\begin{figure}[htbp]
  \centering
  \small
  \includegraphics[width=0.85\textwidth]{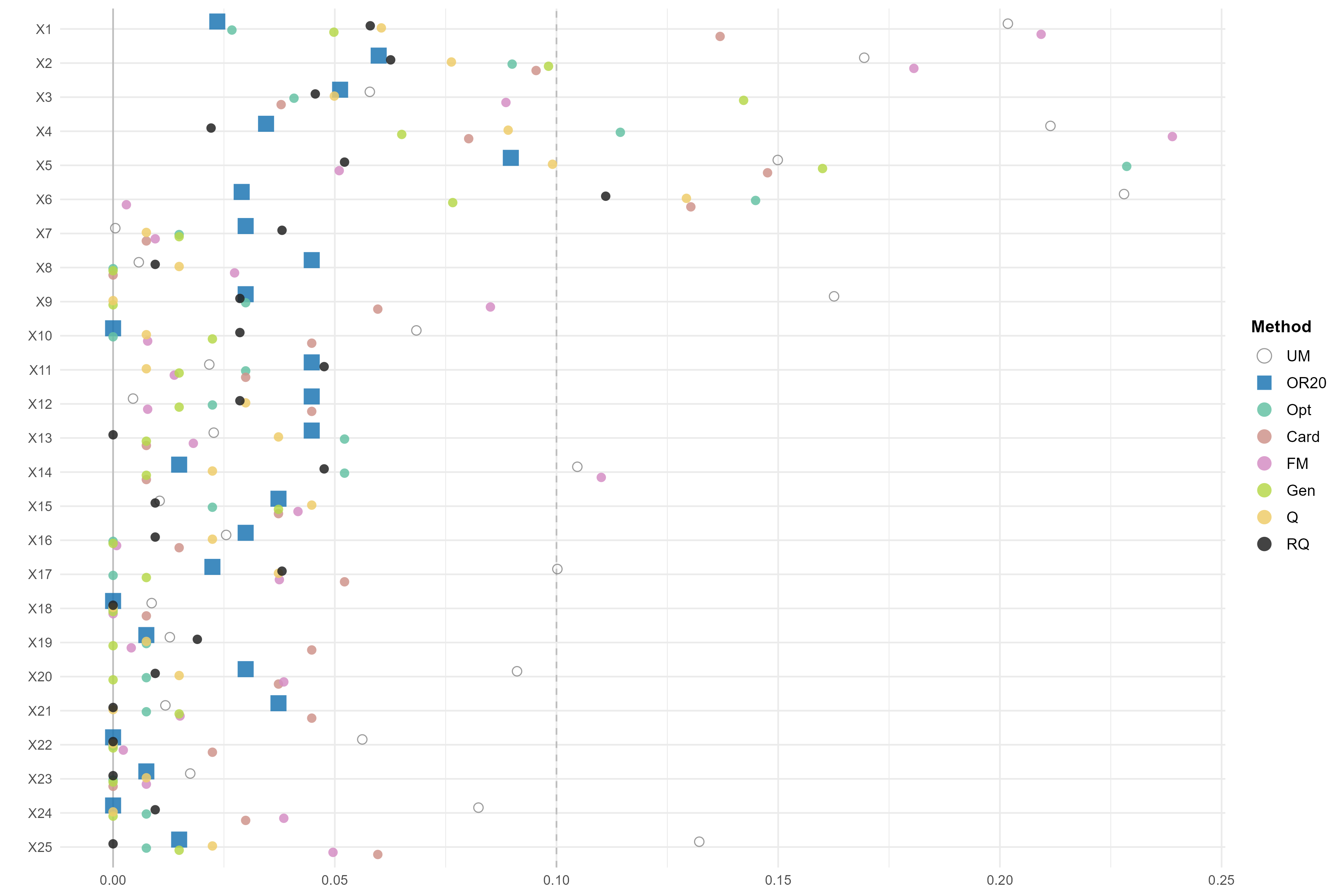} 
  \caption{Love plot for all models, IHDP}
  \label{fig:love_plot_ihdp}
\end{figure}

Thus, our suggested solution \texttt{OR20} performs better than all other alternatives, and we suggest using it for this \texttt{IHDP} dataset regardless of the exact priority rule suggested by a given expert.

\clearpage

\section{NHEFS dataset}
The National Health and Nutrition Examination Survey (NHEFS) dataset \citep{hernan2020causal, cdc_nhefs} and specifically its generalized form \texttt{nhefs\_complete} provided by the \texttt{causaldata} R package, originates from the Epidemiologic Follow-up Study initiated jointly by the National Center for Health Statistics (NCHS) and the National Institute on Aging. The dataset is widely utilized as a benchmark in causal inference literature (e.g., Hern\'an and Robins) to evaluate the effect of smoking cessation on weight gain. The results for the NHEFS dataset are again exceptionally significant because the underlying causal effect of smoking cessation on weight gain is pronounced and consistently estimated with low standard errors across the full sample.

The \texttt{nhefs\_complete} subset contains $1556$ observations (representing participants with complete record data) across $67$ variables. For the purpose of causal estimation and matching, the variables are structured as follows:

\begin{itemize}
    \item \textbf{Treatment (\texttt{qsmk})}: Binary indicator representing smoking cessation ($1 =$ quit smoking between baseline and follow-up, $0 =$ did not quit).
    \item \textbf{Outcome (\texttt{wt82\_71})}: Continuous measure representing weight change in kilograms between 1971 and 1982.
    \item \textbf{Key Baseline Confounders}:
    \begin{itemize}
        \item \textit{Demographic Factors}: Age (\texttt{age}), sex (\texttt{sex}), race (\texttt{race}), and education level (\texttt{school}).
        \item \textit{Smoking History}: Baseline smoking intensity (\texttt{smokeintensity}) and years of smoking (\texttt{smokeyrs}).
        \item \textit{Physical Activity and Health}: Recreational exercise level (\texttt{exercise}), daily physical activity level (\texttt{active}), and baseline body weight in 1971 (\texttt{wt71}).
    \end{itemize}
\end{itemize}

A comprehensive background, documentation, and original data files for the survey are available via the CDC NHEFS documentation portal \cite{cdc_nhefs}.

This dataset represents the largest sample evaluated in our benchmarking. The original sample contains $403$ treated and $1,163$ control units. After restricting units to the common support region $[0.0753, 0.6347]$, $402$ treated and $1,153$ control units remain, yielding a total sample size of $1,555$ units.

Table~\ref{tab:delta_all_est_conv_nhefs} provides key insights into the behavior of \texttt{OSIP} across parameter configurations for the \texttt{NHEFS} dataset. We evaluated \texttt{OSIP} for $\delta \in \{0.15, 0.175, 0.20\}$ in both strict and robust modes (letting \texttt{OS175} and \texttt{OR175} denote the outputs under $\delta=0.175$). At $\delta = 0.15$, the initial Step 1 dynamic programming solution reaches a local optimum in both modes, and none of the local search heuristics yield further improvements. This behavior largely persists at $\delta = 0.175$, with the exception of strict mode, where Simulated Annealing (SA) lowers the estimated cost from $519.9464$ to $513.9753$ after $14$ iterations. In contrast, at $\delta = 0.20$, every heuristic successfully improves upon the Step 1 baseline across both formulations. Notably, this configuration marks the first instance where Extended Local Search (ELS) achieves a strictly lower final cost than Local Search (LS), rather than merely accelerating convergence. Convergence for ELS was rapid across both modes at $\delta = 0.20$, requiring at most $4$ improving steps. Finally, across all values of $\delta$, the robust mode solution yields a strictly lower estimated cost than its strict counterpart, and increasing $\delta$ leads to a monotonic decrease in the final estimated cost.

\begin{table}[htbp]
  \centering
  \small
  \caption{OSIP: Estimated cost and conversion rates for all $\delta$ values, NHEFS}
  \label{tab:delta_all_est_conv_nhefs}
  \begin{filecontents*}{appendix_schemes/nhefs/cost_steps_nhefs.csv}
method,step 1,LS ,ELS ,SA
OS15,525.5083,NC,NC,NC
OR15,524.5724,NC,NC,NC
OS175,519.9464,NC,NC,513.9753; 14 
OR175,513.2459,NC,NC,NC
OS20,510.5231,508.0978; 3,505.7892; 4,509.7856; 2
OR20,500.554,497.7725; 2,495.61; 4,497.9888; 1
\end{filecontents*}
\pgfplotstabletypeset[
    col sep=comma,
    string type,
    every head row/.style={before row=\toprule, after row=\midrule},
    every last row/.style={after row=\bottomrule}
  ]{appendix_schemes/nhefs/cost_steps_nhefs.csv}
\end{table}

Table~\ref{tab:delta_all_sig_nhefs} reports the statistical significance metrics for the various \texttt{OSIP} configurations on the \texttt{NHEFS} dataset. Across all specifications, the results remain remarkably uniform: $p$-values are consistently $< 0.0001$, $t$-statistics strictly exceed $5.8071$, and Rosenbaum sensitivity bounds ($\Gamma$) remain tightly bounded between $1.5317$ and $1.5714$. Because every \texttt{OSIP} variant demonstrates overwhelming statistical significance,
we turn to distance and covariate balance metrics to discriminate among the candidate specifications and select the optimal output.


\begin{table}[htbp]
  \centering
  \small
  \caption{OSIP: Significance table for all values of $\delta$, NHEFS}
  \label{tab:delta_all_sig_nhefs}
  \begin{filecontents*}{appendix_schemes/nhefs/nhefs_delta_all_sig.csv}
method,ate,se(ate),p-value,gamma-shift,ci-low,ci-high,t-stat
OS15,3.4548,0.5440,$< 0.0001$,1.5606,2.3886,4.5210,6.3507
OR15,3.3295,0.5486,$< 0.0001$,1.5714,2.2543,4.4047,6.0691
OS175,3.2185,0.5323,$< 0.0001$,1.5453,2.1751,4.2619,6.0464
OR175,3.2918,0.5501,$< 0.0001$,1.5597,2.2136,4.3701,5.9840
OS20,3.2242,0.5294,$< 0.0001$,1.5375,2.1866,4.2619,6.0903
OR20,3.1678,0.5455,$< 0.0001$,1.5317,2.0985,4.2370,5.8071
\end{filecontents*}
\pgfplotstabletypeset[
    col sep=comma,
    string type,
    every head row/.style={before row=\toprule, after row=\midrule},
    every last row/.style={after row=\bottomrule}
  ]{appendix_schemes/nhefs/nhefs_delta_all_sig.csv}
\end{table}

Table~\ref{tab:delta_all_dist_nhefs} reports the distance and balance metrics across all evaluated \texttt{OSIP} configurations. Notably, two strict variants: \texttt{OS175} ($\text{Max-SMD} = 0.1160$) and \texttt{OS20} ($\text{Max-SMD} = 0.1019$) fail to achieve complete covariate balance ($\text{Nib} = 1;0$), exhibiting higher maximum standardized mean differences than the remaining four specifications. Consistent with the estimated cost trajectory, every robust variant strictly outperforms its corresponding strict counterpart across all four Mahalanobis distance metrics for each value of $\delta$. Finally, comparing the sole fully balanced strict variant (\texttt{OS15}) with the balanced robust variants (\texttt{OR175} and \texttt{OR20}) reveals that while \texttt{OS15} attains similar balance metrics ($\text{Max-SMD} = 0.0773$), it incurs substantially higher Mahalanobis matching distance (e.g., $\text{Mean-Mahal} = 1.4173$ for \texttt{OS15}, compared to $1.3083$ and $1.2749$ for \texttt{OR175} and \texttt{OR20}, respectively).


\begin{table}[htbp]
  \centering
  \small
  \caption{OSIP: Distances table for all values of $\delta$, NHEFS}
  \label{tab:delta_all_dist_nhefs}
  \begin{filecontents*}{appendix_schemes/nhefs/nhefs_delta_all_dist.csv}
method,Mean-Mahal,Median-Mahal,Max-Mahal,SD-Mahal,Max-SMD,Mean-SMD,Max-KS,Nib
OS15,1.4173,1.1818,4.3111,0.8737,0.0773,0.0337,0.0920,0;0
OR15,1.3479,1.1565,3.8380,0.7610,0.0891,0.0282,0.0821,0;0
OS175,1.3624,1.2018,4.1888,0.8107,0.1160,0.0404,0.1095,1;0
OR175,1.3083,1.1639,3.8547,0.7233,0.0736,0.0267,0.0896,0;0
OS20,1.3461,1.1717,4.8149,0.8114,0.1019,0.0384,0.1045,1;0
OR20,1.2749,1.1468,3.7351,0.6968,0.0912,0.0343,0.0970,0;0
\end{filecontents*}
\pgfplotstabletypeset[
    col sep=comma,
    string type,
    every head row/.style={before row=\toprule, after row=\midrule},
    every last row/.style={after row=\bottomrule}
  ]{appendix_schemes/nhefs/nhefs_delta_all_dist.csv}
\end{table}

Since the significance measures of all \texttt{OSIP} solutions are similar we will use the distance measures and the balancing measures as our tie-breaking rules. \texttt{OR175} achieves the most refined covariate balance, yielding the lowest maximum standardized mean difference ($\text{Max-SMD} = 0.0736$) and mean standardized difference ($\text{Mean-SMD} = 0.0267$). Alternatively, \texttt{OR20} minimizes global distance metrics, obtaining the lowest overall mean Mahalanobis distance ($1.2749$) and maximum Mahalanobis distance ($3.7351$).   This pair of solutions are analyzed further below in comparison to all other methods.

Table~\ref{tab:delta0.20_all_nhefs_sig} summarizes the statistical significance and sensitivity metrics across all benchmark models for the \texttt{NHEFS} dataset. Within the common support region, $402$ treated units are available, and most methods retain this full sample ($N_{\text{treated}} = 402$). However, \texttt{CEM} suffers extreme sample loss, retaining only $32$ treated units, while \texttt{RQ} retains $308$ treated units. This severe sample reduction in \texttt{CEM} inflates its standard error to $1.6884$, yielding a $p$-value of $0.0729$---rendering it the only model that fails to achieve statistical significance at standard levels ($p$-value $< 0.05$). All other models maintain $p$-value $< 0.0001$.

Focusing on the methods that preserve the entire treated sample, \texttt{Opt} exhibits noticeably weaker statistical robustness than its peers:
\begin{itemize}
    \item Its $t$-statistic is $5.4797$, whereas all other full-sample matched specifications achieve a $t$-statistic of at least $5.8071$ (\texttt{OR20}).
    \item Its Rosenbaum sensitivity parameter ($\Gamma$) drops to $1.3581$, whereas every other full-sample matched model demonstrates substantial robustness with $\Gamma \ge 1.5317$.
\end{itemize}
Because the remaining full-sample methods (\texttt{OR175}, \texttt{OR20}, \texttt{Card}, \texttt{Gen}, and \texttt{Q}) exhibit similarly strong significance metrics, secondary balance and distance metrics must be utilized as tie-breakers to evaluate their final performance.


\begin{table}[htbp]
  \centering
  \small
  \caption{Significance table for all models, NHEFS}
  \label{tab:delta0.20_all_nhefs_sig}
  \begin{filecontents*}{appendix_schemes/nhefs/consolidated_results_nhefs_delta_all_power_1_sig.csv}
method,treated,control,ate,se(ate),p-value,gamma-shift,ci-low,ci-high,t-stat
UM,402,1153,3.4201,0.4358,$< 0.0001$,NA,2.5659,4.2742,7.8479
OR175,402,402,3.2918,0.5501,$< 0.0001$,1.5597,2.2136,4.3701,5.984
OR20,402,402,3.1678,0.5455,$< 0.0001$,1.5317,2.0985,4.237,5.8071
Opt ,402,402,3.0111,0.5495,$< 0.0001$,1.3581,1.934,4.0882,5.4797
Card,402,402,3.2707,0.5486,$< 0.0001$,1.5631,2.1956,4.3459,5.9619
CEM,32,33,3.0838,1.6884,0.0729,NA,-0.2255,6.3932,1.8265
FM,402,1153,3.2917,0.4456,$< 0.0001$,NA,2.4183,4.1651,7.3871
Gen ,402,402,3.6148,0.5541,$< 0.0001$,1.5782,2.5288,4.7008,6.5237
Q ,402,402,3.2149,0.5339,$< 0.0001$,1.5448,2.1685,4.2612,6.0215
RQ,308,308,3.1365,0.5773,$< 0.0001$,1.5717,2.0049,4.2681,5.4331
\end{filecontents*}
\pgfplotstabletypeset[
    col sep=comma,
    string type,
    every head row/.style={before row=\toprule, after row=\midrule},
    every last row/.style={after row=\bottomrule}
  ]{appendix_schemes/nhefs/consolidated_results_nhefs_delta_all_power_1_sig.csv}
\end{table}

\begin{figure}[htbp]
  \centering
  \small
  \includegraphics[width=0.85\textwidth]{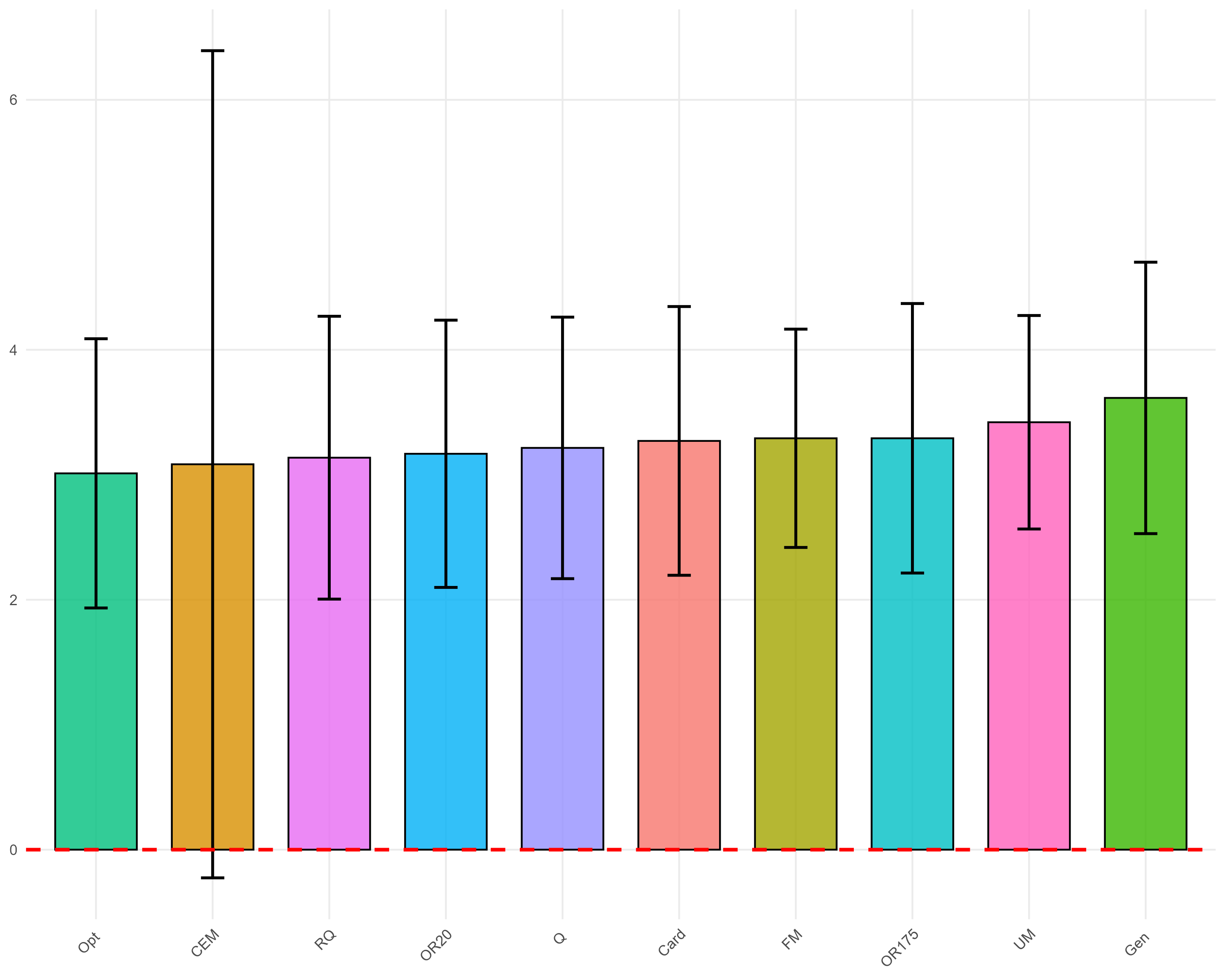} 
  \caption{ATE plot for all models, NHEFS}
  \label{fig:ate_plot_nhefs}
\end{figure}

Turning to the Mahalanobis distance and covariate balance metrics in Table~\ref{tab:delta0.20_all_nhefs_dist}, we observe substantial performance variation across the candidate models:

\begin{itemize}
    \item \textbf{Covariate Imbalance in Baselines:} Both \texttt{OSIP} variants (\texttt{OR175} and \texttt{OR20}) achieve complete covariate balance ($\text{Nib} = 0;0$) with $\text{Max-SMD} \le 0.0912$. In contrast, several baselines fail to achieve balance: \texttt{Opt} leaves $2$ covariates imbalanced ($\text{Max-SMD} = 0.1353$), \texttt{Card} leaves $3$ imbalanced ($\text{Max-SMD} = 0.1490$), and \texttt{FM} leaves $1$ imbalanced ($\text{Max-SMD} = 0.1234$). Even with significant sample loss, \texttt{CEM} fails to achieve balance ($\text{Nib} = 1;0$, $\text{Max-SMD} = 0.1708$). While \texttt{RQ} achieves tight balance ($\text{Max-SMD} = 0.0576$), its severe sample loss precludes it from being a top-tier choice.
    
    \item \textbf{Full-Sample Balanced Models:} Among the four methods that retain all $402$ treated units and achieve full covariate balance (\texttt{OR175}, \texttt{OR20}, \texttt{Gen}, and \texttt{Q}), \texttt{OR175} delivers the best balance profile, attaining the lowest $\text{Max-KS}$ ($0.0896$) and the second-lowest $\text{Mean-SMD}$ ($0.0267$, behind \texttt{Gen}'s $0.0224$).
    
    \item \textbf{Mahalanobis Distance Efficiency:} While \texttt{Opt} yields the lowest overall Mahalanobis distances (e.g., $\text{Mean-Mahal} = 1.2186$), its inability to balance all covariates disqualifies it. Among the valid, fully balanced models, \texttt{OR20} achieves the smallest matching distance ($\text{Mean-Mahal} = 1.2749$, $\text{Max-Mahal} = 3.7351$), closely followed by \texttt{OR175} ($\text{Mean-Mahal} = 1.3083$).
\end{itemize}

Consequently, \texttt{OR175} stands out as the primary choice when prioritizing balance metrics, while \texttt{OR20} represents the optimal trade-off when minimizing overall matching distance.

\begin{table}[htbp]
  \centering
  \small
   \caption{Distances table for all models, NHEFS}
  \label{tab:delta0.20_all_nhefs_dist}
  \begin{filecontents*}{appendix_schemes/nhefs/consolidated_results_nhefs_delta_all_power_1_dist.csv}
method,Mean-Mahal,Median-Mahal,Max-Mahal,SD-Mahal,Max-SMD,Mean-SMD,Max-KS,Nib
UM,NA,NA,NA,NA,0.2696,0.084,0.14,4;2
OR175,1.3083,1.1639,3.8547,0.7233,0.0736,0.0267,0.0896,0;0
OR20,1.2749,1.1468,3.7351,0.6968,0.0912,0.0343,0.097,0;0
Opt ,1.2186,1.1,3.9272,0.6609,0.1353,0.0356,0.1393,2;0
Card,1.7554,1.5482,5.4054,0.9504,0.149,0.0558,0.0896,3;0
CEM,NA,NA,NA,NA,0.1708,0.0293,0.2500,1;0
FM,NA,NA,NA,NA,0.1234,0.0226,0.0843,1;0
Gen ,1.5317,1.2055,4.9861,0.982,0.0809,0.0224,0.1045,0;0
Q ,1.4432,1.2686,4.6315,0.8322,0.0938,0.0365,0.107,0;0
RQ,1.4752,1.3397,3.9726,0.7974,0.0576,0.0125,0.0714,0;0
\end{filecontents*}
\pgfplotstabletypeset[
    col sep=comma,
    string type,
    every head row/.style={before row=\toprule, after row=\midrule},
    every last row/.style={after row=\bottomrule}
  ]{appendix_schemes/nhefs/consolidated_results_nhefs_delta_all_power_1_dist.csv}
\end{table}

\clearpage

Finally, we examine the individual covariate Standardized Mean Difference (SMD) values provided in Table~\ref{tab:bal_nhefs_all} and visualized in the Love plot in Figure~\ref{fig:love_plot_nhefs}. The unmatched baseline (\texttt{UM}) exhibits severe imbalance across four key covariates: \texttt{age} ($\text{SMD}=0.2696$), \texttt{smokeintensity} ($\text{SMD}=-0.2010$), \texttt{smokeyrs} ($\text{SMD}=0.1582$), and \texttt{wt71} ($\text{SMD}=0.1264$).

The candidate models demonstrate varying capabilities in correcting these baseline imbalances:
\begin{itemize}
    \item Both \texttt{OSIP} variants successfully bring all four problematic covariates well within the standard threshold without degrading the balance of remaining variables. \texttt{OR175} attains its maximum imbalance at \texttt{smokeintensity} ($\text{Max-SMD} = 0.0736$), whereas \texttt{OR20} reaches its maximum at \texttt{age} ($\text{Max-SMD} = 0.0912$).
    \item \texttt{Opt} achieves exact zero imbalance on \texttt{sex}, \texttt{race}, and \texttt{active-2}, yet fails to correct the initial imbalances of \texttt{age} ($0.1353$) and \texttt{smokeintensity} ($-0.1282$).
    \item \texttt{Card} similarly fails on \texttt{age} ($0.1459$) and \texttt{smokeintensity} ($-0.1490$), while exacerbating the imbalance of \texttt{school} from $0.0535$ in \texttt{UM} up to $0.1405$.
    \item Both \texttt{CEM} and \texttt{FM} achieve tight balance across most variables, but neither resolves the severe imbalance in \texttt{smokeintensity} ($-0.1708$ and $-0.1234$, respectively).
    \item \texttt{Gen} eliminates all initial imbalances while attaining exact zero balance on three indicator covariates (\texttt{active-0}, \texttt{active-1}, and \texttt{active-2}), reaching its peak imbalance at \texttt{age} ($\text{Max-SMD} = 0.0809$).
    \item \texttt{Q} achieves full balance across all variables, securing exact zero balance on \texttt{active-2} and attaining its peak imbalance at \texttt{wt71} ($\text{Max-SMD} = 0.0938$).
    \item \texttt{RQ} balances all covariates effectively, reaching a low peak imbalance of $\text{Max-SMD} = 0.0576$ at \texttt{smokeyrs}.
\end{itemize}

\begin{table}[htbp]
  \centering
  \small
   \caption{Balance table for all models, NHEFS}
  \label{tab:bal_nhefs_all}
  \begin{filecontents*}{appendix_schemes/nhefs/love_plot_table_all_nhefs.csv}
covariate,UM,OR175,OR20,Opt ,Card,CEM,FM,Gen ,Q ,RQ
sex,-0.0789,-0.0075,-0.005,0,0.0522,0,-0.0169,-0.0448,-0.0249,-0.0097
race,-0.0536,0,-0.005,0,-0.0348,0,-0.0212,-0.0075,0.0075,0
age,0.2696,0.0558,0.0912,0.1353,0.1459,-0.0051,0.0548,0.0809,0.0526,0.0559
school,0.0535,-0.0191,-0.0278,-0.031,0.1405,0.0619,0.0071,-0.0128,-0.076,-0.0034
smokeintensity,-0.201,-0.0736,-0.0793,-0.1282,-0.149,-0.1708,-0.1234,-0.0713,-0.0779,0.0033
smokeyrs,0.1582,0.0328,0.0671,0.0738,0.0321,-0.022,0.0321,-0.0154,0.062,0.0576
exercise-0,-0.0436,0.0075,0.01,0.0025,0,0,-0.0058,0.0025,0.0149,0
exercise-1,0.0198,-0.0224,-0.0274,-0.005,0.0721,0,-0.009,0.0025,-0.0249,-0.0032
exercise-2,0.0239,0.0149,0.0174,0.0025,-0.0721,0,0.0148,-0.005,0.01,0.0032
active-0,-0.0316,0.0075,0.0025,-0.0075,-0.0025,0,-0.0001,0,0.0149,0
active-1,0.0123,-0.0274,-0.0174,0.0075,-0.0075,0,0.0028,0,-0.0149,-0.0032
active-2,0.0193,0.0199,0.0149,0,0.01,0,-0.0027,0,0,0.0032
wt71,0.1264,0.0591,0.081,0.0699,-0.0069,0.0923,-0.0033,0.049,0.0938,0.0196
\end{filecontents*}
\pgfplotstabletypeset[
    col sep=comma,
    string type,
    every head row/.style={before row=\toprule, after row=\midrule},
    every last row/.style={after row=\bottomrule}
  ]{appendix_schemes/nhefs/love_plot_table_all_nhefs.csv}
\end{table}

\begin{figure}[htbp]
  \centering
  \small
  \includegraphics[width=0.85\textwidth]{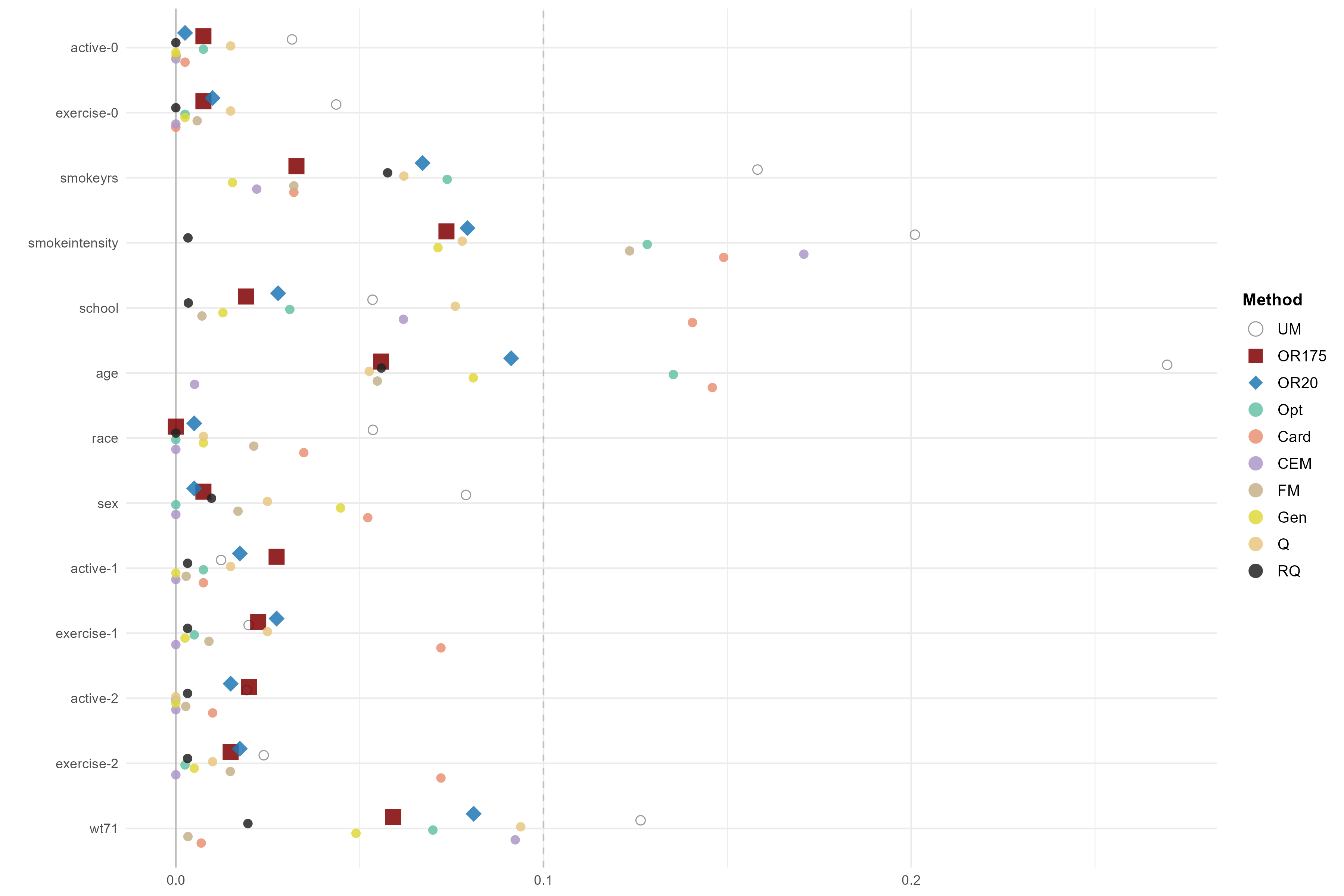} 
  \caption{Love plot for all models, NHEFS}
  \label{fig:love_plot_nhefs}
\end{figure}

In summary, as established through our benchmarking analysis, the two robust \texttt{OSIP} variants, \texttt{OR175} and \texttt{OR20}, emerge as the overall preferred solutions for the \texttt{NHEFS} dataset. Practitioners should select \texttt{OR20} when prioritizing minimal Mahalanobis matching distance, or choose \texttt{OR175} when prioritizing strict covariate balance and Kolmogorov-Smirnov metrics.


\end{document}